\documentclass[a4paper,reqno]{amsart}
\usepackage[utf8]{inputenc}

\usepackage[english]{babel} 
\usepackage{csquotes}

\usepackage{longtable}

\usepackage[colorlinks=true,linkcolor=black,anchorcolor=black,citecolor=black,filecolor=black,menucolor=black,runcolor=black,urlcolor=black]{hyperref} 
\usepackage{graphicx,color,comment}  
\usepackage[dvipsnames]{xcolor} 
\usepackage{enumitem} 
\usepackage{rotating} 
\usepackage[absolute]{textpos}

\usepackage{mdframed}
\usepackage{xfrac}\usepackage{nicefrac}
\usepackage{comment}

\usepackage{cancel}

\usepackage{amsmath,amssymb,amsfonts,amsthm} 
\usepackage{mathtools,stmaryrd}  
\definecolor{light-gray}{gray}{0.95}
\usepackage[color=light-gray]{todonotes}
\usepackage{tikz-cd}  
\usepackage[all,cmtip]{xy} 
\usepackage[bbgreekl]{mathbbol} 
\usepackage{mathrsfs,slashed}

\usepackage{mathrsfs}
\usepackage[mathcal,mathscr]{euscript} 

\usepackage[backend=biber, giveninits=true, style=alphabetic, sorting=nyt,isbn=false, maxalphanames=5,url=false,doi=false,eprint=false, maxbibnames=99]{biblatex}
\usepackage{thmtools}

\newcommand{\thmspacing}{.65em}

\declaretheoremstyle[
  spaceabove=\thmspacing, spacebelow=\thmspacing,
  headfont=\bfseries,
  notefont=\bfseries, notebraces={(}{)},
  bodyfont=\normalfont,
  postheadspace=1em,
  qed=$\Diamond$
]{pluto}
    \declaretheorem[style=pluto,name=Definition, numberwithin=section]{definition}
    
    \declaretheorem[style=pluto,name=Assumption,
    ]{assumption}

\declaretheoremstyle[
  spaceabove=\thmspacing, spacebelow=\thmspacing,
  headfont=\itshape,
  notefont=\itshape, notebraces={\emph{(}}{\emph{)}},
  bodyfont=\normalfont,
  postheadspace=1em,
  qed=$\Diamond$
]{pluto2}
    \declaretheorem[style=pluto2,name=Remark,    sibling=definition]{remark}

\declaretheoremstyle[
  spaceabove=\thmspacing, spacebelow=\thmspacing,
  headfont=\bfseries,
  notefont=\bfseries, notebraces={(}{)},
  bodyfont=\itshape,
  postheadspace=1em,
]{pluto3}
    \declaretheorem[style=pluto3,name=Theorem,    sibling=definition]{theorem}
    \declaretheorem[style=pluto3,name=Lemma,    sibling=definition]{lemma}
    \declaretheorem[style=pluto3,name=Corollary,    sibling=definition]{corollary}
    \declaretheorem[style=pluto3,name=Proposition,    sibling=definition]{proposition}
    \declaretheorem[style=pluto3,name=Calculation,    sibling=definition]{calculation}
    \declaretheorem[style=pluto3,name=Notation,    sibling=definition]{notation}

\newcommand{\pp}{\partial}

\renewcommand{\tilde}{\widetilde}
\newcommand{\rel}{\mathrm{rel}}

\newcommand{\cV}{\check{V}}

\newcommand{\Pexp}{\mathrm{Pexp}}

\newcommand{\E}{E}

\newcommand{\dfU}{{V}}
\newcommand{\cU}{{\check{V}}}

\newcommand{\dfuo}{{\upsilon_\circ}}
\newcommand{\dfUo}{{U_\circ}}

\DeclareMathAlphabet{\mathbbmsl}{U}{bbm}{m}{sl}

\newcommand{\G}{\mathscr{G}}
\newcommand{\Go}{\G_\circ}

\newcommand{\tGo}{\G_{\rel}}

\newcommand{\AS}{{\X_{\mathsf{AS}}}}
\newcommand{\XAS}{{\X_{\mathsf{eAS}}}}
\newcommand{\Xas}{{{\X}^{\mathsf{lin}}_{\mathsf{eAS}}}}

\newcommand{\omAS}{{\varpi_{\mathsf{AS}}}}
\newcommand{\omeAS}{{\varpi_{\mathsf{eAS}}}}
\newcommand{\omeas}{{{\varpi}^{\mathsf{lin}}_{\mathsf{eAS}}}}

\newcommand{\Gred}{\underline{\G}}

\renewcommand{\L}{\mathbb{L}}
\newcommand{\bi}{\mathbb{i}}
\newcommand{\Lie}{\mathrm{Lie}}
\newcommand{\bom}{\boldsymbol{\omega}}
\newcommand{\bHo}{\boldsymbol{H}_\circ}
\newcommand{\bH}{\boldsymbol{H}}
\newcommand{\bh}{\boldsymbol{h}}
\newcommand{\balpha}{{\boldsymbol{\alpha}}}

\newcommand{\bE}{{\boldsymbol{E}}}

\newcommand{\tr}{\mathrm{tr}}
\newcommand{\btr}{\mathbf{tr}}
\newcommand{\scri}{\mathscr{I}}

\newcommand{\into}{\hookrightarrow}
\newcommand{\onto}{\twoheadrightarrow}

\renewcommand{\pp}{\partial}
\newcommand{\wt}[1]{\widetilde{#1}}

\renewcommand{\Im}{\mathrm{Im}}
\renewcommand{\ker}{\mathrm{Ker}}

\newcommand{\rad}{{\mathrm{rad}}}
\newcommand{\nYM}{\mathsf{nYM}}

\newcommand{\ASa}{{\mathsf{a}}}
\newcommand{\ASe}{{\mathsf{e}}}

\newcommand{\vol}{{\boldsymbol{vol}}}
\renewcommand{\i}{{\mathrm{in}}}
\newcommand{\f}{{\mathrm{fin}}}

\newcommand{\uuPi}{\underline{\underline{\Pi}}{}}

\newcommand{\hcalA}{\widehat{\mathcal{A}}}

\newcommand{\F}{\mathfrak{F}}
\newcommand{\fG}{\mathfrak{G}}
\newcommand{\fGp}{{\fg^{\pp\Sigma}}}

\newcommand{\fg}{\mathfrak{g}}
\newcommand{\fGo}{\mathfrak{G}_\circ}

\newcommand{\fGred}{\underline{\mathfrak{G}}}

\newcommand{\Ho}{{H}_\circ}

\newcommand{\h}{{h}}

\newcommand{\X}{\mathcal{P}}
\newcommand{\C}{\mathcal{C}}

\newcommand{\Ann}{\mathrm{Ann}}
\newcommand{\uCo}{\underline{\C}}

\newcommand{\uh}{\underline{\h}{}}
\newcommand{\uhAS}{\underline{\h}_{\mathsf{eAS}}}

\newcommand{\uomegao}{\underline{\omega}}

\newcommand{\uS}{\underline{\mathcal{S}}{}}

\newcommand{\uuS}{\underline{\underline{\mathcal{S}}}{}}
\newcommand{\uuC}{\underline{\underline{\C}}{}}

\newcommand{\cD}{\mathcal{D}}
\newcommand{\cL}{\mathcal{L}}

\newcommand{\av}{\mathrm{avg}}
\newcommand{\dif}{\mathrm{diff}}

\newcommand{\so}{\mathrm{soft}}

\newcommand{\uuomegao}{\underline{\underline{\omega}}{}}
\newcommand{\oloc}{\Omega_{\loc}}

\newcommand{\loc}{\mathrm{loc}}

\newcommand{\tbox}[2]{\mbox{\parbox{#1}{\center\sffamily\footnotesize #2}}}

\newcommand{\bd}{\mathbb{d}}

\newcommand{\Ad}{\mathrm{Ad}}
\newcommand{\ad}{\mathrm{ad}}

\title[Null Hamiltonian Yang--Mills theory]{Null Hamiltonian Yang--Mills theory: \\ Soft symmetries and memory as superselection}
\author{A. Riello}
\address{Perimeter Institute for Theoretical Physics, 31 Caroline St. N., Waterloo, ON N2L 2Y5, Canada}
\address{Department of Applied Mathematics, University of Waterloo, Waterloo, Ontario, Canada}
\email{ariello@perimeterinstitute.ca}
\author{M. Schiavina}
\address{Department of Mathematics, University of Pavia, Via Ferrata 5, 27100 Pavia, Italy}
\address{INFN Sezione di Pavia, via Bassi 6, 27100 Pavia, Italy}
\email{michele.schiavina@unipv.it}

\date{}

\begin{document}

\maketitle

\begin{abstract}
    Soft symmetries for Yang--Mills theory are shown to correspond to the residual Hamiltonian action of the gauge group on the Ashtekar--Streubel phase space, which is the result of a partial symplectic reduction. The associated momentum map is the electromagnetic memory in the Abelian theory, or a nonlinear, gauge-equivariant, generalisation thereof in the non-Abelian case. 
    This result follows from an application of Hamiltonian reduction by stages, enabled by the existence of a natural normal subgroup of the gauge group on a null codimension-1 submanifold with boundaries.  The first stage is coisotropic reduction of the Gauss constraint, and it yields a symplectic extension of the Ashtekar--Streubel phase space (up to a covering). Hamiltonian reduction of the residual gauge action leads to the fully-reduced phase space of the theory. This is a Poisson manifold, whose symplectic leaves, called superselection sectors, are labelled by the (gauge classes of the generalised) electric flux across the boundary. 
    In this framework, the Ashtekar--Streubel phase space arises as an intermediate reduction stage that enforces the superselection of the electric flux at only one of the two boundary components. These results provide a natural, purely Hamiltonian, explanation of the existence of soft symmetries as a byproduct of partial symplectic reduction, as well as a motivation for the expected decomposition of the quantum Hilbert space of states into irreducible representations labelled by the Casimirs of the Poisson structure on the reduced phase space.
\end{abstract}

\tableofcontents

\section{Introduction}

\subsection{Overview} Asymptotic quantisation of Maxwell and gravitational theories in asymptotically flat spacetimes is an idea that goes back to \cite{AshtekarAsympt,ashtekar1987asymptotic}, motivated by questions on the non-perturbative quantisation of theories with long-range interactions, and it requires studying the observables of the theory at \emph{null} infinity (see \cite{AshtekarCmapigliaLaddha2018} for a recent account). The long-range nature of the interaction is tied to Gauss's law, whose implementation in a quantum setting has been shown to require a decomposition of quantum observables into ``superselection sectors'' (e.g. \cite{Buchholz86}; see item \ref{item:quantisation} in Section \ref{sec:physicalinterpretation} for more references, and different approaches).

More recently, a third observation enriched the discussion on this topic: that there should be an underlying symmetry descriptor linking asymptotic quantisation to soft scattering theorems. Specifically, this is the observation that certain soft theorems, most notably \cite{Weinberg65}, are to be interpreted as Ward identities for certain new symmetries that the theory enjoys at null infinity \cite{HeMitra2014,KapecPateStrominger}. A host of literature was subsequently produced to understand the nature of these soft (or asymptotic or, sometimes, ``large'') symmetries in the context of Maxwell, non-Abelian Yang--Mills, and gravity theories (for an overview, see \cite{StromingerLectureNotes} and references therein). However, a full description of this phenomenon from a Hamiltonian point of view, which could link soft behaviour to superselection, has been lacking. This is where our work finds its main application.

We analyse the \emph{Hamiltonian} assignment in Yang--Mills theory (YM) to a null codimension-1 submanifold with boundary, and construct the reduced phase space of the theory by means of a procedure known as Hamiltonian reduction by stages \cite{marsden2007stages}. 

We show that the (extended\footnote{The extension we find is formally analogous to e.g.\ that in \cite{CampigliaPeraza21}, but a one-to-one mapping is far from obvious.}) Ashtekar--Streubel (eAS) phase space is the result of \emph{partial} Hamiltonian reduction by appropriate subgroups of the gauge group that naturally exist due to the presence of boundaries. Since the reductions leading to these spaces are only partial, we characterise the residual (gauge) group action on the AS phase space (and its extensions), and show it recovers the ``soft symmetries'' of \cite{HeMitra2014, KapecPateStrominger, StromingerLectureNotes} (however, cf. \cite{HerdegenRevisited}), to which we thus give a purely Hamiltonian explanation.

Additionally, our explicit description of the reduced phase space goes through the construction of classical analogues of what in the literature are known as (quantum) ``superselection sectors''. This interplay between reduction and superselection is at the heart of our explanation of the relationship between asymptotic quantisation, superselection sectors and soft/asymptotic symmetries.   

More specifically, our (classical) superselection sectors arise from the fact that the charge (momentum map) which generates gauge transformations \emph{fails} to vanish on-shell in the presence of boundaries---which, in this  case, one can picture as the past and future ``celestial spheres''. 

Instead, said charge is given by the flux\footnote{Note: in this work, the wording ``electric flux'' is a shorthand for ``normal component of the electric field at the boundary".} of the (generalised) electric field through the boundary, whose gauge classes then label the superselection sectors. In the \emph{partially-reduced} extended Ashtekar--Streubel phase space, the charge generates residual gauge transformations corresponding to the above-mentioned soft symmetries.

In the case of Abelian Yang--Mills,  one can alternatively choose one of the superselection labels to be the \emph{electromagnetic memory} of\footnote{In \cite{herdegen2023velocity}, Herdegen advances a criticism of the interpretation of memory proposed by \cite{BieriGarfinkle} . This is irrelevant for our purposes, but worth noting.} \cite{Staruszkiewicz:1981,Herdegen95,GarfinkleHollandsetal,Pasterskimemory}, which is then tied both to symmetry \emph{and} superselection.
Our formalism naturally extends to the subtler non-Abelian case, where we define a non-Abelian generalisation of memory as a superselection label. However, our proposal differs from the notion of ``color memory'' introduced in \cite{PateRaclariuStrominger} (see Definition \ref{def:non-Abelianmemory} and Remark \ref{rmk:non-Abelianmemory}), and other terms and conditions apply (see Theorem \ref{thm:memorySSS} and Section \ref{sec:semisimplememory} for precise statements). 

\medskip
Our results are obtained through a careful and rigorous application of symplectic reduction by stages \cite{MarsdenWeinstein, Meyer, AbrahamMarsden,marsden2007stages}---adapted to the Fr\'echet setting as in \cite{DiezHuebschmann-YMred,DiezPhD}---and the construction is completely and manifestly gauge-invariant. (We follow \cite{RielloSchiavina} for the implementation.)

Moreover, each stage of the reduction procedure has its own physical meaning:
\begin{enumerate}
\item The first stage implements  the (Gauss) constraint, and produces  the extended Ashtekar--Streubel phase space (Theorem \ref{thm:nonAb-constr-red}).
\item The second stage reduces the residual ``boundary'' gauge symmetries (when present) and leads to the superselection sectors (Theorem \ref{thm:modelHamreduction}). 
\item The Ashtekar--Streubel phase space is recovered as the result of an \emph{intermediate} reduction, interpreted as superselection of the electric flux through only one boundary component (Proposition \ref{prop:MaxPartialSSS} and Theorem \ref{thm:YMPartialSSS}).
\end{enumerate}

Many of these statements can be expressed effectively by means of a new symplectic basis for the Ashtekar--Streubel phase space, which we introduce in Proposition \ref{prop:modedecomp}. This basis generalises the naive Fourier basis to correctly take into account zero-modes, which play a crucial role in the soft theorems and the memory effect.

Although our results hold in any spacetime dimension at finite distance, in dimension 4, where YM theory is classically conformal, they provide relevant information on the asymptotic data as well, see Section \ref{sec:alternatelabelmemory}. Then, our approach recovers results of \cite{HeMitra2014,KapecPateStrominger,Pasterskimemory} and places them within a precise Hamiltonian framework (see Section \ref{sec:alternatelabelmemory}).

\subsection{A note on boundary gauge transformations}
A conceptual question that emerges from our framework is whether ``boundary symmetries'' should be quotiented out or not: after all only ``bulk symmetries'' are in the kernel of the (on-shell) symplectic structure and their reduction yields a symplectic space, which could be a viable candidate for the phase space (this corresponds to the first-stage reduced phase space described above). Residual boundary symmetries could therefore be  interpreted differently than ``gauge''. In the literature, these are commonly thought of as a new type of global symmetries (see e.g.\ \cite{Carlip1995, Giulini1995-covering, Balachandran1994, DonnellyFreidel16,HenneauxTrossaert} among many others).

Our interpretation of boundary symmetries as gauge hinges on the meaning we assign to ``boundaries'' of region in field theory.
A manifold with boundary $(\Sigma,\pp\Sigma)$ is here thought of as the closure of an open subspace of a boundary-less universe. From this perspective, we find it natural to demand that the the``observables'' that are supported on $(\Sigma,\pp\Sigma)$ be represented by appropriate observables supported on the entire universe. 

Since observables are gauge-invariant by definition, this suggests that the observables associated to $(\Sigma,\pp\Sigma)$ ought to be invariant with respect to the action of the entire gauge group---i.e.\ including those gauge transformations that are nontrivial at the boundary. 
Insisting on this point, and thus proceeding to the \emph{second} stage of reduction, we obtain a fully reduced phase space that is only Poisson and hence exhibits a superselection structure defined by its symplectic foliation. 

However, regardless of the philosophical perspective on this question, we believe this article provides necessary clarifications on the degrees of freedom and symmetries present in both the first-stage (i.e.\ ``bulk'') and second-stage (i.e.\ ``boundary'') reduced phase spaces of Yang-Mills theory on a null codimension-$1$ submanifold of a Lorentzian manifold.

\subsection{Specifics}
The Hamiltonian description of field theory on a (spacetime) Lo\-rentzian manifold $M$ assigns, to a codimension-$1$ submanifold $\Sigma\hookrightarrow M$, a symplectic manifold of fields $(\X,\omega)$, together with a locus of constrained configurations $\C\subset\X$ representing a necessary condition\footnote{Observe that sometimes this condition is also sufficient, but it tipically requires restricting to spacelike hypersurfaces. We will not elaborate on this aspect any further, as it would lead us astray.} that fields should satisfy to be extendable to a solution of the Euler--Lagrange equations in $M$. In regular cases, $\C$ is a coisotropic submanifold and its reduction by the characteristic foliation $\C^\omega$ is a symplectic manifold $\uCo\doteq\C/\C^\omega$.

\emph{When $\Sigma$ is a closed manifold}, in local Hamiltonian gauge theory, the constraint surface $\C$ can be seen as the zero level-set of an equivariant momentum map for a Hamiltonian gauge group action $\G\circlearrowright\X$; the reduced phase space of gauge-inequivalent physical configurations ${\uuC}\doteq \C/\G$ coincides with $\uCo \doteq \C/\C^\omega$, and is therefore a symplectic manifold.

However, \emph{when $(\Sigma,\pp\Sigma)$ is a submanifold with boundary}, we have previously shown\footnote{See \cite{DonaldsonYM,MeinrenkenWoodward} for earlier results in the particular case of Chern--Simons theory.} that there exists a normal subgroup $\Go\subset\G$, called the \emph{constraint gauge group}, such that coisotropic reduction of the constraint set coincides with symplectic reduction w.r.t. the action of $\Go$, i.e.\ $\uCo \doteq \C/\C^\omega\simeq \C/\Go$ \cite{RielloSchiavina}.  We called $\uCo$ the \emph{constraint-reduced phase space}, and showed that it carries a residual Hamiltonian action of $\Gred\doteq\G/\Go$. Note that $\Go$ is an appropriate closure of the set of gauge transformations supported in the interior of $\Sigma$; in this sense $\Go$ and $\Gred$ give a precise meaning to the informal notions of ``bulk'' and ``boundary'' gauge transformations.

As a consequence of this residual action, the fully-reduced phase space ${\uuC} = \C/\G \simeq \uCo/\Gred$---defined as the space of constrained configurations modulo gauge---fails in general to be a symplectic manifold, and it is instead\footnote{This is true only up to important details related to the infinite-dimensionality of the problem. However, in some cases this can be checked explicitly, as we do in this paper.} only Poisson. In particular, ${\uuC}$ is foliated by symplectic leaves $\uuS_{[f]}$, which we called \emph{flux superselection sectors} and proved to be labelled by co-adjoint orbits $\mathcal{O}_f$ of \emph{fluxes}, i.e.\ elements $f\in \mathrm{Lie}(\Gred)^*\subset \mathrm{Lie}(\G)^*$ that are in the image of the (on-shell) momentum map $\iota^*_\C H\colon \C\to \mathrm{\Lie}(\G)^*$ (see Definitions \ref{def:LHGT} and \ref{def:constraintsurface}):
\[
{\uuC} = \bigsqcup_{\mathcal{O}_f} \uuS_{[f]}, \qquad \uuS_{[f]} = (\iota_\C^*H)^{-1}(\mathcal{O}_f)/\G.
\]

In this paper we present the Hamiltonian gauge-theory assignment, in the form of the data above, for Yang--Mills theory (YM) in the nontrivial case where $(\Sigma,\pp\Sigma)$ is a \emph{null}\footnote{The case of $\Sigma$ non-null is discussed in \cite{RielloSchiavina}.} codimension-$1$ submanifold-with-boundary of a Lorentzian manifold $M$, of the form $\Sigma \simeq I\times S$ with $S$ spacelike and $I=[-1,1]$ a null interval.
Starting from the standard Yang--Mills theory assignment of the geometric phase space $(\X,\omega_{\nYM})$ to a null submanifold $\Sigma$ (Definition \ref{def:geom-ph-sp}), we provide a description of the constraint-reduced symplectic manifold $\uCo$ and of the fully-reduced Poisson manifold ${\uuC}$, including an explicit characterisation of the flux superselection sectors $\uuS_{[f]}$.

More precisely, we find (Theorem \ref{thm:nonAb-constr-red}) that $\uCo$ is a smooth symplectic Fr\'echet manifold,\footnote{Provided the $\Go$-action is proper and the symmetry action generates a symplectically closed distribution (see \cite[Section 4]{DiezPhD} and references therein).} and that it can be described as a symplectic covering space of the \emph{extended Ashtekar--Streubel phase space} $\XAS$:\footnote{The symbol $\simeq_\loc$ denotes a local symplectomorphism. If $G$ is Abelian, we show that $\uCo \simeq \hat{\mathcal{A}}\times T^*\fg^S$ is a \emph{global} symplectomorphism, and the symplectic covering is given by the branches of the logarithm for the toric factors of $G$.}
\[
    \uCo\simeq_\loc \XAS \doteq \hat{\mathcal{A}}\times T^*G^S_0,
\]
where the Ashtekar--Streubel phase space $\hat{\mathcal{A}}\simeq C^\infty([-1,1],\Omega^1(S,\fg))$ is a symplectic space of ``spatial'' connections (see Definition \ref{def:hatA-sympl} after \cite{AshtekarStreubel}), and $G^S_0$ the identity component of the mapping group'' $G^S\doteq C^\infty(S,G)$, for $G$ the (connected) structure group of Yang--Mills theory.

The fibre $\mathcal{K}$ of the covering $\uCo \to \XAS$ is the group of components of the relative (connected) mapping group $G^\Sigma_{0,\rel} \doteq \{ g\in G^\Sigma_0 \ : \ g\vert_{\pp\Sigma}=0\}$, which we explicitly characterise in a number of particular cases, such as when $G$ is simply connected and $\mathrm{dim}(\Sigma)=1,2,3$, or when $G$ is Abelian (Theorem \ref{thm:fibrecharacterisation}). (Earlier results in this direction can be found, e.g.\ in \cite{Giulini1995-covering}.)

To explicitly describe the (infinite dimensional) symplectic reduction $\uCo\simeq_\loc\XAS$, we develop a rigorous version of the dressing field method \cite{Francois2012}, related to a particular choice of gauge fixing of the $\Go$-action.
This provides a concrete model for the reduction $\uCo$, by means of  a \emph{nonlocal} map $\C\to \XAS$, which can be thought of as a dressing by a family of Wilson lines along the null direction of $\Sigma$.

The groups of residual gauge transformations acting on either the constraint-reduced phase space $\uCo$ or the model $\XAS$---i.e.\ $\Gred$ and $G^{\pp\Sigma}_0=G_0^S\times G_0^S$ respectively---are shown to differ by a discrete central extension by $\mathcal{K}$ (Proposition \ref{prop:Gred}), $\Gred/\mathcal{K} \simeq G_0^{\pp\Sigma}$. In both cases, the corresponding group actions are Hamiltonian (Propositions \ref{lem:residualAction} and \ref{prop:red-flux-map}), and the local symplectomorphism $\uCo\simeq_\loc\XAS$ is equivariant with respect to them.

These residual gauge actions must similarly be reduced, and it is only after this ``second stage'' reduction that one is led to the fully-reduced phase space 
\[
{\uuC}\doteq \C/\G \simeq \uCo/\Gred \simeq \XAS/G_0^{\pp\Sigma},
\]
which has the structure of a Poisson manifold and whose symplectic leaves we call (flux) superselection sectors. We describe these superselection sectors in Theorem \ref{thm:modelHamreduction}, and more explicitly in Theorem \ref{thm:memorySSS} (for the Abelian case).

The algebra of (classical) observables of the theory can then be identified with the space of  functions\footnote{\label{fnt:hamiltonianfunction}Truly, since ${\uuC}$ is an infinite-dimensional (partial) Poisson manifold, we have to restrict to Hamiltonian functions. See \cite[Definition 2.5]{RielloSchiavina}.} over the fully-reduced Poisson manifold ${\uuC}$. Since ${\uuC}$ is not symplectic, this algebra has a center $\mathcal{Z}\subset C^\infty({\uuC})$: the Casimirs of the Poisson structure which label the symplectic leaves. 

We stress that the Ashtekar--Streubel phase space $\hcalA$ does not correspond to either $\uCo$ or ${\uuC}$. It instead arises as one of the symplectic leaves of an \emph{intermediate} reduction of $\uCo\simeq_\loc \XAS$ by one of the two copies of $G^S_0 \hookrightarrow G^{\pp\Sigma}_0$; as such it also carries an action by the remaining copy of $G^S_0$. This fact can be summarised with the slogan: \emph{the Ashtekar--Streubel phase space is a partially reduced and partially superselected space}.

More specifically, in going from $\uCo$ to $\hcalA$, the initial (resp.\ final) value of the electric field is superselected (to zero, see Proposition \ref{prop:MaxPartialSSS} and Theorem \ref{thm:YMPartialSSS}, and the subsequent remarks), and the Ashtekar--Streubel space carries a residual (nonlocal) gauge action of $G^S_\f$ (resp.\ $G^S_\i$). An important byproduct of these results is the realisation that, in 4 spacetime dimensions, the residual symmetry acting on the Ashtekar--Streubel phase space can be identified with the ``soft/large gauge symmetry'' of \cite{HeMitra2014,KapecPateStrominger}.

(We note that, contrary to folklore, there is no need of fixing the \emph{magnetic} fields through the two components of $\pp\Sigma$ in order to have well defined phase spaces.)

Additionally, one of the labels for the superselection sectors of null, Abelian, Yang--Mills theory is the ``electromagnetic memory'' \cite{BieriGarfinkle,Pasterskimemory}---a gauge invariant quantity parametrised by the difference between the value of the electric fields at the two boundary components of $(\Sigma,\pp\Sigma)$ (Theorem \ref{thm:memorySSS}).
Finally, ``color memory'', as defined in \cite{PateRaclariuStrominger}, is \emph{not} the correct label for the superselection sectors---in particular it is not gauge invariant, as opposed to its non-linearised version (see Theorem \ref{thm:YMPartialSSS} and Remark \ref{rmk:YMpartialSSS}). As a consequence, ``color memory'' fails to generalise to the non-Abelian case the property of electromagnetic memory of being a viable superselection label.

\subsection{A note on quantisation}\label{sec:Introquantisation}
The considerations contained in this paper are purely classical. However, we believe that they provide insights about what is expected from quantisation. 

\emph{If} a quantisation of the (infinite dimensional) constraint reduction $\uCo\simeq_\loc \hcalA \times T^*G^S$ were available, in the sense of a Lie algebra morphism between its Poisson algebra of functions and operators on some Hilbert space that is equivariant under the action of $\Gred$, then one could extract a subrepresentation given by functions on $\uuC$ pulled back along $\uCo\to\uuC$ to $\Gred$-invariant functions on $\uuC$. Since $C^\infty(\uuC)$ has a center (the Casimirs of the Poisson structure), the subrepresentation will also have a center, and thus decompose the Hilbert space into ``blocks''. 

If a statement such as ``quantisation commutes with reduction'' \cite{GuilleminSternberg,HallKirwin} were to hold in this scenario (of $\uCo$ w.r.t.\ $\Gred$), one would have that the Hilbert ``blocks'' would also correspond to the quantisation of the symplectic leaves $\uuS_{[f]}$ of the Poisson manifold\footnote{The correct statement is that $\uuC$ is a partial Poisson manifold, meaning that there is a subalgebra of all smooth functions that admits a Poisson structure \cite{PelletierCabau,RielloSchiavina}.} $\uuC$. (Observe that Verlinde formula \cite{VerlindeFormula} was proven as an application of reduction by stages to Chern--Simons theory in \cite{MeinrenkenWoodward_Verlinde}. See also \cite[Section 8.1]{RielloSchiavina}.)  The quantisation of the Casimir functions of the fully reduced phase space should, then, generate the \emph{quantum} flux superselection sectors (see e.g.\ \cite{FrohlichEtAl1979AnnPhys,FrohlichEtAl1979PhysLett, Buchholz86, MundEtAl2022} as well as \cite{DHR,DHR2} and \cite{GiuliniSuperselectionRules} for a review).

To obtain a quantisation of $\uCo$ one can resort to several techniques. (1) Since $\uCo$ is symplectic, directly applying geometric quantisation is---at least in principle---an option, although highly nontrivial due to the infinite dimensionality of the problem. In \emph{null}, Abelian, YM theory, $\uCo$ has a relatively simple structure: it is (non-canonically) isomorphic to the direct product of two spaces of local fields (i.e. sections of a bundle), one affine ($\hcalA$) and the other a linear cotangent bundle ($T^*\fg^S$). In the non-Abelian case a similarly simple description holds locally in field space, with the linear cotangent bundle replaced by the non-linear $T^*G^S$.\footnote{Cf. \cite{RielloSchiavina} for results on $\uCo$ in non-null YM theory.}

(2) Another option is to apply Batalin--Fradkin--Vilkovisky (BFV) quantisation. The (classical)  BFV formalism starts with the data of $\C\subset\X$ as input and resolves it by constructing a (classical) complex $\mathfrak{C}^\bullet_{BFV}$ whose cohomology in degree zero is $C^\infty(\uCo)$ \cite{StasheffConstraints88,SchaetzBFV}. Then, quantisation of this structure outputs a (quantum) complex whose cohomology in degree zero can then be taken as a quantisation of $\uCo$ (ideally, a Hilbert space of ``states'' for the theory). Note that this procedure trades the addition of nonphysical (ghost) fields, required by the cohomological resolution, for linearity and locality: One now needs to quantise a local symplectic dg vector space instead of the nonlinear and nonlocal symplectic manifold $\uCo$.
This procedure fits within the program of quantisation of field theory on manifolds with boundary of \cite{CMR1,CMR2}, which has the advantage of communicating with the bulk Batalin--Vilkovisky (perturbative) quantisation of Yang--Mills theory. The BV-BFV bulk-to-boundary correspondence at the quantum (or at least semi-classical) level is key to link soft symmetries and their Ward identities to soft scattering phenomena, which was first analysed at physics level of rigor in \cite{Strominger} (for more references, see Section \ref{sec:physicalinterpretation}).
The BV quantisation of YM theory has been also studied within the perturbative algebraic quantum field theory setting \cite{HaagKastler,BrunettiFredenhagen} by \cite{rejzner2016perturbative} (although without boundaries). A classical link between BV-enriched PAQFT quantisation and the BV-BFV analysis is given by \cite{RejznerSchiavina}, where soft symmetries (and their conserved charges) are interpreted in this language. 

(3) Finally, since $C^\infty(\uuC)$ is Poisson, one could try to directly perform deformation quantisation. One obvious difficulty is that $\uuC$ has a more involved structure than $\uCo$ even in the Abelian (matterless, linear) theory---it is in particular spatially nonlocal. For interacting QFT's deformation quantisation of Poisson algebras over infinite dimensional manifolds is subject of current study \cite{HawkinsRejzner}. 

Consequently, we shall defer any further consideration on a rigorous quantisation of the relation between memory as a momentum map for soft symmetries and scattering phenomena to further work.
In Section \ref{sec:physicalinterpretation}, we provide a few more considerations on the interpretation and expectation we place on the quantisation of our classical picture, especially in relation to the role and emergence of superselection sectors---with references to the literature.

\subsection{Structure of the paper} 
In Section \ref{sec:theoreticalframework} we outline the preliminaries of Hamiltonian gauge theories on manifolds with corners, thought of as boundaries of codimension $1$ hypersurfaces over which the Hamiltonian theory is specified. This is mostly a review of \cite{RielloSchiavina}.

In Section \ref{sec:null-setup} we lay the geometric foundations of null Hamiltonian Yang--Mills theory, the topic we will focus on throughout the rest of the paper. We also introduce a number of key ingredients that will play a role in the remainder of the work, such as the Ashtekar--Streubel phase space and a novel ``Darboux basis'', expressed in terms of Fourier modes (Proposition \ref{prop:modedecomp}).

Section \ref{sec:YMSuperselectionShort} describes the superselection structure for null, Yang--Mills theory as a consequence of the general theorem \cite[Theorem 1]{RielloSchiavina}. This is a short-hand version of the paper which gives direct access to Section \ref{sec:alternatelabelmemory}.

Sections \ref{sec:reduction-first} and \ref{sec:reduction-second} are dedicated to an explicit description of the first and second stage of the symplectic reduction, leading to the reduced phase space of Yang--Mills theory described in Section \ref{sec:YMSuperselectionShort}. Appendix \ref{app:Abelianreduction} details the specifics of the Abelian case, where global results can be obtained.

Finally, in Section \ref{sec:alternatelabelmemory} we apply our results to the problem of asymptotic/soft symmetries, and we draw a detailed comparison between our work and the literature, most notably \cite{StromingerLectureNotes}.

While the logical development of the paper is linear, the reader interested in the applications to soft symmetries can skip Sections \ref{sec:reduction-first} and \ref{sec:reduction-second} at first, and go directly to Section \ref{sec:alternatelabelmemory} after Section \ref{sec:YMSuperselectionShort}, which is necessary to its understanding.  

\noindent\fbox{%
    \parbox{\textwidth}{\small\center The reader is invited to use Appendix \ref{app:notations} as quick reference, as it summarises the notations used throughout the paper.}}

\subsection*{Acknowledgements}
We thank Kasia Rejzner for useful discussions. AR thanks Hank Chen for a fruitful discussion on homotopy theory.
Research at Perimeter Institute is supported in part by the Government of Canada through the Department of Innovation, Science and Economic Development and by the Province of Ontario through the Ministry of Colleges and Universities. The University of Waterloo and the Perimeter Institute for Theoretical Physics are located on the traditional territory of the Neutral, Anishnaabe and Haudenosaunee peoples.

\section{Theoretical framework}\label{sec:theoreticalframework}    

In this section we review the theoretical framework for the symplectic reduction of a locally Hamiltonian gauge theory. This framework was developed in a previous publication \cite{RielloSchiavina} to which we refer for details. A summary of the application of this theoretical framework to Maxwell theory on a spacelike submanifold of a Lorentzian manifold is given in Appendix \ref{app:Maxwell-short}. This condenses the results of \cite{RielloSchiavina}, as well as earlier results from \cite{RielloSciPost} (see also \cite{RielloEdge} for a pedagogical overview, as well as \cite{RielloGomesHopfmuller, RielloGomes}).

\medskip

In this introductory section, $(\Sigma,\pp\Sigma)$ denotes a smooth, orientable, manifold with boundary, with
\[
    n \doteq \dim\Sigma\geq1.
\]
This manifold should be understood as a boundary component of a spacetime manifold $M$ (with corners), however we will omit discussing the Lagrangian origin of the field-theoretic data we employ. The induction of a locally Hamiltonian gauge theory (on $\Sigma$) from a Lagrangian theory (on a cobordism of $\Sigma$) is explained in \cite[Appendix D]{RielloSchiavina} (see also \cite{KT1979}).

\begin{remark}[Null $\Sigma$]
If $\Sigma$ is a compact null submanifold of a globally hyperbolic manifold $M$, then $\pp\Sigma \neq \emptyset$. The goal of this article is to explore the role of $\pp\Sigma$.
\end{remark}

\subsection{Locally Hamiltonian gauge theory}

In short, the notion of a \emph{locally} Hamiltonian gauge theory is a special case of an infinite-dimensional Hamiltonian $\G$-space for which there exist stronger, \emph{local}, versions of all the defining quantities and relations---and in particular of the Hamiltonian flow equation---which must hold \emph{pointwise} over $\Sigma$. The reader not familiar with the notion of locality is referred to Appendix \ref{rmk:locality} and references therein.

\begin{definition}[Equivariant locally Hamiltonian gauge theory \cite{RielloSchiavina}]\label{def:LHGT}
A \emph{locally Hamiltonian gauge theory}  $(\X,\bom,\G,\bH)$ over $\Sigma$ is given by 
\begin{enumerate}[label=(\emph{\roman*})]
\item a \emph{space} of local fields $\phi\in \X \doteq \Gamma(\Sigma,E)$, for $E\to \Sigma$ a fibre bundle, called ``geometric phase space'',
\item\label{def:LHGT-sympldensity} a \emph{local symplectic density} $\bom \in \oloc^{\mathrm{top},2}(\Sigma\times \X)$, i.e.\ a $\bd$-closed local $(\mathrm{top},2)$-form such that $\omega = \int_\Sigma \bom \in \Omega^2(\X)$ is  symplectic\footnote{\label{fnt:symplecticconditions}A 2-form $\omega$ on $\X$ is weakly (resp. strongly) symplectic iff it is closed and $\omega^\flat: T\X \to T^*\X$ is injective (resp. bijective). In this paper, all symplectic forms are ``weak'', and we thus drop the qualifier.};
\item a \emph{local Lie algebra action} $\rho : \X \times \fG \to T\X$ of a real  Lie algebra $\fG$, that exponentiates to an action of the (connected) Lie group $\G \doteq \langle \exp \fG \rangle$,
\item an $\mathbb{R}$-linear local map $\bH: \fG \to \oloc^{0,\mathrm{top}}(\Sigma\times \X), \ \xi \mapsto \langle \bH, \xi\rangle$, called  \emph{(co-)momentum form} which is \emph{equivariant},\footnote{\label{fnt:centralextension}A more general treatment of the equivariance properties of $\bH$ allowing for nontrivial \emph{boundary} Chevalley--Eilenberg cocycles---and thus central extensions---is given in \cite[Sections 3.5 and 3.6]{RielloSchiavina} (see also \cite{MeinrenkenWoodward}).} $\mathbb{L}_{\rho(\xi)} \langle \bH,\eta\rangle (\varphi) = \langle \bH(\varphi), [\xi,\eta]\rangle$ for all $\xi,\eta\in\fG$ and $\varphi\in\X$,
\end{enumerate}
such that the following \emph{local Hamiltonian flow equation} holds:
\[
\bi_{\rho(\xi)} \bom = \bd \langle \bH, \xi\rangle \qquad \forall \xi \in \fG.
\qedhere
\]
\end{definition}

\begin{remark}[Duals]\label{rmk:dualspaces} 
Let $\mathcal{W}=\Gamma(M, W)$ be the space of sections of a vector bundle over a compact manifold, $W\to M$. This space can be given the structure of a nuclear Fr\'echet vector space, and one can consider the topological dual with strong topology $\mathcal{W}^*_{\mathrm{str}}$, which is itself a nuclear vector space. Being $\mathcal{W}$ a space of sections, we can further introduce the useful notion of \emph{local dual} $\mathcal{W}^*_\loc\subset \mathcal{W}^*_{\mathrm{str}}$ 
given by integrals of \emph{local}, $\mathbb{R}$-linear, maps from $\mathcal{W}$ into $\Gamma(M,\mathrm{Dens}(M)) \simeq \Omega^\text{top}(M)$ (see Appendix \ref{rmk:locality}).
If one instead considers sections $\Gamma(M,W^*\otimes_M \mathrm{Dens}(M))$, a nondegenerate pairing with $\mathcal{W}$ is given by integration on $M$. The subset of the local dual given by integrating against an element of $\Gamma(M,W^*\otimes_{M} \mathrm{Dens}(M))$ is called \emph{densitised dual}, and we simply denote it by $\mathcal{W}^*$:%
\[
 \int_M \colon \Gamma(M,W^*\otimes_M \mathrm{Dens}(M))\to \mathcal{W}^*_{\mathrm{loc}}, \qquad \mathcal{W}^*\doteq \Im\left(\int_M\right).
\]

Another characterisation of $\mathcal{W}^*$ is in terms of \emph{ultra}local elements $\alpha\in\mathcal{W}^*_\loc$, i.e. elements such that $\langle\alpha,w\rangle$ does not involve any derivative of $w\in\mathcal{W}$.

A thorough discussion of the subtleties arising from dualisation in this setting is given in \cite[Appendix A]{RielloSchiavina}.
(See also \cite[Definition 2.9]{RielloSchiavina} where the densitised dual is instead denoted by $\bullet^\vee$ rather than $\bullet^\ast$, which there stands for the strong dual).
\end{remark}

\subsection{Constraint and flux forms}

Whenever a Lagrangian gauge theory over a spacetime $M$ admits a locally-Hamiltonian formulation at a codimension-1 hypersurface $\Sigma\hookrightarrow M$, a relation between the symplectic generator of gauge transformation and a subset of the equations of motion known as constraints can be established \cite{LeeWald}. In our framework this relation is captured by the following proposition/definition \cite[Appendix D]{RielloSchiavina}:

\begin{definition}[Constraint and flux forms] \label{def:constraintsurface}
A local momentum form $\bH \in \oloc^{\mathrm{top},0}(\Sigma\times\X,\fG^*_\loc)$ can be uniquely written as the sum
\[
\bH  = \bHo  + d \bh ,
\]
where $\bHo$ is \emph{order-0}, and $d\bh$ is $d$-exact  \cite[Prop.4.1]{RielloSchiavina}; $\bHo$ and $d\bh$ are separately equivariant. 

We call $\bHo$ the \emph{constraint form}, $d\bh$ the \emph{flux form} associated to the momentum form $\bH$, and $h = \int_\Sigma d \bh$ the \emph{flux map}.

The \emph{constraint surface} $\C\subset \X$ is the vanishing locus of the constraint form, i.e.
    \[
    \C = \{ \phi\in\X \ : \ \langle\bHo(\phi),\xi\rangle = 0 \ \forall \xi\in\fG\}.
    \]
We denote by $\iota_\C$ the embedding $\C\hookrightarrow\X$ and refer to $\C$ as ``the shell''.
\end{definition}

The split of $\bH$ into $\bHo+d\bh$ can be thought of as integration-by-parts. Note that we will often view the flux map $h$ in a ``dual'' manner, i.e.\ as a map $ \X \to \fG^*_\loc$; we will do so without changing the notation for $\h$. The flux map and its properties control our entire construction.

\begin{remark}[Noether current, charges, and the constraints]\label{rmk:Noethercharge}
For a general analysis of the properties a Lagrangian gauge theory must satisfy to yield a locally Hamiltonian gauge theory (Definition \ref{def:LHGT}), we refer to \cite[Appendix D]{RielloSchiavina} (see also \cite{AshtekarBombelliKoul, CrnkovicWitten1987, Zuckerman} and in particular \cite{LeeWald}).
In particular, if these conditions are met the (co-)momentum form $\bH :\fG\to\Omega^{\mathrm{top},0}(\Sigma\times\X)$ is the pullback to $\Sigma$ of the (off-shell) Noether current associated to a gauge symmetry.
The Noether current then encodes the constraint of the theory---which are given by $\bHo$---up to a total derivative, the flux form $d\bh $; the flux map $\h$ is the (boundary) on-shell Noether charge.
The constraint $\C$ is the space of ``physical'' configurations in $\X$.
\end{remark}


\subsection{Symplectic reduction: overview}\label{sec:symplred-overview}

Our ultimate goal is to apply the theory of Marsden--Weinstein--Meyer symplectic reduction \cite{MarsdenWeinstein,Meyer,RatiuOrtega03}, or the appropriate infinite-dimensional generalisation thereof (e.g. \cite{DiezPhD} and references therein), to understand the symplectic properties of the reduced phase space ${\uuC}$---defined as the space of on-shell configurations modulo \emph{all} gauge transformations:

\begin{definition}[Reduced phase space]\label{eq:fully-red-space} 
Let $(\X,\bom,\G,\bH)$ be a locally Hamiltonian gauge theory.  The \emph{(fully-)reduced phase space} of the theory is $ {\uuC} \doteq \C / \G$.
\end{definition}

Once a locally Hamiltonian gauge theory is specified as per Definition \ref{def:LHGT}, the symplectic reduction of $(\X,\omega)$ by $\G$ proceeds by stages \cite{RielloSchiavina} (for a general reference to reduction by stages, see \cite{marsden2007stages,AbrahamMarsden}).

The two stages can be \emph{loosely} described as follows:
\begin{enumerate}
\item The first stage, called \emph{constraint reduction}, is the reduction of $(\X,\omega)$ by ``bulk'' gauge transformations $\Go$ generated by $\bHo$ alone; it implements coisotropic reduction of the constraint set. 
\item The second stage, called \emph{flux superselection}, takes care of the residual group of ``boundary'' gauge transformations $\Gred \doteq \G/\Go$ which can be thought of as being generated by the flux map $h$. 
\end{enumerate}
(If $\pp\Sigma =\emptyset$, only the first stage is relevant: $\Go=\G$ and one recovers the standard gauge reduction procedure.)

The main difference between the first and second stages is that whereas in the first stage the reduction procedure is performed at a canonical value of the momentum map, i.e.\ at $\bHo=0$ (thus coinciding with the coisotropic reduction of the constraint configurations, $\C/\C^\omega$), in the second stage one is \emph{free} to choose at which value $f\in \fG^*_\loc$ of the flux map $h$ to reduce (as long as the chosen value is compatible with the constraints). This freedom of choosing $f$---or, equivalently, of choosing a coadjoint orbit $\mathcal{O}_f\subset \fG^*_\loc$---is the origin of the fact that $\C/\G$ is not symplectic but instead a (continuous) disjoint union of symplectic leaves, i.e.\ a Poisson manifold.

In sum: whenever $\pp\Sigma = \emptyset$, the space  ${\uuC}$ is symplectic, but when $\pp\Sigma\neq\emptyset$ then ${\uuC}$ is in general at best a Poisson manifold. This Poisson is foliated by disjoint symplectic leaves, each labelled by certain coadjoint orbits in $\mathcal{O}_f \subset \fG^*_\loc$. We will call these symplectic leaves \emph{superselection sectors} and denote them $\uuS_{[f]}$.

This two-stage reduction procedure is summarised in the following commutative diagram. The goal for the remainder of this section is to explain it.

\begin{equation}\label{e:bigdiagram}
\xymatrix@C=.75cm@R=1cm{
(\X,\omega)
    \ar@{~>}[rr]^-{\tbox{2.2cm}{constraint reduction \\(w.r.t. $\Go$ at $0$)}}
&&(\uCo,\uomegao)
    \ar@{~>}[rr]^-{\tbox{2.2cm}{flux superselection (w.r.t. $\Gred$ at $\mathcal{O}_f$)}}
&&(\uuS_{[f]},\uuomegao_{[f]})\\
&{
    \;\qquad\C\qquad\;
    \ar@{_(->}[ul]^-{\iota_\C}
    \ar@{->>}[ur]_-{\pi_\circ}
    }
&&{
    \;\;\uh^{-1}(\mathcal{O}_{f})
    \ar@{_(->}[ul]^-{\underline{\iota}_{[f]}}
    \ar@{->>}[ur]_-{\underline{\pi}_{[f]}}
    }\\
&&{
    \;\;\h^{-1}(\mathcal{O}_f)\cap\C
    \ar@{_(->}[ul]^-{{\iota}^\C_{[f]}}
    \ar@{->>}[ur]_-{\pi_{\circ,[f]}}
    \ar@{_(-->}@/^2.7pc/[uull]^{\iota_{[f]}}
    \ar@{-->>}@/_2.7pc/[uurr]_{\pi_{[f]}}
    }\\
}
\end{equation}
~

\subsection{Constraint reduction}
We start by introducing some definitions and results that allow us to identify the ingredients entering the reduction procedure:

\begin{definition}[Annihilators]\label{def:annihilators}
    Let $\mathcal{W}$ be a nuclear Fr\'echet vector space, or the strong dual of a nuclear Fr\'echet space.\footnote{Recall, if $\mathcal{W}$ is nuclear Fr\'echet, its strong dual is nuclear but in general not Fr\'echet. However, if $\mathcal{W}$ is either a nuclear Fr\'echet vector space or the strong dual of one, $\mathcal{W}$ is reflexive, i.e.\ it is canonically isomorphic to its strong bi-dual: $\mathcal{W}^{**}_{\mathrm{str}} = \mathcal{W}$.} Moreover, let $\mathcal{X}\subset \mathcal{W}$ and $\mathcal{Y}\subset \mathcal{W}^*_{\mathrm{str}}$ be subsets. The \emph{annihilator of $\mathcal{X}$ in $\mathcal{Y}$} is the set 
    \begin{align*}
    \Ann(\mathcal{X},\mathcal{Y}) &= \{y\in \mathcal{Y}\, \vert\, \langle y,x\rangle = 0\  \forall x\in \mathcal{X} \} \subset \mathcal{W}^*_{\mathrm{str}}.
    \qedhere
    \end{align*}
\end{definition}

\begin{lemma}\label{lem:doubleannihilator}
Let $\mathcal{W}$ be as above, and $\mathcal{X}\subset \mathcal{W}$ be a closed vector subspace, then 
\[
\Ann(\Ann(\mathcal{X},\mathcal{W}^*_{\mathrm{str}}),\mathcal{W}) =\mathcal{X}.
\]
\end{lemma}

\begin{proof}
See Appendix \ref{app:proof-annihilators}
\end{proof}

\begin{definition}\label{def:fluxes+constr-red}~
\begin{enumerate}[label=(\roman*),leftmargin=*]
\item The space of \emph{(on-shell) fluxes} is
\[
\F \doteq \Im(\iota_\C^* h) \subset \fG^*_\loc;
\]
its elements are denoted by $f\in\F$, and their coadjoint orbits by $\mathcal{O}_f \subset \F$.
\item The \emph{constraint gauge algebra} $\fGo \subset \fG$ is the annihilator of the space of fluxes, i.e.\footnote{Although ultimately equivalent, this definition is not the same as the one given in \cite{RielloSchiavina}, see Theorem 4.33 \emph{ibidem}. Observe that this definition is well posed because $\F\subset \fG^*_\loc \subset \fG^*_{\mathrm{str}}$ and $\fG=\fG^{**}_{\mathrm{str}}$.\label{fnt:fGo}}
\[
\fGo \doteq  \Ann(\F,\fG)  \equiv \{ \xi \in \fG \ : \ \langle h(\phi), \xi\rangle = 0\ \forall \phi\in\C\}.
\]
\item The \emph{constraint gauge group} is the subgroup of $\G$ generated by $\fGo$:
\[
\Go \doteq \langle \exp \fGo \rangle \subset \G
\]
\item The first-stage reduced, or \emph{constraint-reduced}, phase space is 
\[
\uCo \doteq \C/\Go;
\]
we denote $\pi_\circ : \C \onto \uCo$ the corresponding projection. \qedhere
\end{enumerate}
\end{definition}

\begin{remark}[Equivariance]
Observe that the latter definition makes sense because $\Go\subset \G$ and the action of $\G$ preserve $\C$ as a consequence of the equivariance of $\bHo$ (Definition \ref{def:constraintsurface}).
\end{remark}

From \cite[Theorem 4.33]{RielloSchiavina} it follows in particular that:
\begin{proposition}\label{prop:Go-normal}
The constraint algebra $\fGo$ is a Lie ideal of $\fG$ and hence $\Go$ is a normal subgroup of $\G$.
\end{proposition}

\begin{remark}[Smoothness of $\uCo$]
Although in the following theorem the smoothness of $\C$ is assumed, in the specific cases treated in this article its smoothness can be proven. Similarly, in these cases, one can prove the (weak) non-degeneracy of $\uomegao$ defined as in the following theorem; the relevant assumption in the theorem is the symplectic closure of $\rho(\fGo)$ \cite[Chapter 4]{DiezPhD}.
\end{remark}

\begin{theorem}[First stage: Constraint reduction]
If $\rho(\fGo)$ is symplectically closed and the space $\uCo\doteq \C/\Go$ is smooth then $\uCo\simeq \C/\C^\omega$ when equipped with the unique (symplectic) 2-form $\uomegao\in\Omega^2(\uCo)$ such that
\[
\pi_\circ^*\uomegao = \iota_\C^*\omega, \qquad \pi_\circ\colon \C\onto\uCo.
\]
 We call $(\uCo,\uomegao)$ the constraint-reduction of $(\X,\omega)$.
\end{theorem}

Diagrammatically, this shows the leftmost part of the diagram in Equation \eqref{e:bigdiagram}:
\[
\xymatrix@C=.75cm{
(\X,\omega)
    \ar@{~>}[rr]^-{\tbox{2.2cm}{constraint reduction \\(w.r.t. $\Go$ at $0$)}}
&&(\uCo,\uomegao)\\
&{
    \;\C\;
    \ar@{_(->}[ul]^-{\iota_\C}
    \ar@{->>}[ur]_-{\pi_\circ}
    }
}
\]

\noindent In the diagram above, ``w.r.t. $\Go$ at 0" stands for: ``with respect to the action of $\Go$ at $J_\circ^{-1}(0)$, the \emph{zero}-level set of the corresponding momentum map $J_\circ$''.


\subsection{Flux superselection}\label{sec:theoframe-fluxsupersel}

If $\pp\Sigma\neq \emptyset$, the constraint-reduced phase space $(\uCo,\uomegao)$ fails to be fully gauge invariant: on it one still has the residual action of the \emph{flux group} $\Gred$ whose Lie algebra is the flux gauge algebra $\fGred \doteq \fG / \fGo$. Moreover, this action is itself Hamiltonian. Indeed, a consequence of the equivariance properties of $\bH$, and the fact that $\fGo\subset \fG$ is a Lie ideal, is that the locally Hamiltonian action of $\G$ on $(\X,\bom)$ descends to a Hamiltonian action of the flux gauge group $\Gred$ on the constraint-reduced phase space $(\uCo,\uomegao)$:

\begin{proposition}[Hamiltonian action on $(\uCo,\uomegao)$]~\label{prop:resHamaction}
\begin{enumerate}[label=(\roman*)]
    \item The \emph{flux gauge algebra} is $\fGred\doteq \fG/\fGo$ and 
    \[
    \Gred\doteq\G/\Go
    \]
    is the \emph{flux gauge group}; one has $\mathrm{Lie}(\Gred)=\fGred$.
    \item One can identify
    $\fGred^*_{\mathrm{str}} \simeq \Ann(\fGo, \fG^*_{\mathrm{str}})$, and there is a unique map
    \[
    \uh: \uCo \to \fGred^*_{\mathrm{str}}\quad\mathrm{such\ that}\quad \pi_\circ^* \uh = \iota_\C^*h.
    \]
    We call it the \emph{reduced flux map}. Furthermore, $\F \doteq \Im(\iota^*_\C h) \simeq \Im(\uh)$.
    \item The action $\rho:\X\times \fG \to T\X$ descends to an action $\underline\rho:\uCo \times \fGred \to T\uCo$.
    \item The action $\underline{\rho}$ of $\fGred$ on $(\uCo,\uomegao)$ is Hamiltonian with momentum map $\uh$, i.e.\ for every $\underline \xi \in \fGred$
    \[
    \bi_{\underline{\rho}(\underline \xi)}\uomegao = \bd \langle \uh, \underline \xi\rangle.
    \]
\end{enumerate}
\end{proposition}

\begin{remark}\label{rmk:Gredlocality}
Proposition \ref{prop:resHamaction} is phrased in terms of the strong dual $\fG^*_{\mathrm{str}}$, and it characterises $\Im(\uh)\subset\fGred^*_{\mathrm{str}}$ as a subset of $\fG^*_{\mathrm{str}}$. Note that according to our definitions $\Im(\iota^*_\C h) \subset \fG^*_\loc$ and $\Im(\uh) \subset \fGred^*_{\mathrm{str}}$ live a priori in different dual spaces. However, the statement in the proposition makes sense because in the strong dual there exists the dual of the projection map, $\fGred^*_{\mathrm{str}} \to \fG^*_{\mathrm{str}}$, as well as the embedding $\fG^*_\loc \into \fGred^*_\mathrm{str}$. In the present article we look at null YM theory, where $\fGred$ is either $\fg^{\pp\Sigma}\simeq C^\infty(\pp\Sigma,\fg)$ or $\fg^{\pp\Sigma}/\fg$ (see Definition \ref{def:mappinggroups} to fix the notation). 
We can relate to the previous discussion noting that, in this case, we also have the embedding of the densitised dual in the strong dual $\fGred^* \hookrightarrow \fGred^*_{\mathrm{str}}$.
In fact, in this case, Proposition \ref{prop:resHamaction} can be phrased in terms of \emph{densitised} and \emph{local} duals as: 
\[
\fGred^*\simeq \Ann(\fGo, \fG_\loc^*),
\]
with $\fGred^*\simeq(\fg^{\pp\Sigma})^*$ or $\fGred^*\simeq\Ann(\fg,(\fg^{\pp\Sigma})^*)\subset (\fg^{\pp\Sigma})^*$, respectively. Further details can be found in Lemma \ref{lem:annihilators} and Proposition \ref{prop:fluxes}.
\end{remark}

We can thus proceed to a second stage of the reduction procedure. We start by introducing the \emph{flux superselection sectors} via an orbit-reduction procedure:

\begin{lemma}[Flux superselection sectors]\label{lem:fluxSS}
Consider a flux $f\in \F \simeq \Im(\uh) \subset \fGred^*$ and its coadjoint orbit $\mathcal{O}_f\subset \fGred^*$. Then
\[
\uuS_{[f]} \doteq \uh^{-1}(\mathcal{O}_f)/\Gred
\]
carries a symplectic 2-form $\uuomegao_{[f]}\in\Omega^2(\uuS_{[f]})$ uniquely defined by the equation
\[
\underline{\pi}^*_{[f]} \uuomegao_{[f]} = \underline\iota_{[f]}^*\big( \uomegao - \uh^*\Omega_{[f]}\big),
\]
where (i) $\underline\pi_{[f]} : \uh^{-1}(\mathcal{O}_f) \onto \uuS_{[f]}$ is the projection associated to the quotient by $\Gred$, (ii) $\underline\iota_{[f]}^* : \uh^{-1}(\mathcal{O}_f) \hookrightarrow \uCo $ is the natural embedding, and (iii) $\Omega_{[f]}$ is the canonical (Kirillov--Konstant--Souriau, or KKS) symplectic structure on $\mathcal{O}_f$. We call $\uuS_{[f]}$ a \emph{flux superselection sector}.
\end{lemma}

\begin{theorem}[Second stage: flux superselection]\label{thm:generalsuperselection}
The fully reduced phase-space ${\uuC} \doteq \C/\G \simeq \uCo/\Gred$ is a Poisson space whose symplectic leaves are the flux superselection sectors:
\[
{\uuC} \simeq \bigsqcup_{\mathcal{O}_f\subset \F} \uuS_{[f]}.
\]
\end{theorem}

Diagrammatically:
\[
\xymatrix{
\uh^{-1}(\mathcal{O}_f) \ar@{^(->}[rr]^-{\underline{\iota}_{[f]}} \ar@{->>}[d]_-{\underline{\pi}_{[f]}} && \uCo \ar@{->>}[d]^-{\underline{\pi}}\\
\uuS_{[f]} \ar@{^(->}[rr]_-{\tbox{2cm}{symplectic leaf}} && {\uuC}
}
\]
which in particular encodes the rightmost part of the diagram in Equation \eqref{e:bigdiagram}.

\begin{remark}[quantisation of superselections and Casimirs]
As a consequence of Theorem \ref{thm:generalsuperselection}, the pullback along $\uh$ of any Casimir function on\footnote{The Poisson structure on $\fGred^*$ naturally restricts to $\F$ in virtue of the equivariance of $\uh$. Cf. Footnote \ref{fnt:centralextension} for generalisations allowing for central extensions.} $\F\subset \fGred^*$ is a Casimir function on the Poisson manifold $\uCo$ that descend to a central function in ${\uuC}$. Assuming that a quantisation of $\uCo$---and of its algebra of Hamiltonian\footnote{Cf. Footnote \ref{fnt:hamiltonianfunction}.} functions---exists, one deduces that the quantisation of said Casimirs induces a decomposition into irreducible blocks of the quantum (Hilbert) space associated to $\uCo$. This is what motivates our use of the term ``superselection sectors''. More on this in Section \ref{sec:physicalinterpretation} (bullet point \ref{item:quantisation}).
\end{remark}

\begin{remark}[Noether charge algebra]
If $\pp\Sigma\neq\emptyset$, the would-be Dirac--Bargmann algebra of first class constraints, is replaced on-shell by an algebra of non-vanishing (boundary) Noether charges. Given the relationship between $\bH$ and the Noether current, and the fact that on-shell $H=\int_\Sigma\bH$ reduces to the flux map $h = \int_\Sigma d\bh$, one expects the boundary Noether charge algebra to coincide with the Poisson manifold $\F \subset \fGred^*$ and thus to provide a representation of the algebra $\fGred$.

A rigorous, and general, construction of the ``Noether charge algebra'' is given in \cite[Theorem 3]{RielloSchiavina} where a Poisson manifold $(\C_\pp, \Pi_\pp)$ of on-shell, boundary field configurations is constructed out of the flux map $\uh$. The space $(\C_\pp, \Pi_\pp)$ fibrates over the same space of superselection sectors as $({\uuC},\uuPi)$ and can hence be used as a simpler proxy for the latter when it comes to studying the space of superselections of a given theory. 
\end{remark}

A succinct exemplification of the results discussed in this section through Maxwell theory on a spacelike manifold with boundary is available in Appendix \ref{app:Maxwell-short}.


\section{Yang--Mills theory on null boundaries: geometric setup}\label{sec:null-setup}

In this section we outline the geometric setup underpinning the phase space of Yang Mills theory on a manifold with a \emph{null} boundary $\Sigma$.\footnote{More precisely: on a pseudo-Riemannian manifold $M$ with corners, whose boundary $\partial M$ has a null connected component $\Sigma$.}  The relationship between the definitions provided in this section and the theory's Lagrangian formulation is standard, and reviewed in \cite[Appendix D]{RielloSchiavina}.

\subsection{The null manifold}\label{sec:nullmfd}

Let $(\Sigma,\gamma)$ be a null, $n$-dimensional, manifold of signature $(1,\cdots,1,0)$, and cylindrical topology $\Sigma \simeq S \times I$, $I = [-1,1]$ with $S$ a connected, closed manifold. Since the boundary of $\Sigma$ has two diffeomorphic connected components, it will be convenient to denote:
\[
S_{\i/\f}\doteq S\times\{\mp 1\} \hookrightarrow \Sigma.
\]

Denote by $\ell\in\mathfrak{X}(\Sigma)$ the global null vector $\gamma(\ell,\ell)=0$. We consider coordinates $(x^i,u)$ on $S\times I$, such that the metric $\gamma$ and the null vector field $\ell$ take the form:
\[
\gamma = \gamma_{ij}(x)dx^i dx^j, \qquad \ell = \pp_u,
\]
with $\gamma_{ij}$ non degenerate. With reference to the spacetime picture (Remark \ref{rmk:spacetime}), we refer to $u$ as ``retarded time''. Notice that we assumed $\gamma$ to be $u$-\emph{in}dependent. This is done only for simplicity of exposition.

\begin{definition}[Spatial forms and vectors]\label{def:spatial-f+v}
A vector field $v\in\mathfrak{X}(\Sigma)$ and a differential form $\alpha\in \Omega^k(\Sigma)$ are called \emph{spatial} iff, respectively, $i_vdu = 0$ and $i_{\ell}\alpha = 0$. We denote by a hat the operation of projecting along the ``spatial'' foliation: for any $v\in\mathfrak{X}(\Sigma)$ and $\alpha \in \Omega^k(\Sigma)$, define
\[
\hat v \doteq v - (i_v du) \ell
\quad\mathrm{and}\quad
\hat \alpha \doteq \alpha - du \wedge (i_\ell \alpha),
\]
so that $\hat{v}$ and $\hat{\alpha}$ are spatial: $i_{\hat v} d u = 0$ and $i_\ell\hat\alpha=0$.
In coordinates, spatial vector and forms read $\hat v = \hat v^i\pp_i$ and $\hat \alpha = \hat \alpha_{i_1 \dots i_k} d x^{i_1} \wedge \dots\wedge dx^{i_k}$.

Denote by $\Omega^{\bullet}_{\mathrm{spatial}}(\Sigma)$ the space of horizontal forms, equipped with the nilpotent spatial differential
\[
D\hat \alpha \doteq d\hat \alpha - du \wedge L_\ell\hat\alpha, \qquad D^2 = 0,
\]
which can be identified with the de Rham differential over $S$, i.e.\ $D \equiv d_S$.
\end{definition} 

Note that $\gamma(v,w) = \gamma(\hat v, \hat w)$ and---despite the fact that the inverse of $\gamma$ is not defined---the quantities $\gamma^{-1}(\hat \alpha, \hat \beta)$ for $\alpha, \beta$ 1-forms (or covectors) are well defined. In our coordinate system, $\gamma(\hat v, \hat w) = \gamma_{ij}\hat v^i\hat v^j$ and $\gamma^{-1}(\hat \alpha, \hat \beta) = \gamma^{ij}\hat \alpha_i \hat \beta_j$, where $\gamma^{ij} \doteq (\gamma_{ij})^{-1}$ is the inverse of the positive definite metric on $S$. On spatial tensors, we use $\gamma^{ij}$ and $\gamma_{ij}$ to raise and lower indices.

\begin{notation}\label{notation:vol}
We denoted the measures over $\Sigma$ and $S$ respectively by
\[
\vol_\Sigma \doteq du \wedge \vol_S 
\qquad\mathrm{and}\qquad
\vol_S = \sqrt{\det(\gamma_{ij})}\ d^\mathrm{top}x.
\]
\end{notation} 

\subsection{The geometric phase space}\label{sec:geomphsp}

Let $G$ be a real Lie group which we will assume to be either (\emph{i}) Abelian, or (\emph{ii}) semisimple. 
Let $\fg\doteq\mathrm{Lie}(G)$ be its (real) Lie algebra, and denote by $\tr( \cdot \cdot )$ a non-degenerate, $\Ad$-invariant, bilinear form on $\fg$; if $\fg$ is semisimple, $\tr$ can be chosen to be its Killing form. 

Let $P\rightarrow\Sigma$ be a $G$-principal bundle over $\Sigma$, and
\[
\mathcal{A} \doteq \mathrm{Conn}(P) \ni A
\]
be the space of principal connections, or \emph{gauge potentials} over $P\to\Sigma$; this is the space of sections of the bundle $J^1P/G\to \Sigma$, which is locally modelled on $\Omega^1(\Sigma,\fg)$. 

For simplicity of exposition we will assume that:
\begin{assumption}[Connectedness \& Trivial bundles]\label{ass:trivialP}
$G$ is connected and the principal $G$-bundle $P\to \Sigma$ is trivial, i.e.\ $P\simeq \Sigma \times G$.
\end{assumption}

As a consequence, one has a (global, non-canonical) isomorphism
\[
\mathcal{A} \simeq \Omega^1(\Sigma,\fg).
\]

In our coordinate system over $\Sigma \simeq S \times I$, all gauge potentials $A\in\mathcal{A}$ can be decomposed as
\[
\mathcal{A}\simeq \mathcal{A}_\ell\times \hat{\mathcal{A}},\qquad
A {\mapsto} (A_\ell\, ,\, \hat A) \doteq ( i_\ell A \, , \,  A - (i_\ell A) d u ),
\]
where $\mathcal{A}_\ell \simeq \Omega^0(\Sigma,\fg)$ and $\hcalA \simeq \Omega^1_\mathrm{spatial}(\Sigma,\fg)$.

If $F = d A + \tfrac12[A,A]$ is the principal curvature of $A$, we introduce the spatial 1-form (in components, $F_{\ell i} \equiv F_{ui}$)
\[
F_\ell \doteq i_\ell F = F_{\ell i} dx^i \in \Omega^1_\mathrm{spatial}(\Sigma,\fg).
\]
Its variation is given by
\begin{equation}
    \label{eq:deltaF_ell}
\bd F_\ell = \cL_\ell \bd \hat A - \cD \bd A_\ell,
\end{equation}
where
\[
\cL_\ell \doteq L_\ell + [A_\ell, \cdot ] 
\quad \mathrm{and} \quad
\cD\doteq D + [\hat A, \cdot].
\]
When acting on $\fg$-valued scalars (0-forms), one can replace $\cL_\ell$ with $\pp_u + [A_\ell, \cdot]$.

Now, let $\Ad^* P \to \Sigma$ denote the associate coadjoint bundle\footnote{\label{fnt:ADP}The adjoint bundle is $\Ad P \doteq P\times_{\Ad} \fg = (P\times \fg)/\sim$, defined by the equivalence relation $(p\cdot g, \xi) \sim (p,\Ad_g\xi)$. The coadjoint bundle is defined analogously.} to $P$, i.e.\ $\Ad^* P \doteq P \times_{\Ad^*} \fg^*$. Thanks to Assumption \ref{ass:trivialP} this bundle is also trivial: $\Ad^* P\simeq \Sigma \times \fg^*$. Hence, we introduce the space of \emph{electric fields} 
\[
\mathcal{E} \doteq \Omega^{n-1}_{\mathrm{spatial}}(\Sigma,\Ad^* P) \stackrel{(\mathrm{ass.\ref{ass:trivialP}})}{\simeq}\Omega^{n-1}_\mathrm{spatial}(\Sigma, \fg^*), 
\]
and we denote by $\boldsymbol{E}\in \Omega^{n-1}_{\mathrm{spatial}}(\Sigma, \Ad P)$ the \emph{spatial} two form representing the electric field. Note that any $\boldsymbol{E}\in\mathcal{E}$ can be uniquely encoded in a Lie algebra-valued function $E\in C^\infty(\Sigma,\fg)$ subordinate to a choice of a fixed volume form\footnote{Recall, we assumed that the metric $\gamma$ is constant in $u$.} on $S$, viz.
\[
\mathcal{E} \simeq C^\infty(\Sigma, \fg),
\qquad
\boldsymbol{E} = \tr(E \,\cdot\,) \vol_S.
\]
To keep the notation more consistent with the literature, we note here that most of our formulas will be written in terms of $E$---as opposed to $\boldsymbol{E}$---and we will indeed leave the above isomorphism \emph{implicit}. Nevertheless, it pays off to keep in mind the definition of electric field given above.

\begin{remark}[The retarded-time evolution picture]\label{rmk:u-functions}
We can identify 
\[
\hat{\mathcal{A}}\simeq C^\infty(I,\Omega^1(S,\fg)),\qquad \mathcal{A}_\ell\simeq C^\infty(I,\Omega^0(S,\fg)),\qquad  \mathcal{E}\simeq C^\infty(I,\Omega^0(S,\fg)), 
\]
i.e.\ we can view the decomposition of fields along the spatial foliation as maps that assign to each value of the retarded time $u\in I$ a spatial 1- or 0-form on $S$. In the following, we will seamlessly switch between these different points of view.
\end{remark}

\begin{notation}\label{notation:tr}
For brevity:  $\int_N \tr(\ \cdot\ )\vol_N \equiv \int_N \btr(\ \cdot\ )$.
\end{notation}

\begin{notation}\label{not:subsupscripts}
Let $W$ be a vector space and $Q$ a spatial $W$-valued $p$-form, i.e.\ $Q\in \Omega^p_\mathrm{spatial}(\Sigma,W) \simeq C^\infty(I,\Omega^p(S,W))$. Then, we denote 
\begin{subequations}%
    \label{eq:def-i/f-dif/av}
\begin{align}
Q^{\i} \doteq Q(u=-1) \quad\mathrm{and}\quad
Q^{\f} \doteq Q(u=1),
\end{align}
\end{subequations}
and view them as elements of $\Omega^p(S,W)$.
Similarly, it is convenient to introduce:
\begin{align*}
Q^{\int} \doteq \int_{-1}^1 du'\, Q(u'),
\quad
Q^\av  \doteq  \tfrac12 (Q^\i + Q^\f), 
\quad\mathrm{and}\quad
Q^\dif \doteq  Q^\f - Q^\i.
\end{align*}
\end{notation}

\medskip

With all of this at hand, we can thus define:

\begin{definition}[Null YM: geometric phase space]\label{def:geom-ph-sp}
The \emph{geometric phase space of null Yang--Mills theory} (nYM) $(\X,\bom_{\nYM})$ is the space
\[
\X \doteq \mathcal{A} \times \mathcal{E} \simeq \hat{\mathcal{A}}\times\mathcal{A}_\ell \times\mathcal{E},
\]
equipped with the symplectic density
\begin{align*}
\bom_{\nYM} & \doteq  \Big( \tr(\bd \E \wedge \bd A_\ell ) + \tr( \bd F_\ell^i \wedge \bd \hat A_i) \Big) \vol_\Sigma \in \Omega_\loc^{\mathrm{top},2}(\Sigma\times \X).
\qedhere
\end{align*}
\end{definition} 

\begin{remark}\label{rmk:nondeg}
In Appendix \ref{app:proof-nondeg} we prove that $\bom_{\nYM}$ is indeed a symplectic density in the sense of Definition \ref{def:LHGT}\ref{def:LHGT-sympldensity}: that is, we prove $\bd\bom_{\nYM} = 0$ and $\ker(\omega_{\nYM}^\flat)=0$ where $\omega_{\nYM} = \int_\Sigma \bom_{\nYM}$ (Cf. Footnote \ref{fnt:symplecticconditions}.)
\end{remark}

An important space for nYM theory is the Ashtekar--Streubel phase space over $\Sigma$. We now discuss some of its basic properties.

\begin{definition}[Ashtekar--Streubel phase space \cite{AshtekarStreubel}]\label{def:hatA-sympl}
The \emph{Ashtekar--Streubel} (AS) phase space $(\hcalA,\omAS)$ over $\Sigma\simeq I\times S$ is the space of purely spatial connections $\ASa\in\hcalA \simeq C^\infty(I,\Omega^1(S,\fg))$ equipped with the symplectic 2-form
\[
\omAS \doteq  \int_\Sigma \btr( (L_\ell \bd \ASa_i) \wedge \bd \ASa^i)\in \Omega^2(\hcalA).\qedhere
\]
\end{definition}

The proof that $(\hcalA,\omAS)$ is a symplectic manifold is analogous to that for $(\X,\omega_\nYM)$ provided in Appendix \ref{app:proof-nondeg}. It is easy to see that the AS phase space can be embedded as a symplectic submanifold of the geometric phase space of nYM theory e.g. as the submanifold $\{(A,E)=(0,0)\}$. Similarly, $T^*\mathcal{A}_\ell$ can also be embedded in $\X$ as the symplectic submanifold $\{\hat{A}_i = 0\} \simeq \mathcal{A}_\ell \times \mathcal{E} \simeq T^*\mathcal{A}_\ell$. However, although $(\X,\omega_\nYM)$ is diffeomorphic to the product $\hcalA\times T^*\mathcal{A}_\ell$, the two are not symplectomorphic.

The space $(\hcalA,\omAS )$ is often taken as the phase space of ``physical degrees of freedom'' of YM theory on a null surface. One of the goals of this paper is to assess this statement by clarifying the relationship between $\omAS $ and $\omega_{\nYM}$. This is achieved in Theorems \ref{thm:nonAb-constr-red} and \ref{thm:memorySSS}, with further clarifications provided in Section \ref{sec:alternatelabelmemory} (see e.g. Proposition \ref{prop:MaxPartialSSS} and Theorem \ref{thm:YMPartialSSS} and the subsequent remarks).

The Ashtekar--Streubel field $\ASa_i$, seen as a function of $u\in[-1,1]$ (Remark \ref{rmk:u-functions}), can in principle be expanded on the Fourier modes $e^{\pm i\pi k u}$, with $k$ a positive integer. The AS symplectic structure will then be block-diagonal in $k$, with the real and imaginary parts of $e^{i \pi k u}$ conjugate to \emph{each other}.
There are however two major (related) problems with this expansion, both due to the fact that all the terms in the Fourier series are periodic. 

The first problem is that the quantity $\ASa^\dif \doteq \ASa(u=1)-\ASa(u=-1)$, which encodes the lack of periodicity, does not appear in the Fourier expansion but nevertheless features in $\omAS$ and plays a crucial role in the reduction procedure. In particular, it is central for our understanding of the memory effect (see Section \ref{sec:alternatelabelmemory}).
The second problem is that the zero-mode in the expansion,\footnote{Meaning the coefficient of $e^{i\pi 0 u}=1$.} being purely real, lacks a symplectic partner w.r.t. $\omAS $ \emph{among the Fourier modes $e^{i\pi ku}$}; this problem could have also been detected by noting that it is rather the real and imaginary parts of $\sqrt{\pi k}e^{i \pi k u}$ that are \emph{canonically} conjugate to each other, and these are not well-defined for $k=0$.

Both these issues can be solved by including in the Fourier analysis one extra ``zero'' mode linear in $u$---and then performing the Gram-Schmidt algorithm to find the (complex) Darboux basis $\{\psi_k(u)\}$ described in the following lemma.

\begin{lemma}
\label{lemma:ASmodes}
Equip $C^\infty(I,\mathbb{C})$ with the Hermitian structure 
\[
\mathbb{G}(\phi_1,\phi_2) \doteq -\frac{i}{2} \int_{-1}^1 (\dot \phi_1 \phi_2^* - \phi_1 \dot \phi_2^*) du .
\]
Then,\footnote{Although one can interpret $\mathfrak{Im}(\psi_{k=0})$ as the limit $k\to0$ of $\mathfrak{Im}(\psi_k)$, this limiting procedure would fail by a factor of 2 in the case of the $\mathfrak{Re}(\psi_k)$---say after setting $(-1)^k = \cos(\pi k)$. Also, $\psi_k$ and $\psi_{-k}$ are equal rather than complex conjugate to each other. For this reason in the following we only consider $\psi_k$ with $k\geq 0$.} 
\[
\psi_k(u) \doteq 
\begin{dcases}
1 +  i \frac{u}{2} & \mathrm{if } k=0\\
(-1)^k + \cos(\pi k u)  +i \frac{\sin(\pi k u) }{2\pi k} &\mathrm{if }k\geq 1
\end{dcases}
\]
is a (complex) orthonormal basis of $C^\infty(I,\mathbb{R})$, i.e.\ 
\[
\mathbb{G}(\psi_k,\psi_l) = \delta_{kl}
\qquad\text{and}\qquad \mathbb{G}(\psi_k,\psi_l^*) = 0,
\]
and for all $f\in C^\infty(I)$ the sequence $\{f_N\}_{N\in\mathbb{N}}$ converges uniformly,\footnote{Recall, uniform convergence of $f_N(u)\to f(u)$ means that for every $\epsilon>0$ there exists an $N_\epsilon$ such that $|f_{N}(u) - f(u)| < \epsilon$ for all $u\in I$ and $N>N_\epsilon$.}
\[
f_N \doteq\sum_{k=0}^N\left( \tilde f(k)^* \psi_k + \mathrm{c.c.}\right) \xrightarrow[N\to\infty]{\mathrm{unif.}} f, \qquad \tilde f(k) \doteq \mathbb{G}(\psi_k ,f).
\]
\end{lemma}

\begin{proof}
See Appendix \ref{proof:ASmodes}.
\end{proof}

\begin{proposition}[Ashtekar--Streubel mode decomposition]\label{prop:modedecomp}
Let $\fg = (\mathbb{R},+)$, and $\ASa\in\hcalA \simeq C^\infty(I,\Omega^1(S))$.
For each $k\in\mathbb{N}$, define the \emph{Ashtekar--Streubel $k$-mode} $\tilde{\ASa}_i(k)\in\Omega^1(S)$ as
\[
\tilde{\ASa}_i(k,x) \doteq \mathbb{G}( \psi_k, \ASa_i (x));
\]
in particular, the Ashtekar--Streubel zero-mode is
\[
2\, \tilde{\ASa}_i(k=0,x) = (\ASa_i^{\int}(x) - \ASa_i^\av(x)) + i \ASa_i^\dif(x),
\]
Then, the expansion of $\ASa_i(u,x)$ on the basis $\psi_k(u)$ converges uniformly, i.e.
\[
\sum_{k=0}^N \left(\,\tilde{\ASa}_i^*(k,x) \psi_k(u) + c.c.\right) \xrightarrow[N\to\infty]{\mathrm{unif.}} \ASa_i(u,x),
\]
and
\begin{align*}
    \omAS
         = & {2i} \int_S \sum_{k=0}^\infty  \bd\, \tilde{\ASa}^*_i(k,x)\wedge \bd\, \tilde{\ASa}^i(k,x)\ \vol_S,\\
         =& \int_S \bd\ASa^\dif \wedge \bd(\ASa^{\int} - \ASa^\av)\ \vol_S + {2i} \int_S \sum_{k=1}^\infty  \bd\, \tilde{\ASa}^*_i(k,x)\wedge \bd\, \tilde{\ASa}^i(k,x)\ \vol_S.
\end{align*}

\end{proposition}

\begin{proof}
See Appendix \ref{proof:ASmodes-sympl}.
\end{proof}

\begin{remark}
For $G$ a more general structure group, $\ASa\in \hcalA \simeq C^\infty(I,\Omega^1(S,\fg))$ takes values in $\fg$, and therefore $\mathbb{G}$ needs to be tensored with the bilinear form $\tr(\cdot\cdot)$. The above construction carries over with minimal changes, see Section \ref{sec:semisimplememory}. 
\end{remark}

\begin{remark}[Comparison with \cite{StromingerLectureNotes}]\label{rmk:Stromingercf}
The mode decomposition of Proposition \ref{prop:modedecomp} allows us to compare with the standard reference \cite{StromingerLectureNotes}, and address some of the issues raised there. 
The first important distinction is that our mode expansion is aimed at functions on a bounded interval $[-1,1]$, where it is rigorous, as opposed to the whole real line. This allows us to include the linear basis element $\mathfrak{Im}(\psi_{k=0}) \doteq u/2$ as well as the quantity $\ASa^\dif$.

Next, in \cite[Eq. 2.6.6--7]{StromingerLectureNotes}, the author extracts a $u$-constant term from $\ASa$, and assumes it to be exact. More precisely, the $u$-constant term is identified with $\ASa^\av$ and thus there is a decomposition: $\ASa = \ASa' + \ASa^\av$ with $\ASa^\av = D\phi$ for $\phi\in \fg^S$, constant in $u$, and $\ASa'$ the ``$u$-nonconstant'' remainder. Furthermore,  in \cite[Eq. 2.5.16]{StromingerLectureNotes} the quantity $N\in \fg^S$ is introduced, so that $D N = \ASa^\dif$ (analogously to $\ASa^\av$, the difference $\ASa^\dif$ is also assumed to be exact). We can use this dictionary to rewrite the second term in the ``zero mode'' contribution to $\omAS$ as defined in \cite[Equation 2.6.8]{StromingerLectureNotes} as:
    \[
    \text{\cite{StromingerLectureNotes}} \quad \int_S \gamma^{ij}\bd \pp_iN \wedge \bd \pp_j \phi \ \vol_S \
    \leftrightsquigarrow \ \int_S \gamma^{ij} \bd\ASa^\dif_i \wedge \bd\ASa^\av_j \ \vol_S \quad \mathrm{[here]}.
    \]
But the latter expression is \emph{not} the zero-mode part of $\omAS$, since this reads
    \[
    \omAS = \int_S \gamma^{ij} \bd\ASa_i^\dif \wedge \bd( \ASa_j^{\int} - \ASa_j^\av)\ \vol_S + (\text{AS-modes with $k>0$}).
    \]
Indeed, apart from the restrictive request of \cite{StromingerLectureNotes} that $\ASa^\av$ and $\ASa^\dif$ be exact (see below), the main difference between the two approaches is that the zero mode of the AS field, in our basis, reads instead
    \[
    2\tilde{\ASa}(k=0) = (\ASa^{\int} - \ASa^\av) + i \ASa^\dif,
    \]
with its real and imaginary parts canonically conjugate to each other.

This means that, in our symplectic basis $\wt{\ASa}(k,x)$, the symplectic companion of $\ASa^\dif=\mathfrak{Im}(\ASa(0))$---which in \cite{StromingerLectureNotes} is denoted $DN$, assumed to be exact---is the zero mode $\mathfrak{Re}(2\tilde{\ASa}(0))=\ASa^{\int} - \ASa^\av$. Note that  this in general differs from $\ASa^\av$ (although they coincide if $\ASa$ happens to be $u$-constant).

Finally, observe that we do \emph{not} assume that $\ASa^\av$ and $\ASa^\dif$---or equivalently $\ASa^\i$ and $\ASa^\f$---are $D$-exact, i.e.\ we do not assume that ``the magnetic field vanishes at the boundary
[of the null surface $\Sigma$]'' (cit.\ p.\ 23 \emph{ibid.}). That is, our analysis covers the case where the magnetic fields through $\pp\Sigma$ are included in the picture (cf.\ Footnote 5 \emph{ibidem}).
We expand on this comparison in Section \ref{sec:alternatelabelmemory}.
\end{remark}

\subsection{Gauge transformations}\label{sec:Gaugetransformations}

We start the discussion of the gauge group for null YM theory setting out the notation and providing some preliminary considerations on mapping algebras and group (see \cite{kriegl1997convenient,WockelPhD,Neeb-locallyconvexgroups} for extensive discussions).

\begin{definition}[Mapping algebras and groups]\label{def:mappinggroups}
Given a compact manifold with boundary $(N,\pp N)$, and $G$ a Lie group, the \emph{mapping Lie algebra} and \emph{mapping Lie group} are 
\[
\fg^N \doteq C^\infty(N,\fg)
\qquad\mathrm{and}\qquad
G^N \doteq C^\infty(N, G)
\]
equipped with the natural pointwise Lie algebra structure (i.e.\ $[\xi,\eta](x) = [\xi(x),\eta(x)]$) and group multiplication, respectively. 

Furthermore, the \emph{relative} mapping Lie algebra and group are the Lie ideal of functions that vanish at the boundary, and the normal subgroup of functions whose value at the boundary is the identity, respectively:
\begin{align*}
\fg^N_\rel  \doteq\{ \xi\in \fg^N\ : \ \xi\vert_{\pp N} = 0\}  
\quad\mathrm{and}\quad
G^N_\rel  \doteq \{ g \in G^N\ : \ g\vert_{\pp M} = 1\}
\end{align*}

Finally, if $G$ is Abelian,\footnote{See Remark \ref{rmk:constantgaugetransf}, explaining why we introduce $\fg\into\fg^N$ only in the Abelian case.} we denote
\[
\fg\into \fg^N \quad\text{and}\quad G\into G^N
\]
the space of constant-valued functions in $\fg^N$ and $G^N$.
\end{definition}

\begin{definition}[Identity and relative components]\label{def:id+rel-comp}~
\begin{enumerate}[label=(\roman*)] 
\item For $\mathcal{H}$ a group, $\mathcal{H}_0$ is the \emph{identity component} of $\mathcal{H}$; if $\mathcal{H}$ is a subgroup, $\mathcal{H}_0$ still denotes the set of elements of $\mathcal{H}$ which are connected to the identity through paths that lie \emph{within} $\mathcal{H}$ itself.
\item For $\mathcal{H} \subset G^M$ a subgroup, the \emph{relative component} of $\mathcal{H}$ is $\mathcal{H}_\rel \doteq \mathcal{H} \cap G^N_\rel$, i.e.\ the set of elements of $\mathcal{H}$ that are equal to the identity at the boundary.
\end{enumerate}
In the following, we use commas instead of parentheses---e.g. $\mathcal{H}_{0,\rel,0} \equiv ((\mathcal{H}_0)_\rel)_0$---and, since $G= G_0$ is assumed connected, we forgo the parentheses around $G^N$ as well, e.g.\ $G^N_0 \equiv (G^N)_0$ etc.
\end{definition}

\begin{remark}\label{rmk:normalsubgrps}
If $\mathcal{H}$ is a subgroup of $\tilde{\mathcal{H}}$, $\mathcal{H}_0$ need not coincide with\footnote{This is also a subgroup of $\mathcal{H}$, but one for which we will not not need a notation.}  $\mathcal{H}\cap \tilde{\mathcal{H}}_0$, i.e.\ there might be elements in $\mathcal{H}$ which are connected to the identity within $\tilde{\mathcal{H}}$, but not within $\mathcal{H}\subset \tilde{\mathcal{H}}$. In particular, if $\mathcal{H}$ is a subset of $G^N$, $\mathcal{H}_0$ needs not coincide with $\mathcal{H}\cap G^N_0$, and in particular $G^N_{\rel,0}$ needs not coincide with $G^N_{0,\rel}$---a fact that will be relevant later.
\end{remark}

\begin{remark}[Locally exponential (sub)groups]\label{rmk:exponentialityofmappinggroup}
The mapping Lie algebra is the Lie algebra of the mapping group, $\fg^N = \mathrm{Lie}(G^N)$, and the mapping group is locally exponential, i.e.\ it admits an exponential map $\exp: \fg^N \to G^N$ which is a local diffeomorphism at the identity. Therefore, the identity component $G^N_0$ is the subgroup of $G^N$ generated by the mapping algebra, 
\[
G^N_0 = \langle \exp \fg^N\rangle.
\]
Moreover, although $\mathrm{Lie}(G^{N}_\rel) = \fg^N_\rel$ and $\fg^N/\fg^N_\rel \simeq \fg^{\pp N}$, exponentiating one only recovers the identity component:
\[
G^N_{\rel,0} = \langle \exp \fg^N_\rel\rangle 
\qquad\mathrm{and}\qquad
G^N/G^N_\rel \subset G^N/G^N_{\rel,0},
\]
and the equalities hold in the absence of topological obstructions. In Lemma \ref{lemma:pezzo2}, we will characterise the topological obstructions that might arise when $N = \Sigma \simeq I \times S$ and $S\simeq S^{n-1}$.
\end{remark}

After these preliminaries, we can now consider the gauge structure on $\X$.
Note that the following definitions \emph{exclude} gauge transformations that are not connected to the identity.

\begin{definition}[Gauge group and relative gauge group]\label{def:gaugetransformations}
Given a principal $G$-bundle $P\to\Sigma$ (trivial by Assumption \ref{ass:trivialP}), the local Lie algebra of gauge transformations, or \emph{gauge algebra}, is 
\[
\fG \doteq \Gamma(\Sigma,\Ad P) \stackrel{(\mathrm{Ass.\ref{ass:trivialP}})}{\simeq} \fg^\Sigma.
\]
Moreover, we call \emph{gauge group} the group generated by $\fG$,
\[
\G\doteq \langle \exp \fG \rangle \stackrel{(\mathrm{Ass.\ref{ass:trivialP}})}{\simeq}G^\Sigma_0,
\]
while the \emph{relative gauge group}  $\tGo$ is the normal subgroup of $\G$ corresponding to the relative component of $\G\simeq G^\Sigma_0$,
\[
\tGo \stackrel{(\mathrm{ass.\ref{ass:trivialP}})}{\simeq} G^\Sigma_{0,\rel} \equiv G^\Sigma_0 \cap G^\Sigma_\rel.\qedhere
\]
\end{definition}

\begin{remark}[Normality of $\G_\rel$]\label{rmk:Grel-normal}
Since $G^\Sigma_{0,\rel}$ is the intersection of two normal subgroups of the mapping group $G^\Sigma$, it is a normal subgroup itself. In particular, the relative gauge group $\tGo$ is a normal subgroup of the gauge group $\G$.
\end{remark}

\begin{definition}[Gauge action]\label{def:gauge-action}
The right action of the group of gauge transformations on the geometric phase space of nYM theory is
\[
\G\times \X \to \X,
\quad
(g,A,\E) \mapsto (g^{-1}Ag + g^{-1}dg, g^{-1}\E g).
\]
The corresponding \emph{fundamental vector fields} are
\[
\rho: \fG \to \mathfrak{X}(\X), \quad
\rho(\xi) = \int_\Sigma 
d_A\xi \frac{\delta}{\delta A} + [\E, \xi] \frac{\delta}{\delta E}.
\qedhere
\]
\end{definition}

\begin{remark}[Constant gauge transformations]\label{rmk:constantgaugetransf}
Mapping algebras arise as a particular case of the gauge algebra $\fG \doteq \Gamma(\Ad P)$: if $G$ is Abelian then $\Gamma(\Ad P) \simeq \fg^\Sigma$ without further assumptions, whereas in the non-Abelian case $\Gamma(\Ad P)\simeq \fg^\Sigma$ only if $P$ is trivial (Assumption \ref{ass:trivialP}). Therefore, the notion of ``constant gauge transformations" is generally meaningful (and, in fact, useful for us) only in the Abelian case. Ultimately, this is because, for $G$ Abelian, $\fg\into\fg^\Sigma$ is  a \emph{global} stabiliser---i.e.\ $\fg = \ker(\rho)$ for $\rho$ seen as a map $\fg^\Sigma \to \mathfrak{X}^1(\X)$ (see \cite[Remark 4.7]{RielloSchiavina}). We restricted the definition of $\fg\hookrightarrow\fg^N$ to the Abelian case to ensure that all our statements generalise to non trivial bundles.
\end{remark}

We can now show that the triplet $(\X,\bom,\G)$ indeed defines a locally Hamiltonian $\G$-space by identifying the momentum form $\bH:\fG\to\Omega^{\mathrm{top},0}(\Sigma\times \X)$. Subsequently, we will split $\bH$ into its constraint and flux form components, $\bHo$ and $d\bh$.

\begin{proposition}\label{prop:nYMisLocHam}
$(\X,\bom)$ is a locally Hamiltonian $\G$-space,
\[
\bi_{\rho(\xi)}\bom = \bd \langle \bH , \xi\rangle,
\]
with momentum form
\[
\langle \bH(A,E),\xi\rangle \doteq - \tr( \E \cL_\ell \xi + F_\ell^i \cD_i\xi)\vol_\Sigma.
\]
Moreover, with reference to Definition \ref{def:constraintsurface}, the momentum form splits into an equivariant (Gauss) constraint form $\bHo$ and an equivariant flux form $d\bh = \bH - \bHo$, respectively given by
\begin{subequations}
\label{eq:constr-and-flux-forms}
\begin{equation}
    \label{eq:gaussdef}
\langle\bHo(A,E),\xi\rangle = \tr(\mathsf{G}\, \xi)\vol_\Sigma,
\qquad
\mathsf{G}(A,E) \doteq \cL_\ell \E + \cD^i F_{\ell i},
\end{equation}
and
\begin{equation}
\label{eq:fluxformYM}    
\langle d\bh(A,E),\xi\rangle  = - \big( \pp_u \tr( \E\, \xi)  + D^i\tr( F_{\ell i}\, \xi)\big)\vol_\Sigma.
\end{equation}
\end{subequations}
Whence, the constraint surface $\C$ of ``on-shell'' configurations is the space of field configurations that satisfy the Gauss constraint:
\[
\C \doteq \bHo^{-1}(0) = \{(A,E)\in\X\ : \ \mathsf{G}(A,E)=0\}.
\]
\end{proposition}
\begin{proof}
Contracting $\bom$ with an infinitesimal gauge transformation $\rho(\xi)$, one obtains:
\begin{align*}
\bi_{\rho(\xi)} \bom 
& =   \big( \tr([\E,\xi]  \bd A_\ell ) -  \tr(\bd \E  \cL_\ell\xi )+ \tr( [F_\ell^i,\xi]   \bd \ASa_i) -\tr(  \bd F_\ell^i   \cD_i\xi) \big) \vol_\Sigma
\end{align*}
Rearranging, one finds the sought expression for the momentum form $\bH$.
We can then split this expression into its constraint and flux terms, $\bH = \bHo + d\bh$.
Indeed, from
\begin{align*}
\bH 
& = \big( 
\tr( (\cL_\ell \E  + \cD^iF_{\ell i} )\xi)
-L_\ell \tr( \E \xi)  - D^i\tr( F_{\ell i} \xi)  \big)\vol_\Sigma,
\end{align*}
we can readily isolate the part $\bHo$ of $\bH$ which is of order-0 in $\xi\in\fG$,\footnote{Recall: ``order-0'' means that $\langle\bHo,\xi\rangle$ is linear and ultralocal in $\xi$, i.e.\ it does not involve any derivative of $\xi$. This property uniquely determines the constraint form $\bHo$ once the momentum form $\bH$ is given, see \cite[Proposition 4.1]{RielloSchiavina}.} as well as the remainder flux form, $d\bh = \bH - \bHo$, as per Equations \eqref{eq:gaussdef} and \eqref{eq:fluxformYM}.
The equivariance of $\bHo$ and $d\bh$ is manifest. 
\end{proof}

In sum, YM theory on a null boundary is a locally Hamiltonian gauge theory with an equivariant flux map and therefore complies with the symplectic reduction by stages framework summarised in Section \ref{sec:theoreticalframework}.

\section{Superselection in null Yang--Mills theory}\label{sec:YMSuperselectionShort}

We now investigate the superselection structure of YM theory on a null boundary.
Recall, superselection sectors are labelled by the coadjoint orbits $\mathcal{O}_f$ of the on-shell fluxes $f \in \F \doteq \Im(\iota_\C^*h)$.
To understand what these are in null YM theory we need first to have a better grasp on the shell condition---i.e.\ of the constraint surface $\C$.

The Gauss constraint \eqref{eq:gaussdef},
\begin{equation}
    \label{eq:Gauss-explicit}
\mathsf{G}(A,E)\doteq \cL_\ell \E + \cD^i F_{\ell i} = 0,
\end{equation}
can be viewed as a parallel transport equation for $E(u,x)$ along the null direction $\ell$, i.e.\ as a linear first-order evolution equation (ODE) for $\E(u,x)$ in the retarded time $u$ (Remark \ref{rmk:u-functions}). Therefore, the Gauss constraint admits a solution $\E(A,E_\i)$ fully and uniquely determined by the value of $A$ over $\Sigma=S\times I$ as well as the ``initial'' value of $\E$ at $u=-1$ (Lemma \ref{lemma:paralleltr}).

We summarise this discussion in the following proposition. First, however, we introduce some notation:
\begin{remark}\label{rmk:dual}
The densitised dual of $\fg^S = C^\infty(S,\fg)$ is the space of local, $C^\infty(S)$-linear maps from $C^\infty(S,\fg)$ to $\mathbb{R}$. Integration over $S$ yields the isomorphism: $(\fg^S)^* \simeq \Omega^\mathrm{top}(S,\fg^*)$.
Note that $\fg^S$ is isomorphic to the pullback of the elements of $\fG$  to either one of the two boundary components $S_{\i/\f}\subset\pp\Sigma$. 
When it comes to the initial values of the electric field $\iota^*_{S_{\i}}\boldsymbol{E}\in\Omega^{\mathrm{top}}(S,\fg^*)$ and their identification with functions $E_\i\in C^\infty(S,\fg)$, $\iota_{S_\i}^*\boldsymbol{E} \equiv \tr(E_\i\cdot)\vol_S$, here and below we will most often leave the following isomorphism \emph{implicit}:
\[
\fg^S \ \xrightarrow{\simeq}\ \Omega^{\mathrm{top}}(S_\i,\fg^*) \simeq (\fg^S)^*, \qquad E_\i\mapsto \int_S \btr(E_\i\cdot) = \int_{S_{\i}} \iota^*_{S_{\i}}\boldsymbol{E}(\cdot)
\]
and thus often abuse notation and write $E_\i\in(\fg^S)^*$. This will be useful in late sections.
\end{remark}

\begin{proposition}[Constraint surface]\label{prop:constr-surf}
The constraint surface $\C\subset \X$ is a smooth connected submanifold.
Moreover, the map $s_\i$, defined as follows, is a diffeomorphism:
\[
\mathcal{A}\times (\fg^S)^* \xrightarrow{\,s_\i\,} \C\xhookrightarrow{} \X =\mathcal{A}\times\mathcal{E}, \qquad (A,{E_\i})\stackrel{s_\i}{\longmapsto} (A, E(A,E_\i))
\]
where $E\equiv E(A,E_\i) \in \mathcal{E}$ is the unique solution to the Gauss constraint with initial condition $E(u=-1)={E_\i}\in C^\infty(S,\fg^*)\simeq(\fg^S)^*$,
\[\label{eq:Gauss-in-cond}
\begin{cases}
\cL_\ell E + \cD^i F_{\ell i} = 0,\\
E(u=-1) = {E_\i}.
\end{cases}
\]
In the Abelian case, the solution $E(A,E_\i)$ can be written explicitly as
\begin{equation}
    \label{eq:Abeliangauss}
    E(A,E_\i)(u,x) \stackrel{\mathrm{(Ab.)}}{=}  {E_\i}(x) - \int_{-1}^u du'\, D^iF_{\ell i}(u',x).
\end{equation}
\noindent Henceforth we will keep the map $s_\i$ implicit and simply write $(A,E_\i)\in\C$.
\end{proposition}
\begin{proof}
Follows from Lemma \ref{lemma:paralleltr}.
\end{proof}

Now that we have characterised the constraint surface, we can address the on-shell fluxes after introducing some notation.

\begin{notation}\label{not:subsupscripts-2}
We use $\bullet^{\i/\f}$ (and $\bullet^{\av/\dif}$) to denote maps from (spatial) objects defined over $\Sigma$ to objects defined over $S$, as per Notation \ref{not:subsupscripts}; instead, we use $\bullet_{\i/\f}$ as mere labels for objects intrinsically defined on $S_{\i/\f}\subset\pp\Sigma$. This subscript/superscript notation allows us to formally keep track of the nature of the various quantities, but in practice one can simply ignore the distinction.
\end{notation}

\begin{remark}
We observe that
\begin{equation}\label{e:cornergaugealgebra}
\fGp\simeq \fg^S \times \fg^S,
\end{equation}
where each copy of $\fg^S$ corresponds to the mapping algebra on $S_{\i/\f} \subset \pp\Sigma$, respectively. Accordingly, we will often write (cf.\ Notation \ref{not:subsupscripts-2})
\[
(\xi_\i,\xi_\f) \in  \fGp.
\]
The Gauss constraint boundary condition of Proposition \ref{prop:constr-surf} can be written as $E^\i={E_\i}$, while the restriction map $\iota_{\pp\Sigma}^* :\fg^\Sigma \to \fg^{\pp\Sigma}$ reads $(\xi_\i,\xi_\f) = (\xi^\i,\xi^\f)$. 

Moreover, using Equation \eqref{eq:fluxformYM}, the flux map $h:\X\to\fG^*_\loc$ can be written as
\begin{equation}
    \label{eq:fluxmap}
\langle h(A,E),\xi\rangle = - \int_S \btr(  E^\f \xi^\f - E^\i \xi^\i)  =- \int_S \btr(E^\av \xi^\dif + E^\dif\xi^\av);
\end{equation}
similarly, in the Abelian case, this notation, allows us to rewrite  Equation \eqref{eq:Abeliangauss} for the on-shell difference of initial and final electric fluxes in terms of the zero-mode of $F_{\ell i}$:
\begin{equation}
    \label{eq:EdiffAbelian}
E^\dif = - D^i F_{\ell i}^{\int} \qquad \mathrm{(Abelian)}.
\end{equation}
\end{remark}

In the following lemma---and throughout the rest of the paper---we will use the notion, and notation, for the annihilators $\Ann(\mathcal{X},\mathcal{Y}) \subset \mathcal{W}^*_\mathrm{str}$ introduced in Definition \ref{def:annihilators}, with $\mathcal{W} = \fg^N$, $\mathcal{X}=\fg\hookrightarrow\fg^{N}$, and $\mathcal{Y}=(\fg^{N})^*$ (the densitised dual) or $\mathcal{Y}=\mathcal{W}^*_\text{str} = (\fg^{N})^*_{\mathrm{str}}$ (the strong dual). For example, we have:
\begin{align}
\Ann(\fg,(\fg^{\pp\Sigma})^*) 
&\doteq \{ f\in(\fGp)^*\,:\, \langle f,\chi\rangle = 0\ \forall \chi \in \fg\}\notag\\
&\simeq \{ (f_\i,f_\f) \in (\fg^S)^*\times(\fg^S)^* \, : \, \langle f_\i,\chi\rangle = \langle f_\f, \chi\rangle \ \forall \chi \in \fg\}.\label{e:Annpp}
\end{align}

We also note that we have the following natural embedding of the \emph{densitised} dual $(\fg^{\pp\Sigma})^*$ into the \emph{strong} dual $(\fg^\Sigma)^*_\mathrm{str}$ (notice the change in the domain, from $\pp\Sigma$ to $\Sigma$):
\[
C^\infty(\pp\Sigma,\fg) \xrightarrow{\simeq} (\fg^{\pp\Sigma})^*\hookrightarrow (\fg^\Sigma)^*_{\mathrm{str}},
\quad
\eta \mapsto \int_{\pp\Sigma} \btr(\eta \cdot) \mapsto \int_\Sigma d\tr(\bar\eta \cdot) \wedge \vol_S
\]
where in the rightmost term $\bar\eta$ is \emph{any} element of $C^\infty(I,S)\simeq \fg^\Sigma$ such that $\bar\eta\vert_{\pp\Sigma} =\eta$ (e.g. one that vanishes outside of a tubular neighborhood of $\pp\Sigma$).

\begin{lemma} \label{lem:annihilators}
Let $(\fg^{\pp\Sigma})^*\hookrightarrow (\fg^\Sigma)^*_{\mathrm{str}}$. Then,
    \[
    \Ann((\fg^{\pp\Sigma})^*, \fg^\Sigma) \simeq  \fg^\Sigma_\rel, \qquad\mathrm{and}\qquad \Ann(\fg^\Sigma_\rel, (\fg^\Sigma)^*_{\mathrm{str}}) \simeq (\fg^{\pp\Sigma})^*.
    \]
\end{lemma}
\begin{proof}
Looking for elements in $\fg^\Sigma$ that annihilates the image of the embedding $(\fg^{\pp\Sigma})^*\hookrightarrow (\fg^\Sigma)^*_{\mathrm{str}}$ means looking at 
\begin{multline*}
\Ann((\fg^{\pp\Sigma})^*, \fg^\Sigma) \\ \doteq \big\{\xi \in \fg^\Sigma \, : \, \langle \eta,\xi\rangle = \textstyle{\int}_\Sigma d\tr(\bar\eta \xi)\wedge\vol_S  = \textstyle{\int}_{\pp \Sigma} \btr(\eta \xi\vert_{\pp\Sigma})= 0\ \forall \eta\in C^\infty(\pp\Sigma,\fg^*) \big\},
\end{multline*}
which is precisely given by $\fg^\Sigma_\rel=\{ \xi\in\fg^\Sigma\ : \ \xi\vert_{\pp\Sigma}=0\}$. 
The second isomorphism follows from Lemma \ref{lem:doubleannihilator} which implies  $\Ann(\Ann((\fg^{\pp\Sigma})^*, \fg^\Sigma),(\fg^\Sigma)^*_{\mathrm{str}})=(\fg^{\pp\Sigma})^*$ since $(\fg^{\pp\Sigma})^*$ is a vector subspace of the nuclear space $(\fg^{\pp\Sigma})^*_{\mathrm{str}}$.
\end{proof}

We can now characterise the flux space:

\begin{proposition}[On-shell fluxes]\label{prop:fluxes}
Let $h:\X\to\fG^*_\loc$ be the flux map $h=\int_\Sigma d\bh$, and $\F \doteq \Im(\iota_\C^*h)\subset \fG^*_\loc$ the space of on-shell fluxes.
Then, with reference to Equations (\ref{e:cornergaugealgebra}--\ref{e:Annpp}),
\[
\F \simeq \begin{dcases}
    \Ann(\fg^\Sigma_\rel,(\fg^\Sigma)^*_\mathrm{str})\simeq (\fGp)^* & \text{$G$ semisimple},\\
    \Ann(\fg+\fg^\Sigma_\rel,(\fg^{\Sigma})^*_\mathrm{str})\simeq\Ann(\fg,(\fg^{\pp\Sigma})^*) & \text{$G$ Abelian}.
\end{dcases}
\]
Finally, the diffeomorphism
\[
(\fg^{\pp\Sigma})^*\to (\fg^S)^*\times(\fg^S)^*, \quad \ f \mapsto (f_\i,f_\dif) = (f_\i, f_\f-f_\i)
\] 
identifies $\Ann(\fg,(\fg^{\pp\Sigma})^*) \simeq (\fg^S)^*\times \Ann(\fg, (\fg^S)^*)$. 
\end{proposition}

\begin{proof}
Following Remark \ref{rmk:dual}, the maps
\[
e \to \tr(e\cdot)\vol_S \to \int_S \tr(e\cdot)\vol_S \equiv \int_S \btr(e \cdot)
\]
represent the chain of diffeomorphisms $C^\infty(S,\fg)\simeq\Omega^{\mathrm{top}}(S,\fg^*)\simeq(\fg^S)^*$. We will seamlessly switch between these three spaces. Moreover, the diffeomorphism $s_\i: \mathcal{A}\times (\fg^S)^* \to \C$, left implicit, allows us to  write $(A,E_\i)\in\C$ (Proposition \ref{prop:constr-surf}). 
  
(\emph{i}) \emph{$G$ semisimple} -- 
Restricting to $\C$ the expression for $h$ of Equation \eqref{eq:fluxmap} (and implicitly precomposing with $s_\i$), it is sufficient to show that the map 
\[
\C\simeq\mathcal{A}\times (\fg^S)^* \to (\fg^S)^*\times (\fg^S)^*,
\qquad (A,{E_\i}) \mapsto ({E_\i}, E(A,E_\i)^\f)
\]
is surjective, since:
\[
\langle \iota^*_\C h (A,E_\i)),\xi\rangle = - \int_S \btr\big(  E(A,E_\i)^\f \xi^\f - E_\i \xi^\i\big)
\]
is equivalent to $\iota^*_\C h(A,E_\i) = ({E_\i}, -E(A,E_\i)^\f)\in (\fg^{\pp\Sigma})^*$. This shows, in particular, that $\F\hookrightarrow (\fg^{\pp\Sigma})^*$.

Surjectivity of $(A,E_\i) \to (E_\i, - E(A,E_\i)^\f)$ is the statement that for every pair $({E_\i},-E_\f)\in(\fg^S)^*\times (\fg^S)^*$ there exists an $A\in\mathcal{A}$ such that $E(A,E_\i)^\f = E_\f$. We are now going to prove this statement constructively.

Consider the subset of $\C$ given by configurations $(A=A_\ell + \hat{A},E(A,E_\i))$ with $A_\ell = 0$ and $\hat A = a + u b$, where $a$ and $b$ are constant in $u$, i.e.\ $a,b\in\Omega^1(S,\fg)$. Denoting\footnote{The definition of $\cD_a^\dagger$ is equivalent to $\cD_a^\dagger = \star_S \cD_a \star_S$, where $\star_S$ is the Hodge operator on $S$, with respect to the induced metric $\gamma$.}
\[
\begin{array}{rl}
\cD_a:&C^\infty(S,\fg) \to \Omega^1(S,\fg)\\
&\quad\quad\quad\;\, \xi\mapsto  D\xi+[a,\xi] 
\end{array}
\qquad\mathrm{and}\qquad
\begin{array}{rl}
\cD_a^\dagger: &\Omega^1(S,\fg)\to C^\infty(S,\fg)\\
&\quad\quad\quad b \mapsto D^i b_i + [a^i,b_i]
\end{array}
\]
for these configurations, $\cL_\ell E = \pp_u E$, $F_\ell =b$, and $\cD^iF_{\ell i} = \cD_a^\dagger b$, whence the Gauss constraint \eqref{eq:Gauss-in-cond} becomes
\[
\pp_u E + \cD_a^\dagger b = 0, \qquad E^\i = {E_\i}.
\]
and, being $a$ and $b$ $u$-independent,
\[
E(a + u b,E_\i)^\f = {E_\i} - 2  \cD_a^\dagger b.
\]
Therefore, the statement follows if we can prove that the following map is surjective:
\[
\Omega^1(S,\fg)\times\Omega^1(S,\fg) \to C^\infty(S,\fg), \qquad
(a,b)\mapsto  \cD_a^\dagger b.
\]
i.e.\ that for any $e\in C^\infty(S,\fg)$ there exists a pair $(a,b)$ such that $e=\cD_a^\dagger b$. We will now show that for any $a$ irreducible, there exists a (unique) $b$ for which this is true.
An element $a\in\Omega^1(S,\fg)$ is said irreducible iff $\cD_a$ has trivial kernel; if $\fg$ is semisimple, irreducible elements exist (in fact, they are dense in $\Omega^1(S,\fg)$). The key point is that, if $a$ is irreducible, then the covariant Laplacian $\Delta_a \doteq \cD_a^\dagger \circ \cD_a$ on $S$ is an elliptic operator with trivial kernel. Then, by the Fredholm alternative theorem, $\Delta_a$ is invertible and we can thus define $\eta=\eta(e)$ as the unique solution to the equation
\[
\Delta_a \eta = e.
\]
Therefore, for any $e\in C^\infty(S,\fg)$, we can construct a preimage $(a,b) = (a, \cD_a \eta(e))\in \Omega^1(S,\fg)\times\Omega^1(S,\fg) $ thus proving the statement.

(\emph{ii}) \emph{$G$ Abelian} -- This case can be addressed along similar lines.
Using Equations (\ref{eq:fluxmap}--\ref{eq:EdiffAbelian}), we express $\iota_\C^*h$ as
\begin{equation}
    \label{eq:fluxproof}
\langle \iota^*_\C h(A,{E_\i}),\xi\rangle = \int_S \btr\big( {E_\i} (\xi^\i-\xi^\f)  + (D^iF_{\ell i}^{\int}) \xi^\f \big),
\end{equation}
where we recall the notation $Q^{\int} \doteq \int_{-1}^1 du'\, Q(u')$.
From this formula it is readily clear that given $\chi\in \fg\hookrightarrow \fg^\Sigma$ we have $\langle \iota^*_\C h,\chi\rangle = 0$, that is to say $\F\subset \Ann(\fg,(\fg^\Sigma)^*_{\mathrm{str}})$. 
Moreover, we also observe that for any $\xi_\circ \in \fg^\Sigma_\rel$ we have $\langle \iota^*_\C h,\xi_\circ\rangle = 0$, that is $\Ann(\fg^\Sigma_\rel,(\fg^\Sigma)^*_{\mathrm{str}}) $. Then, using Lemma \ref{lem:annihilators}, we obtain:
\[
\F\subset \Ann(\fg,(\fg^\Sigma)^*_{\mathrm{str}}) \cap \Ann(\fg^\Sigma_\rel,(\fg^\Sigma)^*_{\mathrm{str}}) \simeq \Ann(\fg,(\fg^\Sigma)^*_{\mathrm{str}}) \cap (\fg^{\pp\Sigma})^*.
\]

To prove the opposite inclusion, we first consider the isomorphism of vector spaces
\[
(\fGp)^*\simeq (\fg^S)^*\times(\fg^S)^* \xrightarrow{\simeq} (\fg^S)^*\times(\fg^S)^*, \quad (f_\i,f_\f)\mapsto (f_\i , f_\dif) \doteq (f_\i, f_\f - f_\i).
\]
A moment of reflection shows that, in the light of the identity
\[
\langle f_\f, \xi_\f\rangle - \langle f_\i,\xi_\i \rangle = \langle f_\i,(\xi_\f-\xi_\i)\rangle + \langle f_\dif, \xi_\f \rangle,
\]
we have
\[
\Ann(\fg,(\fg^\Sigma)^*_{\mathrm{str}})\cap (\fg^{\pp\Sigma})^* \simeq (\fg^S)^* \times \Ann(\fg,(\fg^S)^*) \ni (f_\i,f_\dif).
\]
(Observe that this equation proves the last statement of the proposition.)
Therefore, we can equivalently prove that
\[
(\fg^S)^* \times \Ann(\fg,(\fg^S)^*) \subset \F.
\]
Using the (vector space) identification $(\fg^S)^*\simeq \fg^S$, we can now identify $f \in (\fg^S)^* \times \Ann(\fg,(\fg^S)^*)$ with the pair $(\eta_{\i},\eta_{\dif})\in \fg^{S}\times\fg^S$ such that $\int_S \btr(\eta_\dif \chi) \vol_S = 0$ for all $\chi\in\fg\hookrightarrow\fg^S$. 

In light of this identification as well as Equation \eqref{eq:fluxproof}, we find that $(\fg^S)^* \times \Ann(\fg,(\fg^S)^*) \subset \F$ iff for $(\eta_\i,\eta_\dif)$ as above there exists a $(A,E_\i)$ such that
\[
E_\i = \eta_\i
\quad\text{and}\quad
D^i F_{\ell i}^{\int} = \eta_\dif.
\]
The first condition is immediate, the second requires us to find at least one $A\in\mathcal{A}$ that satisfies it.

We thus look for such an $A$ among the connections of the form $A_\ell=0$ and $\hat A= u b$, with $b\in\Omega^1(S,\fg)$. Then, $F_\ell = b$ and $D^i F_{\ell i}^{\int} = 2 D^ib_i$.
Denoting $\{\tau_\alpha\}$ a basis of $\fg$ and $\eta_\dif^\alpha = \tr(\eta^\alpha_\dif\tau_\alpha)$, Hodge theory \cite{SchwarzHodgeBook} then tell us that the equation
\[
\Delta \lambda^\alpha = \eta_\dif^\alpha
\]
has a solution $\lambda^\alpha=\lambda^\alpha(\eta_\dif)$ (unique up to the addition of harmonic functions, i.e.\ constants if $S$ is a sphere) iff $\eta_\dif^\alpha$ integrates to zero---i.e.\ iff $\eta_\dif$ corresponds to an element of $\Ann(\fg,(\fg^S)^*)\subset(\fg^S)^*\simeq \fg^S$. Therefore, setting $b=\tfrac12 D \lambda(\eta_\dif)$ we conclude the proof.
\end{proof}
    
Having characterised the shell $\C$ and the space of on-shell fluxes $\F$ (Propositions \ref{prop:constr-surf} and \ref{prop:fluxes}), we apply Theorem \ref{thm:generalsuperselection} of the review Section \ref{sec:theoreticalframework} (cf. \cite[Theorem 1]{RielloSchiavina}) to obtain the superselection structure of null YM theory:

\begin{theorem}[Superselection of null YM theory]\label{thm:superselections}
The fully-reduced phase space of Yang--Mills theory on the null manifold $\Sigma\simeq I \times S$, defined as the space of on-shell configurations modulo \emph{all} gauge transformations
\[
{\uuC}\doteq\C/\G = \bigsqcup_{\mathcal{O}_f\subset \F} \uuS_{[f]},
\]
is a Poisson manifold. Explicitly, the superselection sectors $\uuS_{[f]}$ correspond to:  
\begin{enumerate}[label=(\roman*)]
    \item if $G$ is semisimple, $\G$-equivalence classes of those on-shell configurations whose (electric) flux belongs to the same coadjoint orbit in $(\fGp)^*$, i.e.\ pairs $(A,E)\in\C \hookrightarrow\X$ such that
    \[
    \Big(\textstyle{\int}_S \btr(E^\i \cdot), -\textstyle{\int}_S \btr(E^\f \cdot) \Big)\in(\mathcal{O}_{f_\i},\mathcal{O}_{f_\f})\subset (\fg^S)^*\times(\fg^S)^*
    \]
    \item if $G$ is Abelian, $\G$-equivalence classes of those on-shell configurations with the same value of the (electric) flux in $\F\simeq \Ann(\fg,(\fg^{\pp\Sigma})^*)\subset(\fGp)^*$, i.e.\ pairs $(A,E)\in\C \hookrightarrow\X$ such that \
\[
\Big(\textstyle{\int}_S\btr(E^\i\cdot), -\textstyle{\int}_S\btr(E^\dif\cdot) \Big) = (f_\i,f_\dif) \in (\fg^S)^* \times \Ann(\fg,(\fg^S)^*)
\]
where ${E}^\dif = {E}^\f - {E}^\i= - D^i F_{\ell i}^{\int}$ (Equation \eqref{eq:EdiffAbelian}).
\end{enumerate}
\end{theorem}

In the next two sections we are going to explicitly compute the symplectic structure on the superselection sectors $(\uuS_{[f]},\uuomegao_{[f]})$, and compare them to the symplectic structure on the constraint-reduced phase space $(\uCo,\uomegao)$ and on the AS phase space $(\hcalA,\omAS )$---spoiler: they all differ. 

In order to attain an explicit characterisation of the physical d.o.f. of YM theory on a null surface we will have to choose a gauge fixing and thus introduce a level of arbitrariness in the description of said degrees of freedom.

\section{Symplectic reduction of null YM theory: first stage}\label{sec:reduction-first}

\noindent\fbox{%
    \parbox{\textwidth}{\small\center
        The reader interested in applications to soft symmetries can skip this section at first.
    }%
}

\medskip

First stage, or constraint reduction is about enforcing the Gauss constraint $\mathsf{G} = 0$ and quotienting out the action of the constraint gauge group $\Go$  that it generates. As discussed in Section \ref{sec:theoreticalframework}, the constraint gauge group $\Go = \exp\fGo$ is infinitesimally generated by the constraint gauge ideal $\fGo = \Ann(\F,\fg^\Sigma)\equiv \Ann(\Im(\iota_\C^*h),\fg^\Sigma)$. Diagrammatically:

\[
\xymatrix@C=.75cm{
(\X,\omega)
    \ar@{~>}[rr]^-{\tbox{2.2cm}{constraint reduction \\(w.r.t. $\Go$ at $0$)}}
&&(\uCo,\uomegao)\\
&{
    \;\C\;
    \ar@{_(->}[ul]^-{\iota_\C}
    \ar@{->}[ur]_-{\pi_\circ}
    }
}
\]

The goal of this section is to provide an \emph{explicit} description of the constraint-reduced phase space $(\uCo,\uomegao)$,
\[
\uCo \doteq \C/\Go, \qquad \pi_\circ^*\uomegao = \iota_\C^*\omega,
\]
in terms of the following extensions of the AS phase space $(\hcalA, \omAS )$ defined in Definition \ref{def:hatA-sympl}: 
\begin{definition}[Extended Ashtekar--Streubel phase space]\label{def:extAS}
Denote by $\Omega_S$ the canonical symplectic structure on $T^*G^S_0$, and by $\omega_S$ that on $T^*\fg^S$. The \emph{extended Ashtekar--Streubel phase space} is the symplectic manifold
\[
(\XAS,\omeAS)\doteq(\hcalA \times T^*G^S_0 , \omAS + \Omega_S),
\]
while the \emph{linearly}-extended Ashtekar--Streubel phase space is the symplectic manifold
\[
(\Xas,\omeas)\doteq(\hcalA \times T^*\fg^S , \omAS + \omega_S).
\qedhere
\]
\end{definition}
~

\begin{theorem}[Constraint reduction]\label{thm:nonAb-constr-red}
If $\G\circlearrowright\X$ is proper, the constraint-reduced phase space $(\uCo,\uomegao)$ is a smooth, connected symplectic covering space of the \emph{extended Ashtekar--Streubel phase space} $(\XAS,\omeAS)$.
In particular, these two spaces are \emph{locally} symplectomorphic:
\begin{equation}\label{e:1stagelocalsymplectomorphism}
(\uCo,\uomegao) \simeq_\loc (\XAS,\omeAS).
\end{equation}    
Moreover, if $G$ is Abelian, the constraint-reduced phase space is \emph{globally} symplectomorphic to the linearly-extended Ashtekar--Streubel phase space:
\begin{equation}\label{e:AbelianASreduction}
(\uCo,\uomegao) \simeq (\Xas,\omeas)\simeq_\loc (\XAS,\omeAS) \qquad \mathrm{\emph{(}Abelian\emph{)}}.\qedhere
\end{equation}
\end{theorem}
~

\begin{remark}\label{rmk:uCosmoothness}
Theorem \ref{thm:nonAb-constr-red} concludes that $\uCo=\C/\Go$ is connected and smooth, and that the reduction $\uCo$ is symplectic.  
\emph{Connectedness} follows from the following simple observation: since $\C\simeq \mathcal{A}\times (\fg^S)^*$ is affine it is also connected, which then implies $\uCo \doteq \C/\Go$ is connected as well.
As for \emph{smoothness}: in YM theory the set $\C$ is a smooth submanifold in virtue of Proposition \ref{prop:constr-surf}, but the smoothness of the quotient  depends on whether the action of $\Go$ on $\C$ is free and proper. Properness is generally granted when $G$ is compact, since we are assuming that $\Sigma$ is also compact \cite{NarasimhanRamadas79, Rudolph_2002}. For $G$ semisimple, the action of $\Go$ can be proven to be free, since the equation $d_A\xi = 0$ has no (non-zero) solutions for $\xi\in \fg^\Sigma_\rel$ and $\Go\simeq G^\Sigma_{\rel,0}=\langle \exp\fg^\Sigma_\rel\rangle$, a fact proved in Proposition \ref{prop:discretequotientgroup}. For $G$ Abelian, a similar result holds up to a \emph{global} stabiliser (the kernel of the action map), which does not affect the smoothness of the reduced space $\uCo$. 
Finally, the \emph{symplectic} nature of $\uCo$ follows from \cite[Theorem 1]{RielloSchiavina}, provided one checks that the image of $\fGo$ under the action map $\rho$ is symplectically closed. To prove this abstractly, one could adapt the argument in \cite{DiezHuebschmann-YMred} (see also \cite[Section 5.2]{RielloSchiavina}). Here, instead, we explicitly computed the constraint-reduced symplectic form and showed it is nondegenerate.
\end{remark}

\begin{remark}[$\Xas$ vs. $\XAS$]\label{rmk:localsymp}
If $G$ is Abelian, the cotangent bundle $T^*\fg^S$ is \emph{locally} symplectomorphic to $T^*G^S_0$, i.e.\ $T^*\fg^S \simeq_\loc T^*G^S_0$, through the map $\exp : \fg^S\to G_0^S$. The obstruction to a global extension is directly related to the fact that the exponential map $\fg \to G$ is itself only a \emph{local} diffeomorphism in general. In fact, if $G$ contains a $\mathrm{U}(1)$ factor, this map is a local diffeomorphism, but does not possess a global inverse (it is many-to-one).
Therefore, although in the Abelian case $\uCo$ is \emph{globally} symplectomorphic to $\Xas \doteq \hcalA \times T^*\fg^S$, it is only \emph{locally} symplectomorphic to $\XAS \doteq \hcalA\times T^*G^S_0$. 
\end{remark}

\begin{theorem}[characterisation of the fibre]\label{thm:fibrecharacterisation}
The covering fibre of 
\[
p_V\colon \uCo \to \XAS
\]
is the (discrete) group of components of $\G_\rel$, i.e.
\[
\mathcal{K} \doteq \G_\rel/(\G_\rel)_0.
\]
If, moreover, $S$ is diffeomorphic to the sphere $S^{n-1}$ (with $S^0 \doteq \{\mathrm{pt}\}$), then\footnote{Recall, $G$ is assumed finite dimensional and connected throughout the article.} 
\begin{enumerate}[label=(\roman*)]
\item if $G$ is simply connected and $n=1,2$ then $\mathcal{K}$ is trivial and the symplectomorphism $(\uCo, \uomegao) \simeq (\XAS,\omeAS)$ is global;
\item if $G$ is simply connected and $n=3$ then $\mathcal{K} \simeq \pi_3(G) \simeq \mathbb{Z}^s$ for some $s\in\mathbb{N}$;\footnote{$s>0$ if $G$ is semisimple, and $s=1$ if $G$ is simple.}
\item if $G\simeq U(1)^t\times \mathbb{R}^k$ (Abelian) and $n\neq2$ then $\mathcal{K}\simeq \mathbb{Z}^t$;
\item in general: $\mathcal{K}$ is a subset of 
\[
\pi_0(G^\Sigma_\rel) \simeq 
\begin{dcases}
\pi_1(G) & {\mathrm{if}\ n=1,2}\\
\pi_1(G)\oplus\pi_n(G) & {\mathrm{if}\ n>2}
\end{dcases}
\]
and coincides with it if $\pi_0(G^\Sigma)\simeq \pi_{n-1}(G)$ is trivial.
\end{enumerate}
\end{theorem}

\begin{remark}[Spacetime picture]\label{rmk:spacetime}
Recall that, from a spacetime perspective, $\Sigma$ is a codimension-1 hypersurface. Therefore, since $n \doteq \dim \Sigma$, the cases $n=1,2,3$ analysed in Theorem \ref{thm:fibrecharacterisation} correspond respectively to spacetime regions of dimensions $2,3$ and $4$ with spherical boundary surfaces. In particular, the case $n=3$ of the Corollary is the one relevant for the study of classical YM theory at asymptotic null infinity discussed in Section \ref{sec:alternatelabelmemory}.
\end{remark}

\begin{remark}[Winding number]
In the Abelian case, the isomorphism $\mathcal{K}\simeq \mathbb{Z}^t$ holds \emph{as groups}, in virtue of the fact that one can define an additive winding number, $w(g_1g_2)=w(g_1)+w(g_s)$, which provides an irreducible representation of the group of components $\mathcal{K}$, see Remark \ref{rmk:grphomo}. Moreover, an analogous result can be obtain when $n=3$ using the Wess--Zumino winding number computed on the smash product $S^3\simeq S^2 \wedge S^1$ that appears in the proof of point (ii) of the theorem.
\end{remark}

\subsection{Proof of Theorem \ref{thm:nonAb-constr-red}}
This section is devoted to the proof of the local symplectomorphism between $\uCo$ and $\XAS$, claimed to exist in Theorem \ref{thm:nonAb-constr-red}. Loosely, our proof relies on the ``dressing-field method'' associated to the ``gauge condition'' $A_\ell = 0$. The claim of global symplectomorphism between $\uCo$ and $\Xas$ will be proved in Appendix \ref{app:Abelianreduction} using a ``linear" version of the dressing field method which is available in the Abelian case only.

(For an algebraic and geometrical account of the dressing field method see \cite{Francois2012, RielloGomesHopfmuller}, whereas some of its applications to the symplectic structure of gauge theories can be found in \cite{RielloGomesHopfmuller,RielloGomes, CarrozzaHoehn, Francois2021}.)

Before addressing reduction, we must characterise the constraint gauge ideal $\fGo \subset \fG$, and two closely related groups.

\begin{proposition}[Constraint gauge transformations]\label{prop:discretequotientgroup}~

\noindent For $G$ semisimple,
\begin{enumerate}[label=\emph{(1.\roman*)}]
\item\label{1.i} the constraint gauge ideal is
\[
  \fGo \doteq \Ann(\F,\fg^\Sigma) \simeq \fg^\Sigma_\rel;
\]
\item\label{1.ii} the constraint gauge group $\Go \doteq\langle\exp \fGo\rangle$ is the identity component of the relative gauge group $\tGo$ which in turn equals the identity component of the relative mapping group $G^\Sigma_\rel$:
\[
\Go  \simeq (\tGo)_0 \simeq G^\Sigma_{\rel, 0};
\]
\item\label{1.iii} $\mathcal{K}\doteq \tGo/\Go$ is a discrete group, the  \emph{group of components} of $\G_\rel$.
\end{enumerate}
\smallskip

\noindent  For $G$ Abelian, let $\fg\hookrightarrow\fg^\Sigma$ be the Lie ideal of constant gauge transformations, and $G\hookrightarrow G^\Sigma_0 \hookrightarrow G^\Sigma$ the normal subgroup of constant gauge transformations; then
\begin{enumerate}[label=\emph{(2.\roman*)}]
    \item\label{2.i} the constraint gauge ideal is given by gauge transformations over $\Sigma$ whose restriction to $\pp\Sigma$ is constant:\footnote{$\fg + \fg^\Sigma_\rel  = \{  \xi\in \fg^\Sigma \ : \ \exists \chi\in \fg\text{ such that } \xi|_{\pp\Sigma}=\chi\}$.}
    \[
    \fGo \doteq \Ann(\F,\fg^\Sigma) \simeq \fg + \fg^\Sigma_\rel  \qquad \mathrm{\emph{(}Abelian\emph{)}};
    \]
    \item\label{2.ii} the constraint gauge group $\Go \doteq\langle\exp \fGo\rangle$ is the identity component of the group of gauge transformations over $\Sigma$ whose restriction to $\pp\Sigma$ is constant:\footnote{$G\cdot G^\Sigma_\rel = \{  g\in G^\Sigma \ : \ \exists k\in G\text{ such that } g|_{\pp\Sigma}=k\}$.}
    \[
    \Go \simeq (G\cdot \G_\rel)_0 \simeq G\cdot G^\Sigma_{\rel,0} \qquad \mathrm{\emph{(}Abelian\emph{)}};
    \]
    \item\label{2.iii} $\mathcal{K}\doteq (G\cdot \G_\rel)/\Go$ is a discrete group, the  \emph{group of components} of $G\cdot\G_\rel$.
\end{enumerate}
\end{proposition}

\begin{proof}
We start from the case of $G$ semisimple.

\ref{1.i} The characterisation of the constraint gauge ideal $\fGo \doteq \Ann(\F,\fg^\Sigma)$ follows from $\F\simeq (\fg^{\pp\Sigma})^*$ (Proposition \ref{prop:fluxes}) and $\Ann((\fg^{\pp\Sigma})^*,\fg^\Sigma)\simeq\fg^\Sigma_\rel$ (Lemma \ref{lem:annihilators}).

\ref{1.ii} Notice that since $\fGo \subset \fG = \fg^\Sigma$, we have $\Go \subset \G = G^\Sigma_0$, and that since $\fGo = \fg^\Sigma_\rel$, we have also $\Go \subset G^\Sigma_\rel$. Therefore $\Go \subset G^\Sigma_0 \cap G^\Sigma_\rel \equiv G^\Sigma_{0,\rel} \doteq \G_\rel$. Since $\mathrm{Lie}(G^\Sigma_{0,\rel}) = \fg^\Sigma_\rel = \fGo$, we see that $\Go$ is the identity component of $G^\Sigma_{0,\rel}$, i.e.
\[
\Go \simeq G^\Sigma_{0,\rel,0} \equiv (\G_\rel)_0.
\]
In virtue of the equation above, the second claimed characterisation, $\Go\simeq G^\Sigma_{\rel,0}$, is equivalent to
\[
G^\Sigma_{0,\rel,0} = G^\Sigma_{\rel,0},
\]
which we now prove.
Recall: $G^\Sigma_{0,\rel} \equiv G^\Sigma_0 \cap G^\Sigma_\rel$. One can verify that\footnote{This is because the homotopies that connect elements to the identity in $G^\Sigma_{0,\rel,0}$ live in $G^\Sigma_0\cap G^\Sigma_\rel$ and thus, in particular, in $G^\Sigma_\rel$.} $G^\Sigma_{0,\rel,0} \subset G^\Sigma_{\rel,0}$. To prove the opposite inclusion, we start by observing that, on the one hand, $G^\Sigma_{\rel,0} \subset G^\Sigma_\rel$ and that, on the other, $G^\Sigma_\rel \subset G^\Sigma$ implies $G^\Sigma_{\rel,0} \subset G^\Sigma_0$. Therefore, comparing with the definition of $G^\Sigma_{0,\rel}$, we obtain: $G^\Sigma_{\rel,0} \subset G^\Sigma_{0,\rel}$. But since $G^\Sigma_{\rel,0} $ is connected to the identity as a group, one has a fortiori that it is contained in the identity component of $G^\Sigma_{0,\rel}$. Therefore, as desired, $G^\Sigma_{\rel,0} \subset G^\Sigma_{0,\rel,0}$ as well.

\ref{1.iii} From \ref{1.ii} $\Go = (\G_\rel)_0$ and the identity component is always an \emph{open} normal subgroup. Thus $\mathcal{K}\doteq \G_\rel/\Go = \G_\rel/(\G_\rel)_0$ is a discrete group.

Next we consider the case $G$ Abelian.

\ref{2.i} Follows from Proposition \ref{prop:fluxes} and Lemma \ref{lem:doubleannihilator}, since
\[
\Ann(\F,\fg^\Sigma)=\Ann(\Ann(\fg + \fg^\sigma_\rel, (\fg^\Sigma)^*_{\mathrm{str}}),\fg^\Sigma) = \fg + \fg^\sigma_\rel.
\]

\ref{2.ii} One can easily adapt the first argument of \ref{1.ii} to find $\Go \subset G_0^\Sigma \cap G\cdot G^\Sigma_{\rel}$, and hence 
\[
\Go \simeq (G\cdot G^\Sigma_{0,\rel})_0 \equiv (G\cdot \G_\rel)_0.
\]
Then, since $G$ is connected, $(G\cdot G^\Sigma_{0,\rel})_0 = G\cdot G^\Sigma_{0,\rel,0}$ and therefore we can use the second argument in \ref{1.ii} to conclude.

\ref{2.iii} The proof is the same as in \ref{1.iii}.
\end{proof}

\begin{remark}
Note that the relative subgroup of the identity component of $G^\Sigma$ does not necessarily coincide with the identity component of its relative subgroup, i.e.\ in general $G^\Sigma_{0,\rel} \neq G^\Sigma_{\rel,0}$. Indeed, in the previous proposition we proved that the latter is the identity component of the former.
\end{remark}

\begin{remark}[Abelian isotropy]\label{rmk:Abelian-isotropy}
In the Abelian case, $\fGo$ differs from the relative algebra $\fg^\Sigma_\rel$ by constant gauge transformations. These act trivially on $\X$---and in particular on $\C$---and therefore constitute constitute the (\emph{configuration-independent}) isotropy algebra $\ker(\rho) = \fg \hookrightarrow \fg^\Sigma$ of the Lie algebra action, which in particular is an ideal.
Sometimes the isotropy algebra is called the ``reducibility'' algebra and its elements ``reducibility parameters".
Therefore, albeit $\Go\subsetneq \G_{\rel,0}$, both $\Go$ and $\G_{\rel,0} =\langle \exp \fg^\Sigma_\rel\rangle$ have the same orbits on $\C$:
\[
\uCo \doteq \C / \Go = \C / \G_{\rel,0}.
\]
The distinction between $\fGo$ and $\fg^\Sigma_\rel$ remains nevertheless crucial for its relation to the integrated Gauss's\footnote{This is why, \emph{in the presence of charged matter over $\Sigma$}, one would find $\fGo \simeq \fg^\Sigma_\rel$ even in the Abelian case.} for the Abelian fluxes---and therefore to the superselection sectors (second-stage reduction, Section \ref{sec:reduction-second}).
\end{remark}

After this characterisation of $\Go$ and $\tGo$, we turn our attention to the dressing field method as a tool for ``gauge fixing'' $A\in\mathcal{A}$ to some $\ASa$ such that $\ASa_\ell =0$, i.e.\ to some $\ASa\in\hcalA$. With this goal in mind, we use Lemma \ref{lemma:Wilsonline} to introduce the following object---the terminology for which will be explained in Remark \ref{rmk:dressing-terms} below:

\begin{definition}[Dressing field]\label{def:dressingfield}
The \emph{dressing field} $\dfU : \X\to C^\infty(I,G^S_0)$ is the unique solution to the boundary value problem:
\[
\begin{cases}
 L_\ell \dfU \dfU^{-1} = - A_\ell\\
\dfU^\f = 1.
\end{cases}
\qedhere
\]
\end{definition}
\begin{remark}[Wilson lines]\label{rmk:WilsonLine-V}
$\dfU(A)$ can be thought of as a collection of holonomies, a.k.a. path-ordered exponentials or Wilson lines,\footnote{Note: a priori this takes values in $G^S$, but it lies in fact in $G^S_0$ because $V[tA]$ is a homotopy to the identity. See also Lemma \ref{lemma:dressing-properties}\ref{lemma:dressing-properties-i}.} along the flows of the null vector field $\ell\in\mathfrak{X}^1(\Sigma)$. It is thus convenient to introduce the notation:
\[
\dfU(A)(u,x) = \overrightarrow{\Pexp} \int_{u}^1 du' \, A_\ell(u',x)  , \quad V(A)\in C^\infty(I,G^S_0) \simeq G^\Sigma_0 \simeq \G.
\]
(Here we use an arrow on top of $\Pexp$ to stress that the path-ordering composes left to right; cf. item (\ref{lemma:dressing-properties-iii}) of Lemma \ref{lemma:dressing-properties}.)
\end{remark}

We can use the dressing field to define:

\begin{definition}[Dressing map]\label{def:dressingmap}
The on-shell\footnote{Off-shell extensions are possible, but not needed and less natural.} \emph{dressing map} is
\[
\cU : \C \to \XAS  \doteq  \hcalA \times T^*G^S_0 ,
\qquad
(A,{E_\i}) \mapsto \begin{pmatrix}
    \ASa\, \\ \Lambda \\ \ASe
\end{pmatrix} = 
\begin{pmatrix}
     A^{\dfU(A)} \\ \dfU(A)^\i \\ \Ad(\dfU(A)^\i)^{-1}\cdot {E_\i}
\end{pmatrix},
\]
where, henceforth, we leave the left-invariant trivialisation $T^*G^S_0 \simeq G^S_0 \times (\fg^S)^*$ implicit.
\end{definition}

Note that we used $(\ASa,\Lambda,\ASe)$ to denote variables in $\XAS$, and we think of the components of the dressing map $\cU(A,E_\i) = (\ASa(A),\Lambda(A),\ASe(A,E_\i))$ as ($\Sigma$-nonlocal) functions on $\C$. The dressing map is thus expressed in the chosen variables as, e.g.\ $\ASa=\ASa(A)$ etc.

The following lemma shows a few important properties of the dressing field and map, and in particular proves that the previous definition is well-posed (items (i-ii)):

\begin{lemma}\label{lemma:dressing-properties} 
For all $(A,{E_\i})\in\C$, $g\in\G$, and $g_\circ\in\tGo$:
\begin{enumerate}[label={(\roman*)}]
    \item\label{lemma:dressing-properties-i} $\dfU(A)^\i \in G^S_0$;
    \item\label{lemma:dressing-properties-ii} $(A^{\dfU(A)})_\ell = 0$;
    \item\label{lemma:dressing-properties-iii} $\dfU[A^g] = g^{-1} \dfU(A) g^\f$;
    \item\label{lemma:dressing-properties-iv} $\dfU[A^{g_\circ}] = g^{-1}_\circ\dfU(A)$.
\end{enumerate}
\end{lemma}
\begin{proof}
\begin{enumerate}[label={(\roman*)}, wide=0pt]
\item That $\Lambda(A)$, or equivalently $\Lambda(A)= \dfU(A)^\i$, lies in $G^S_0 \subset \G_S$ can be seen in two different ways, both insightful. Indeed, both $H_1(t) = \dfU(A)(u)\vert_{u=1-2t}=\dfU(A)(1-2t)$ and $H_2(t) = \dfU(tA)^\i$ define smooth homotopies $[0,1]\to G^S$ with $H_{1,2}(0) = 1$ and $H_{1,2}(1) = \Lambda(A)$.
\item  Straightforward: $(A^{\dfU(A)})_\ell = (\dfU(A)^{-1} A_\ell \dfU(A) + \dfU(A)^{-1} L_\ell \dfU(A)) = 0$. 
\item  Standard, see Lemma \ref{lemma:Wilsonline}.
\item It follows from \ref{lemma:dressing-properties-iii} and the fact that $g_\circ^\i = g_\circ^\f =1$.
\qedhere
\end{enumerate}
\end{proof}

\begin{remark}[{Locality properties of $\Lambda(A)$ and $\ASe(A,E_\i)$}]\label{rmk:bilocal}
Recall that
\[
\Lambda(A) \doteq  \dfU(A)^\i = \overrightarrow{\Pexp} \int_{-1}^1 du' \, A_\ell(u')
\]
is a collection of Wilson lines stretching from the initial to the final boundary surfaces, $S_\i$ and $S_\f$, along the flows of $\ell$ on top of each $x\in S$. Under the action of $g\in\G$, these transform as
\begin{equation}\label{e:Vbilocal}
\Lambda(A) \mapsto \Lambda[A^g] = (g^\i)^{-1} \Lambda(A) g^\f.
\end{equation}

Interpreted as Wilson lines stretching from $S_\i$ to $S_\f$, $\Lambda(A)$ defines a nonlocal object over $\Sigma$ which transforms under the action of $\G$ according to the left/right bi-local action formula above.
But this can also be re-interpreted as follows: the map $\Lambda$ on $\mathcal{A}$ can be seen as valued in $G^S = C^\infty(S,G)$, and as such it carries two, commuting, local actions of $\G_S$, one on the left and one on the right. If we parametrise the (abstract) left and right actions using $g^\i$ and $g^\f$ respectively, we can reproduce the bi-local transformation we encountered above.

The relationship between these two viewpoints is enabled by the invariance of $\Lambda(A)$ under the action of $\tGo \supset \Go$, which guarantees that the action of $\G$ on $\Lambda$ descends to a residual boundary action---indeed, $\G/\tGo \simeq G^S_0 \times G^S_0 \subset G^{\pp\Sigma} $, as we will prove in Proposition \ref{prop:Gred}. 

Similarly, the dressed field $\ASe \equiv \ASe(A,E_\i) = \Ad(\dfU(A)^\i)^{-1}\cdot {E_\i}$ is a nonlocal functional that depends both on the initial value of the electric field, $E^\i = {E_\i}$, and on the entire Wilson line $V(A)^\i$. Therefore, the action of $g\in\G$ on $\C$ is mapped to,
\begin{align*}
\ASe(A,E_\i) \mapsto \ASe(A^g,e^g_\i) & = \Ad((g^\i)^{-1}\dfU(A) g^\f )^{-1}\cdot (\Ad(g^\i)^{-1}\cdot {E_\i})\\ & =  \Ad(g^\f)^{-1}\cdot \ASe(A,E_\i).\qedhere
\end{align*}
\end{remark}

\begin{remark}[Dressing fields]\label{rmk:dressing-terms}
Our dressing field $V$ is closely related to the notion of \emph{dressing field} for the action of $\tGo$ as introduced by \cite{Francois2012,Francois2021}. Their definition requires $V$ to be both $\tGo$-equivariant \emph{and $\tGo$-valued}, whereas in our case it is $\G$-valued. 
As a consequence of this mismatch, the dressing map $\cV$ assigns to each gauge connection $A$ two objects: not only its corresponding, desired, purely spatial, ``gauge-fixed'' version w.r.t.\ the action of $\tGo$, denoted $\ASa(A)$---but also the Wilson line $\Lambda(A)$ (cf.\ Remark \ref{rmk:bilocal}). (See also \cite{GomesPhD} for a discussion on the difference between gauge-fixing and dressing, and \cite{RielloGomesHopfmuller} for the relation to field-space connections.) In Lemma \ref{lemma:constr-trivialis-Ab} we use the original notion of a $\Go$-dressing (there called $U_\circ$), and the extra field (there called $\lambda$) is introduced by other means.
\end{remark}

In the next two statements, we apply the dressing field method to show that the dressing map descends to a covering map $p_V\colon \uCo \to \XAS$, which is also a local symplectomorphism.

\begin{proposition}\label{prop:step1}
The dressing map
\[
\cU:\C\to \XAS \doteq \hcalA\times T^*G^S_0 
\] 
is surjective and has $\G_\rel$ fibre.
\end{proposition}

\begin{proof}
Recall that the dressing map has the following components:
\[
\cU : (A,{E_\i}) \mapsto 
\begin{pmatrix}
    \ASa(A)\\
    \Lambda(A)\\
    \ASe(A,E_\i)
\end{pmatrix} = 
\begin{pmatrix} 
A^{\dfU(A)} \\ 
\dfU(A)^\i \\ 
\Ad(\dfU(A)^\i)^{-1}\cdot {E_\i}
\end{pmatrix}.
\]
To prove the proposition, we can suppress suppress the third component of the dressing map $\ASe$ from our notation. This is motivated by the fact that $\ASe = \Ad(\Lambda^{-1})\cdot {E_\i}$, the adjoint map $\Ad_\Lambda : \fg^S \to \fg^S$ is a diffeomorphism, and $\tGo$ acts trivially on ${E_\i}$. In the following we will thus focus on the components $(\ASa,\Lambda) = (\ASa(A),\Lambda(A))$ of the dressing map, i.e.\ we prove that $A\mapsto (\ASa(A),\Lambda(A))$ is surjective with $\G_\rel$-fibre.

First, we show that it is surjective, i.e.\ for any $(\ASa_\star,\Lambda_\star)\in \hcalA \times G^S_0 $ we find an $A\in\mathcal{A}$ such that $(\ASa(A),\Lambda(A))=(\ASa_\star,\Lambda_\star)$. Pick any such $(\ASa_\star,\Lambda_\star)$. Since $\Lambda_\star$ is in the identity component $G^S_0\subset\G_S$, a homotopy\footnote{If two maps are homotopic they are also smoothly homotopic \cite[Prop. 10.22 (1st ed.) or Thm. 6.29 (2nd ed.)]{lee2013introduction}. We thank the n-lab for their help in finding accurate references. \label{fnt:smoothhomotopy}} $H:[0,1]\to G^S_0$ such that $H(0)=1$ and $H(1)=\Lambda_\star$ exists. From this homotopy, define $V_H \in \G \simeq C^\infty(I,\G_S)$ as 
\[
\dfU_H(u) = H(\tfrac12(1-u)) 
\]
so that $\dfU_H^\f  = 1$ and $\dfU_H^\i = \Lambda_\star$.
We then claim that 
\[
A_H \doteq \dfU_H\ASa_\star \dfU_H^{-1} - d \dfU_H \dfU_H^{-1} 
\]
is such that $(\ASa[A_H],\Lambda[A_H]) = (\ASa_\star, \Lambda_\star)$.
Indeed, since $(A_H)_\ell = - L_\ell \dfU_H \dfU_H^{-1}$ and $\dfU_H^\f=1$, one uses Lemma \ref{lemma:Wilsonline} to conclude that $\dfU[A_H] = \dfU_H$. Hence,
\begin{align*}
(\ASa[A_H],\Lambda[A_H]) & =  \big(A_H^{\dfU[A_H]},\ \dfU[A_H]^\i \big) = \big(A_H^{\dfU_H},\dfU_H^\i \big) \\
&= \big(\dfU_H^{-1} (\dfU_H \ASa_\star \dfU_H^{-1}  -d \dfU_H^{-1} )\dfU_H + \dfU_H^{-1}  d \dfU_H, \Lambda_\star\big) \\
&= (\ASa_\star, \Lambda_\star)
\end{align*}
and thus $A_H$ is in the preimage of $(\ASa_\star,\Lambda_\star)$ along the $(\ASa,\Lambda)$-components of $\cV$.

We now prove that $\cV$ has $\G_\rel$ as a fibre. This is equivalent to showing that:
given $A,A'\in\mathcal{A}$, $(\ASa(A),\Lambda(A))=(\ASa[A'],\Lambda[A'])$ iff there exists a $g\in \G$ such that $A' = A^g$ and $g\in\tGo$:

\begin{itemize}
\item[($\Leftarrow$)] Assume $A'=A^g$ for $g\in \G_\rel$, then using Lemma \ref{lemma:dressing-properties}:
\[
(\ASa[A'],\Lambda[A']) = (\ASa[A^g],\Lambda[A^g]) = ( \ASa(A)^{g^\f} , (g^\i)^{-1}\Lambda(A) g^\f)= (\ASa(A),\Lambda(A))
\] 
where the last equality follows from $g\in \tGo \implies g^\f = g^\i =1$.
\item[($\Rightarrow$)] Assume now $(\ASa(A),\Lambda(A)) = (\ASa[A'],\Lambda[A'])$. The equality between the $\ASa$-components, i.e.\ $\ASa(A)=\ASa[A']$, reads 
\[
\dfU^{-1} A \dfU + \dfU^{-1} d\dfU = \dfU'{}^{-1} A' \dfU' + \dfU'{}^{-1} d\dfU' ,
\] 
which is equivalent to $A' = A^g$ for $g \doteq \dfU \dfU'{}^{-1}$, where we denoted $\dfU\equiv\dfU(A)$ and $\dfU'\equiv\dfU[A']$. Since $\dfU^\f = \dfU'{}^\f =1$, it follows that $g^\f =1$. 

\noindent The equality between the $\Lambda$-components, i.e.\ $\Lambda(A) = \Lambda[A']$, reads $\dfU^\i = \dfU'{}^\i$, which implies $g^\i = 1$. Therefore $g\in G^\Sigma_{\rel}$. Since $\tGo \doteq G^\Sigma_0\cap G^\Sigma_\rel$, we are only left to show that $g\in G^\Sigma_0$ lies in the identity component of the mapping group. Making its dependence from $(A,A')$ explicit again, i.e.\ $g = \dfU(A) \dfU[A']^{-1}$, we see that a homotopy $H(t) \in G^\Sigma$ between $g$ and the identity is given by $H(t) = \dfU[tA] \dfU[tA']^{-1}$. 
\end{itemize}
\end{proof}

\begin{calculation}\label{prop:step2}
$\cV^*(\omAS  + \Omega_S) = \iota_\C^*\omega$.
\end{calculation}

\begin{proof}
Recall that in the left-invariant trivialisation $T^*G^S_0 \simeq G^S_0 \times (\fg^S)^*$, and under our usual identification of $\fg^S$ and $(\fg^S)^*$ provided by $\ASe\mapsto \tr(\ASe\cdot)\vol_S$, the canonical symplectic form $\Omega_S$ reads
\begin{align*}
\Omega_S(\Lambda,\ASe)
&= \bd \int_S \btr( \ASe\ \Lambda^{-1} \bd \Lambda  ) \\
&= \int_S \btr(\bd \ASe \wedge \Lambda^{-1} \bd \Lambda  - \tfrac12 \ASe [\Lambda^{-1}\bd \Lambda,\Lambda^{-1}\bd \Lambda ]) 
\end{align*}
By a direct calculation, performed in Appendix \ref{app:compute-bom-dress}, one can show that for any $U \in C^\infty(\X,\G)$,
\begin{align*}
\omega(A^U, E^U) = \omega(A,E) & - \bd \int_\Sigma \btr(\mathsf{G}\ \bd U U^{-1}) \\ 
    & + \Omega_S(U^\f, \Ad(U^\f)^{-1}\cdot E^\f) - \Omega_S(U^\i, \Ad(U^\i)^{-1}\cdot E^\i).
\end{align*}
The claim then follows by imposing the Gauss-constraint, setting $U= \dfU(A)$, and noting that: 
\begin{enumerate}
\item on-shell of the Gauss constraint, $\mathsf{G}|_\C=0$, one has $E=E(A,E_\i)$ with $E^\i=E(A,E_\i)^\i=E_\i$ (Proposition \ref{prop:constr-surf});
\item on-shell, $\ASe(A,E_\i) = \Ad(\dfU(A)^\i)^{-1}\cdot E_\i$ (Definition \ref{def:dressingmap}); 
\item $\dfU(A)^\f \equiv 1$ and $\dfU(A)^\i = \Lambda(A)$ (Definition \ref{def:dressingmap});
\item the quantity $\ASa(A) = A^{\dfU(A)}$ is purely spatial, i.e.\ $(A^{\dfU(A)})_\ell=0$ (Lemma \ref{lemma:dressing-properties}\ref{lemma:dressing-properties-ii}); 
\item the restriction of $\omega$ to the space of purely spatial connections i.e.\ to connections such that $A_\ell = 0$, is given by the AS 2-form $\omAS $ (in particular, the dependence on $E$ drops out). 
\end{enumerate}
Indeed, if $(A,E_\i)\in\C$ one finds:
\begin{align*}
\iota_\C^*\omega(A,E_\i) 
& = \omega(A, E(E_\i;A))\\
&= \omega(A^{\dfU(A)}, E(E_\i;A)^{\dfU(A)}) - \Omega_S(\dfU(A)^\f, \Ad(\dfU(A)^\f)^{-1}\cdot E^\f) \\ & \qquad + \Omega_S(\dfU(A)^\i, \Ad(\dfU(A)^\i)^{-1}\cdot E_\i)\\
& =  \omAS (\ASa(A)) + \Omega_S( \Lambda(A), \ASe(A,E_\i))\\
& = (\omAS +\Omega_S)(\cU(A,{E_\i})) = \cU^*(\omAS  + \Omega_S)(A,{E_\i}).
 \qedhere
\end{align*}
\end{proof}

\begin{remark}
Calculation \ref{prop:step2} allows us to give a solid mathematical interpretation to the dressing map $\cV$, by showing how it effectively allows to present (up to a covering, see below) the constraint-reduced phase space $\uCo$, in terms of the---much easier to handle---\emph{extended} AS phase space. This is one way to correctly interpret the following procedure, used in many other contexts (typically in Chern--Simons/$BF$ theory, or general relativity on null or Riemannian/spacelike $\Sigma$'s): (1) choose a ``naive'' gauge fixing, (2) build the associated dressing, (3) plug it into the on-shell symplectic form and interpret the result in terms of an ``extended phase space''. See e.g.\  \cite{Francois2012,Francois2021, RielloGomesHopfmuller, RielloGomes, CarrozzaHoehn, DonnellyFreidel16,}.
\end{remark}

\begin{proof}[Proof of Theorem \ref{thm:nonAb-constr-red}]
The map $\cV:\C\to\XAS$ is a surjection with fibre $\G_\rel$ (Proposition \ref{prop:step1}). Being $\G_\rel$-invariant, it is in particular $\Go$-invariant. Hence, there exists a (smooth!) map $p_V\colon \uCo\to \XAS$ such that $\cV=\pi_\circ^*p_V$. Now, since $\uCo = \C/\Go$, there is a residual action of the discrete group $\mathcal{K} = \G_\rel/\Go$ on $\uCo$. The map $p_V$ is invariant under this discrete action because $\cV$ was invariant under the whole of $\G_\rel$, meaning that the fibre of $p_V$ is discrete and is thus a covering. To show that $p_V$ is a symplectic covering, 
\[
\uomegao \stackrel{!}{=}p_V^*\omeAS,
\]
we use the injectivity of $\pi^*_\circ$ and the following string of equalities:
\[
\pi_\circ^* (p_V^* \omeAS ) = \cV^*\omeAS = \iota_\C^*\omega = \pi_\circ^*\uomegao,
\]
which hold in virtue of the definition of $p_V$ (first equality), of Calculation \ref{prop:step2} (second equality) and the definition of $\uomegao$ (third equality).

In sum, we showed that $p_V :\uCo \to \XAS$, the dashed arrow of diagram \eqref{e:coveringdiagram}, is a symplectic covering and therefore a local symplectomorphism a fortiori.

The proof of the very last point of Theorem \ref{thm:nonAb-constr-red}, the global symplectomorphism $\uCo\simeq \XAS$, only valid in the Abelian case, is performed in Appendix \ref{app:Abelianreduction}. The logic is similar, but $\dfU(A) \doteq \Pexp \int^1_u A_\ell$, well-defined for any $G$, is replaced by its ``logarithm'', well-defined only if $G$ is Abelian. The discrete fibre of the symplectic covering corresponds to the periodicity of the exponential function in the $U(1)$ factors of the Abelian group $G$. (see Appendix \ref{app:Abelianreduction},  in particular Remark \ref{rmk:localsymp}).
\end{proof}

We summarise the relation between first stage reduction and the extended AS phase space with the following diagram:
\begin{equation}\label{e:coveringdiagram}
\xymatrix{
\C \ar[r]^-{\pi_\circ} \ar@/_1.5pc/[rr]_{\cV} & \uCo \ar@{-->}[r]^-{p_V} & \XAS
}
\end{equation}

\begin{remark}[Dressing vs.\ symplectic reduction]
In the discussion above we have always assumed that $\uCo$ is a smooth symplectic manifold. As discussed in \cite{RielloSchiavina} and references therein, smoothness depends on several factors, including checking that the action of $\G$ on $\C$ is free and proper. In infinite dimensions, to have an honest symplectic manifold one additionally needs to check that the image of the (constraint) gauge algebra along the infinitesimal action map is symplectically closed (see the discussion in \cite[Proposition 5.8]{RielloSchiavina}, and \cite[Proposition 4.1.7 and Lemma 4.2.14]{DiezPhD}). This can often be done \cite{DiezHuebschmann-YMred,DiezPhD,DiezRudolph2020}, but we offer here an alternative perspective. Indeed, the ``dressing field method'' we outlined here can be taken as a concrete construction to explicitly produce a manifestly smooth and symplectic ``realisation'' of $\uCo$, in the present case, the extended AS phase space, built as the image of the dressing map, in the sense that the constraint-reduced phase space $\uCo$ is a symplectic covering of said ``realisation''. For a concrete example see Appendix \ref{app:Abelianreduction} where we explicit identify the constraint reduced space in the Abelian case and explain the relevant covering in terms of the branching of the logarithm on the toric factors of $G$.
\end{remark}

\subsection{Proof of Theorem \ref{thm:fibrecharacterisation}}

In order to provide a characterisation of the fibre of the symplectic covering, i.e.\ $\mathcal{K} = \tGo/\Go$, 
we start by characterising the connected components of the mapping group and of its relative component.

\begin{remark}[Pointed maps]
Let us recall a standard notion: given two pointed spaces $M$ and $N$, $C_\star(M,N)$ is the space of \emph{pointed} (continuous) maps which by definition map the base point of $M$ onto the base point of $N$. We also denote $G^M_\star \subset G^M$ the pointed subgroup of elements that are the identity at a given point over $M$.
\end{remark}

\begin{lemma}[Connected components]\label{lemma:pezzo2}
For $G$ any finite-dimensional connected Lie group, $\Sigma\simeq S \times I$ and $S\simeq S^{n-1}$ a sphere, with $S^0=\{\mathrm{pt}\}$,
\[
\pi_0(G^\Sigma) \simeq \pi_{n-1}(G)
\quad\mathrm{and}\quad
\pi_0(G^\Sigma_\rel) \simeq
\begin{dcases}
\pi_1(G) & \mathrm{if }n=1,2\\
\pi_1(G)  \oplus \pi_n(G) & \mathrm{if } n>2.
\end{dcases}
\]
In particular, 
\begin{enumerate}[label=(\roman*)]
    \item if $G$ is simply connected, then
    \begin{itemize}
        \item for $n=1,2$: $\pi_0(G^\Sigma) = \pi_0(G^\Sigma_\rel) = 1$, 
        \item for $n=3$: $\pi_0(G^\Sigma)=1$ and $\pi_0(G^\Sigma_\rel)\simeq\pi_3(G)\simeq\mathbb{Z}^s$ for $s\in\mathbb{N}$;\footnote{We have that $s>0$ if $G$ is semisimple, and $s=1$ if $G$ is simple.}
    \end{itemize}
    \item if $G\simeq U(1)^t\times \mathbb{R}^k$ (Abelian), then
    \begin{itemize}
        \item for $n=2$: $\pi_0(G^\Sigma)\simeq\pi_0(G^\Sigma_\rel) \simeq \mathbb{Z}^t$, 
        \item for $n\neq2$: $\pi_0(G^\Sigma)=1$ while $\pi_0(G^\Sigma_\rel)\simeq\mathbb{Z}^t$.
    \end{itemize}
\end{enumerate}

\end{lemma}
\begin{proof}
First, we (implicitly) replace all smooth mapping groups with continuous ones; this replacement does not affect their homotopy type in view of the ``approximation theorem'' \cite[Theorem V.2.11]{Neeb-locallyconvexgroups} (see references therein), or ``weak homotopy equivalence'' \cite[Theorem 3.2.13]{WockelPhD} (where the result is generalised to gauge groups associated to nontrivial bundles).

Let us begin by computing the zeroth homotopy of $G^\Sigma$. 
Notice that since $\Sigma \simeq I\times S$, any $g(u,x)\in G^\Sigma$ is homotopic to a constant function over $I$---e.g. via the homotopy $H(t)(u,x) = g(tu,x)$. Therefore if   $S\simeq S^{n-1}$,
\[
\pi_0(G^\Sigma) = \pi_0(G^S)  \equiv \pi_0( C(S,G) ) \simeq \pi_0(C_\star(S,G))  \equiv \pi_{n-1}(G),
\]
where the third step holds in virtue of $G$ being connected.

For the second part of the lemma, we start by noting that, since $g\in G^\Sigma_\rel$ iff $g^\i = g^\f =1 $, one can identify $G^\Sigma_\rel\simeq C_\star (S^1, G^S)$. It follows that 
\[
\pi_0(G^\Sigma_\rel) \simeq \pi_1(G^S).
\]

In homotopy theory, the natural definition of the zeroth sphere is as a pair of points $\{- 1,+1\}$. Since here we instead set $S^0 \doteq \{\mathrm{pt}\}$, we need to distinguish the case $n=1$ from the rest.

If $n=1$, $S=\{\mathrm{pt}\}$ and $G^S \simeq G$, which leads to 
\[
\pi_0(G^\Sigma_\rel) \simeq \pi_1(G^S) \simeq \pi_1(G) \qquad (n=1).
\]

If $n>1$,\footnote{One can adapt the proof for the case $n>1$ to the case $n=1$ by noting that, with \emph{our} definition of a zeroth sphere, while $S^1\wedge S^0 \not\simeq S^1$, one has $C_\star(S^0,G) \simeq \{ \mathrm{pt}\}$ and therefore $\pi_1(G^S_\star) = 1$.} recall that $G^S \simeq C(S,G)$, and thus introduce $G^S_\star \simeq C_\star(S,G)$ after choosing a base point over $S$---say the North pole if $S\simeq S^{n-1}$. Hence, we note that from $G^S \simeq G^S_\star \rtimes G$, it follows that, for any connected Lie group $G$,
\[
\pi_1(G^S) \simeq \pi_1(G^S_\star) \oplus \pi_1(G).
\]
We thus need to compute $\pi_1(G^S_\star)$.
Since $\pi_1(G^S_\star) \simeq \pi_0( C_\star( S^1,  C_\star(S, G)) \simeq \pi_0(  C_\star(S^1 \wedge S, G))$ where $S^1\wedge S$ denotes the smash product between $S^1$ and $S$. Then, if $S\simeq S^{n-1}$, we have that $S^1 \wedge S^{n-1} \simeq S^n$ and therefore\[
S\simeq S^{n-1} \implies \pi_1(G^S_\star) \simeq \pi_n( G). 
\]
Thus, putting everything together, we conclude:
\[
\pi_0(G^\Sigma_\rel)\simeq \pi_1(G^S) \simeq \pi_n(G)\oplus\pi_1(G) \qquad (n>1).
\]
Now, let us consider the following facts (see e.g. \cite{MIMURA}): 
\begin{enumerate}
    \item any connected, finite dimensional, Lie group $G$ is diffeomorphic to $H\times \mathbb{R}^k$ for some $k\in\mathbb{N}$ and $H\subset G$ a compact connected Lie subgroup; 
    \item any compact, connected, Lie group $H$ is isomorphic to a quotient $H\simeq\tilde H/Z$ of a group $\tilde H \simeq K_1\times \cdots \times K_s \times U(1)^t$---where $t\in\mathbb{N}$ and the $K_i$'s are compact, connected, simply connected, simple Lie groups---by a finite central subgroup $Z\subset \tilde H$;
    \item for any finite dimensional Lie group $G$, $\pi_2(G)=1$;
    \item for any connected, compact, simple, Lie group $\pi_3(K) \simeq \mathbb{Z}$.
\end{enumerate}
From (1) it follows that $\pi_m(G) \simeq \pi_m(H)$ for all $m$. Thus, from (2) we deduce:
\[
1 \to \mathbb{Z}^t \to \pi_1(G) \to \pi_0(Z) \to 1.
\]
and
\[
\pi_m(G) \simeq \pi_m(K_1) \oplus \cdots \oplus \pi_m(K_s) \quad \mathrm{for }m>1.
\]
Finally, from (3--4), we deduce the particular cases (i--ii) of the lemma's statement. 
\end{proof}

\begin{remark}[Strategy of proof: $\G_\rel$ vs. $G^\Sigma_\rel$]
We can now address the proof of Theorem \ref{thm:fibrecharacterisation}. In Theorem \ref{thm:nonAb-constr-red} we have seen that $\uCo$ is a symplectic covering space of the extended AS phase space, with fibre $\mathcal{K}$ characterised in Proposition \ref{prop:discretequotientgroup} as the group of components of $\G_\rel \simeq G^\Sigma_{0,\rel}$ (resp. $G\cdot \G_\rel \simeq G\cdot G^\Sigma_{0,\rel}$ in the Abelian case). Since $\mathcal{K}$ is discrete, to see whether it is trivial or not it is enough to observe that its elements are in 1-to-1 correspondence with the connected components of $G^\Sigma_{0,\rel}$ i.e.\ with the elements of $\pi_0(G^\Sigma_{0,\rel})$. Computing the homotopy type of $G^\Sigma_\rel$ is, however, more practical in virtue of the fact that this is a mapping group---but of course the two computations coincide whenever $G^\Sigma$ is connected, since obviously $G^\Sigma = G^\Sigma_0$ implies $G^\Sigma_\rel = G^\Sigma_{0,\rel}$. Lemma \ref{lemma:pezzo2} then characterises the connectedness of $G^\Sigma$ and $G^\Sigma_{\rel}$ in terms of the homotopy groups of the finite dimensional group $G$, and thus allows us to deduce information about $\mathcal{K}$ in many cases of interest. 
\end{remark}

\begin{proof}[Proof of Theorem \ref{thm:fibrecharacterisation}]  
From the proof of Theorem \ref{thm:nonAb-constr-red} it emerges that the fibre of the covering ${p_V}:\uCo \to \hcalA \times T^*G^S_0$ is given by $\mathcal{K} \simeq \G_\rel/\Go$ in the semisiple case, and by $\mathcal{K} \simeq (G\cdot \G_\rel)/\Go$ in the Abelian case.

An immediate corollary of Proposition \ref{prop:discretequotientgroup} is that $\mathcal{K}$ is the (discrete) group of components of $G^\Sigma_{0,\rel}$ and that $G_{0,\rel,0}^\Sigma = G^\Sigma_{\rel,0}$, so that---in both the semisimple and Abelian cases---we can write:
\[
\mathcal{K} \simeq G^\Sigma_{0,\rel}/G^\Sigma_{\rel,0}.
\]

Next, introduce the (discrete) group of components of $G^\Sigma_\rel$, i.e.\ $\mathcal{K}' \doteq G^\Sigma_\rel/G^\Sigma_{\rel,0}$. Since $G^\Sigma_{0,\rel}\subset G^\Sigma_{\rel}$ is a normal subgroup (cf. Remark \ref{rmk:normalsubgrps}), we deduce that $\mathcal{K}\subset \mathcal{K}'$ is a normal subgroup, and thus that
\[
\pi_0(\mathcal{K})\subset \pi_0(\mathcal{K}') \simeq \pi_0(G^\Sigma_\rel).
\]

If $G^\Sigma$ is connected, i.e.\ if $G^\Sigma = G^\Sigma_0$, one has that $\mathcal{K}$ is isomorphic to $\mathcal{K}'$. In other words:
\[
\pi_0(G^\Sigma) = 1 \implies \mathcal{K} \simeq \mathcal{K}'.
\]

Therefore, from the last two equations we see that we can deduce properties on the nature of the covering fibre $\mathcal{K}$ by studying the zeroth homotopy type of $G^\Sigma$ and $G^\Sigma_\rel$. 
Lemma \ref{lemma:pezzo2} characterises these sets in many cases of interests, that yield the case analysis given in the statement of the theorem---in particular, the general statement (iv).

In fact,
\begin{enumerate}[label=(\roman*)]
\item when $G$ is simply connected and $n=1$ or $2$, we have that: $\pi_0(G^\Sigma)\simeq \pi_{n-1}(G) =1$ which implies $\mathcal{K}\simeq\mathcal{K}'\simeq \pi_1(G)=1$. Since the covering fibre $\mathcal{K}$ is trivial, we also conclude that the map ${p_V}\colon \uCo \to \XAS$ is a global symplectomorphism;
\item when $G$ is simply connected and $n=3$, we have again that $\pi_0(G^\Sigma)\simeq \pi_{n-1}(G) = \pi_{2}(G) = 1$, and thus $\mathcal{K}\simeq \mathcal{K}'\simeq \pi_0(G^\Sigma_\rel)\simeq \pi_3(G)\simeq \mathbb{Z}^s$ for some $s$ ($s>0$ if $G$ is semisimple, and $s=1$ if $G$ if simple);
\item finally, when $G\simeq U(1)^t\times \mathbb{R}^k$ (Abelian) and $n\not=2$, we have that, in virtue of Lemma \ref{lemma:pezzo2}, $\pi(G^\Sigma)=\pi_{n-1\not=1}(G) = \pi_{n-1\neq 1} (U(1))^{\oplus t} = 1$, which tells us that $\mathcal{K}\simeq\mathcal{K}'\simeq \pi_1(G) \simeq \mathbb{Z}^t$.
\qedhere
\end{enumerate}

\end{proof}


\section{Symplectic reduction of null YM theory: second stage}\label{sec:reduction-second}

\noindent\fbox{%
    \parbox{\textwidth}{\small\center
        The reader interested in applications to soft symmetries can skip this section at first.
    }%
}

\medskip

From the general theory of reduction by stages we know that $\uCo$ carries a Hamiltonian action of $\Gred$ with momentum map $\uh\colon \uCo \to \fGred^*$.
In the previous section we have built a local\footnote{``Local'' here refers to field-space, and not to $\Sigma$.} model for the constraint-reduced phase space $(\uCo,\uomegao)$ in terms of the extended AS phase space $\XAS\doteq\hat{\mathcal{A}}\times T^*G_0^S$.
In particular $(\XAS,\omeAS)$ is a symplectic cover of $(\uCo,\uomegao)$ with fibre $\mathcal{K}$.
We will use the model Hamiltonian action $G^{\pp\Sigma}_0\circlearrowright(\XAS,\omeAS)$ to infer the structure of the second-stage reduction of $\uCo$ by $\Gred$.

\subsection{Strategy}

Our first goal is to show that the Hamiltonian action of $\Gred$ on $(\uCo,\uomegao)$ is related to the Hamiltonian action of $G^{\pp\Sigma}_0$ on $(\XAS,\omeAS)$ by means of the following diagram (cf. Equation \eqref{e:coveringdiagram}):
\[
\xymatrix{
(\C, \G, \iota_\C^*h) \ar[r]^-{\pi_\circ} \ar@/_1.5pc/[rr]_{\cV} & (\uCo, \Gred, \uh) \ar@{-->}[r]^-{p_V} & (\XAS, G^{\pp\Sigma}_0, \uhAS)
}
\]

\begin{remark}
While $(\XAS,\omeAS)$ admits a natural action by $G^{\pp\Sigma}_0$ for any $G$, when $G$ is Abelian the subgroup $G$ of constant transformations acts trivially, whence the $G^{\pp\Sigma}_0$-action effectively reduces to that of $G^{\pp\Sigma}_0/G$.  
For simplicity of exposition, when working over $\XAS$ we will only refer to the action of $G^{\pp\Sigma}_0$ and simply observe that it appropriately descends to $G^{\pp\Sigma}_0/G$ when $G$ is Abelian. (See items \ref{item6-Gred} and \ref{item6-Abelian} in the list below.)
\end{remark}

To prove that the above diagram commutes, we will proceed by steps, showing that:
\begin{enumerate}
    \item $\mathcal{K}\subset \Gred$ is a (discrete) normal subgroup, and one has (Proposition \ref{prop:Gred})\label{item6-Gred}
    \[
    \Gred/\mathcal{K} \simeq 
    \begin{dcases}
        G^{\pp\Sigma}_0 & \text{if $G$ semisimple},\\
        G^{\pp\Sigma}_0/G & \text{if $G$ Abelian};
    \end{dcases}
    \]
    \item the flux map $\uh\colon \uCo \to \fGred^*$ is $\mathcal{K}$-invariant (Lemma \ref{lem:flux-K-invariance});\label{item6-Kinv}
    \item $G^{\pp\Sigma}_0$ admits a Hamiltonian action on $(\XAS,\omeAS)$ with momentum map $\uhAS$ such that (Proposition \ref{prop:red-flux-map}, see also Equation \eqref{e:fluxdiagram})\label{item6-uhAS}
    \[
    \iota_\C^*h = \cV^*\uhAS;
    \]
    \item in the Abelian case $G\hookrightarrow G^{\pp\Sigma}_0$ acts trivially on $\XAS$ and therefore $\uhAS$ automatically defines a momentum map for $\Gred/\mathcal{K}$ in both the semisimple and Abelian cases (Equations \ref{e:reducedfluxmap} and \ref{eq:huAS-Ab}).\label{item6-Abelian}  
\end{enumerate}
The proof of the facts (1--4) can be found in Section \ref{sec:pf1-4} below. From these, recalling that $p_V:\uCo \to \XAS$ is a smooth $\mathcal{K}$-covering, we see that the flux momentum map $\uh$ on $(\uCo,\uomegao)$ descends along $p_V$  to a momentum map $\uhAS$ on $(\XAS,\omeAS)$.

The following commutative diagram summarises the relation between the various flux momentum maps in the semisimple case:
\begin{equation}\label{e:fluxdiagram}
\xymatrix{
\ar@/_1.5pc/[dd]_{\cV} \C \ar[d]^{\pi_\circ}  \ar[rr]^{\iota^*_\C h}&& \fG^*_{\loc} \\
\uCo  \ar@{-->}[d]^-{{p_V}} \ar[rr]^{\uh}&& \fGred^* \ar[u] \\
\XAS \ar[rr]^-{\uhAS} && (\fg^{\pp\Sigma})^* \ar[u]^{\simeq} \ar@/_1.5pc/[uu]
}    
\end{equation}
(When $G$ is Abelian, replace $(\fg^{\pp\Sigma})^* \leadsto (\fg^{\pp\Sigma}/\fg)^* \simeq \Ann(\fg,(\fg^{\pp\Sigma})^*) \subset (\fg^{\pp\Sigma})^*$. Note that $\Ann(\fg,(\fg^{\pp\Sigma})^*)\simeq \fGred^*$ still holds.)

After having clarified the relationship between the Hamiltonian structures on $(\uCo,\uomegao,\uh)$ and $(\XAS,\omeAS,\uhAS)$, we turn our attention to the study of the  symplectic leaves of the second-stage reduced  Poisson space ${\uuC} \doteq \uCo/\Gred = \bigsqcup_{\mathcal{O}_f\subset \F} \uuS_{[f]}$, 
\[
\uuS_{[f]} \doteq \uh^{-1}(\mathcal{O}_f)/\Gred,
\]
called flux superselection sectors. Once again, our strategy will be to model them on the (symplectic) AS sectors (Definition \ref{def:ASSectors}) labelled by the same flux orbit $\mathcal{O}_f\subset \F=\Im(\uh)\simeq \Im(\uhAS)$ (Proposition \ref{prop:red-flux-map}), so that we can use the extended Ashtekar--Streubel phase space to model second stage reduction:
\[
\uuS^\AS_{[f]} \doteq \uhAS^{-1}(\mathcal{O}_f)/G^{\pp\Sigma}_0, \qquad \uuS_{[f]} \simeq\uuS^\AS_{[f]}, \qquad {\uuC} \simeq \XAS/ G^{\pp\Sigma}_0.
\]

\subsection{Proof of the statements (1--4)}\label{sec:pf1-4}

First, in the next proposition and lemma, we clarify the relationship of $\mathcal{K}$ to $\Gred$ and $\uh$:
\begin{proposition}\label{prop:Gred}
$\Gred$ is a connected central extension by the discrete group $\mathcal{K}$ of $G^{\pp\Sigma}_0$ if $G$ is semisimple, or of $G^{\pp\Sigma}_0/G$ if $G$ is Abelian.
In particular, there exists a short exact sequence of groups:
\[
\begin{dcases}
1 \to\mathcal{K} \to \Gred \to G^{\pp\Sigma}_0\to 1 & G \ \mathrm{semisimple}\\
1 \to\mathcal{K} \to \Gred \to G^{\pp\Sigma}_0/G\to 1 & G \ \mathrm{Abelian}
\end{dcases}
\]
\end{proposition}

\begin{proof}
Connectedness of $\Gred \doteq \G/\Go$ is a consequence of that of $\G$, while normality of $\mathcal{K}\doteq \G_\rel/\Go$ in $\Gred\doteq \G/\Go$, follows from the normality of $\G_\rel\subset\G$ (Remark \ref{rmk:Grel-normal}).

To prove the centrality of $\mathcal{K}$ in $\Gred$ we need to show that, for every $[h]\in \mathcal{K}$ and every $[g]\in\Gred$,
\[
[h][g][h]^{-1}=[g]. 
\]
Since $\Go$ is normal we have $[h][g_\circ] = [hg_\circ]$, so that the above condition becomes $[h][g][h]^{-1}[g]^{-1} = [h g h^{-1} g^{-1}] = [1]$, that is to say
\[
\mathcal{K}=\G_\rel/\Go\ \text{central in}\ \Gred = \G/\Go \iff
 \forall (h,g)\in \G_\rel\times \G,
\quad h g h^{-1} g^{-1} \in \Go 
\]
Now, since $h\in \G_\rel$ (a normal subgroup of $\G$), we see that $f\doteq h g h^{-1} g^{-1}\in \G_\rel$ as well, so that for it to be in $\Go \simeq(\G_\rel)_0$ we are left to prove that $f$ is connected to the identity (cf. Proposition \ref{prop:discretequotientgroup}). However, $g\in \G\simeq G^\Sigma_0$ is connected to the identity, and thus there exists a homotopy $g(t)$ with $g(0) = 1$ and $g(1)=g$. Then, we can construct the homotopy $f(t) = h g(t) h^{-1} g(t)^{-1}$ and easily check that $f(0)=1$ and $f(1)=f=h g h^{-1} g^{-1}$, proving the statement.

Finally, to prove the exact sequences given in the statement, we will focus on the semisimple case, where we will show that 
\[
\Gred/\mathcal{K}\ \stackrel{(1)}{\simeq}\  \G/\tGo \ \stackrel{(2)}{\simeq} \iota_{\pp\Sigma}^* \G \ \stackrel{(3)}{\simeq}\  G^S_0 \times G^S_0 \ \stackrel{(4)}{\simeq}\  G^{\pp\Sigma}_0.
\]
(In the Abelian case, one can prove the following sequence of isomorphism in a similar manner:
\[
\Gred/\mathcal{K} \ \simeq \ \G/(G\cdot\tGo) \ \simeq (\iota_{\pp\Sigma}^* \G) /G \ \simeq \ G^{\pp\Sigma}_0/G. 
\]
We omit the explicit proof in this case.)

\begin{enumerate}
\item[$(1)$] This follows from Proposition \ref{prop:discretequotientgroup}: $\G/\tGo \simeq (\G/\Go) / (\tGo/\Go) \doteq \Gred/\mathcal{K}$.
\item[$(2)$] This follows from the definition of $\G_\rel = \{ g\in \G \ : \ g\vert_{\pp\Sigma} = 1\} \simeq G^\Sigma_0 \cap G^\Sigma_\rel$. 
The morphism $\iota_{\pp\Sigma}^*\colon \G \to \iota^*_{\pp\Sigma}\G$ is surjective by construction and $\G_\rel=\ker(\iota^*_{\pp\Sigma})$. Indeed, for any two $g,g'\in G^\Sigma_0\simeq \G$ one has that: $\iota_{\pp\Sigma}^*g = \iota_{\pp\Sigma}^*g' \implies  \iota^*_{\pp\Sigma}(g'g^{-1}) = 1$, i.e.\ $g'g^{-1}\in G^\Sigma_0 \cap G^\Sigma_\rel \simeq \tGo$, and $\ker(\iota^*_{\pp\Sigma}) \subset \G_\rel$. The opposite inclusion is obvious, and we conclude that $\iota^*_{\pp\Sigma}$ descends to a bijection $\G/\G_\rel\to \iota^*_{\pp\Sigma}\G$.
\item[$(3)$] This is a consequence of $\G\simeq G^\Sigma_0$ and the fact that that $\iota_{\pp\Sigma}^* G^\Sigma_0$ and  $G^S_0 \times G^S_0$ are included in each other. 
Indeed, let $g\in G^\Sigma_0$ and thus consider a homotopy  $H(t) \in G^\Sigma_0 $ between $g$ and the identity; then it follows that $\iota_{S_\i}^* H(t) \in G^S$ is a homotopy between $\iota_{S_\i}^* g \in G^S$ and the identity, and therefore that $\iota_{\pp\Sigma}^* G^\Sigma_0 \subset G^S_0 \times G^S_0$. 
For the opposite inclusion, consider $(g_\i,g_\f)\in G^S_0 \times G^S_0$ and $H_{\i/\f}(t)$ the respective homotopies with the identity in $G^S_0$. Then, define $\tilde{g}\in C(I,G^S)$ as $\tilde{g}(u) \doteq \Theta(-u)H_\i(-u) + \Theta(u) H_\f(u)$, a continuous function $I\to G^S$ with the property that $\iota_{\pp\Sigma}^*\tilde{g} = (g_\i,g_\f)$.
Using standard arguments in homotopy theory, we can construct a smooth homotopy $g\in C^\infty(I,G^S) \simeq G^\Sigma$ out of the continuous homotopy $\tilde{g}$. (See Footnote \ref{fnt:smoothhomotopy}.)
It is then easy to see that such a $g$ is connected to the identity in $G^\Sigma$ (retract to $S_\i$ and then to the identity) and thus prove $\iota_{\pp\Sigma}^*G^\Sigma_0 \supset G^S_0 \times G^S_0$. 
\item[$(4)$] Obvious. \qedhere
\end{enumerate}
\end{proof}

\begin{lemma}\label{lem:flux-K-invariance}
 The flux map $\uh\colon \uCo \to \fG^*_\loc$ is $\mathcal{K}$-invariant. 
\end{lemma}
\begin{proof}
This follows from the fact that the flux map $\h\colon\X\to \fG^*_\loc$ is manifestly invariant under the whole of $\G_\rel$, and not just its identity component $\Go$. 
Hence, for $k=[\tilde{g}]\in\mathcal{K}=\G_\rel/\Go$ and $\underline{\phi} = [\phi]\in\uCo = \C/\Go$, if we denote $(k\cdot \uh)(\underline{\phi})\doteq \uh(\underline{\phi}\triangleleft k)$ and $(\tilde g\cdot h)(\phi)\doteq h({\phi}\triangleleft \tilde g) = \Ad^*(\tilde{g})\cdot h(\phi) = h(\phi)$, we have
\[
\pi_\circ^*(k\cdot \uh) = \tilde{g}\cdot \h = \h = \pi^*_\circ \uh,
\]
and we conclude in virtue of the injectivity of $\pi_\circ^*$.
\end{proof}

Next, we study the action of $G^{\pp\Sigma}_0$ on $(\XAS,\omeAS)$:

\begin{lemma}\label{lem:residualAction}
The dressing map $\cV\colon \C\to \XAS$ is equivariant with respect to the (right) action of $\G$ on $\X$ and the following (right) action of $G^{\pp\Sigma}_0\simeq G^S_0\times G^S_0$ on $\XAS\doteq\hcalA\times T^*G^S_0$:
\[
\begin{pmatrix}
\ASa\\ \Lambda \\ \ASe
\end{pmatrix} \triangleleft (g_\i,g_\f)
=
\begin{pmatrix}
g_\f^{-1}\ASa g_\f + g_\f^{-1} D g_\f\\
g_\i^{-1} \Lambda g_\f\\
\Ad(g_\f^{-1}) \cdot \ASe
\end{pmatrix}.
\]
Infinitesimally, denoting $(\xi_\i,\xi_\f)\in \fg^S\times \fg^S \simeq\fg^{\pp\Sigma}$, this action becomes
\[
\underline\varrho: \fg^{\pp\Sigma} \to \mathfrak{X}(\hcalA \times T^*G^S_0),
\qquad
\underline\varrho(\xi_\i,\xi_\f) 
\begin{pmatrix}
\ASa\, \\
\Lambda\\
\ASe
\end{pmatrix}
=
\begin{pmatrix}
{\cD}\xi_\f\\
-\xi_\i\Lambda + \Lambda \xi_\f \\
-\ad(\xi_\f)\cdot \ASe 
\end{pmatrix},
\]
where $\cD = D + [\ASa, \cdot]$ and $D$ is the de Rham differential on $S$.

\noindent Moreover, in the Abelian case, the above action reduces to an action of $G_0^{\pp\Sigma}/G$.
\end{lemma}

\begin{proof}
The fact that the map presented in the statement of the Lemma is indeed a right group action is straightforward from its expression. We then need to show that it commutes with the dressing map $\cV$.
The result follows from the definitions of the dressed, and hence $\tGo$-invariant, quantities  $\ASa(A) \doteq A^{\dfU(A)}$, $\Lambda(A) \doteq \dfU(A)^\i$, and $\ASe(A,E_\i) \doteq \Ad(\dfU(A)^\i)^{-1} \cdot {E_\i}$, as well as from the (bi-local) equivariance of $\Lambda(A)$ and $\ASe(A,E_\i)$ under the action of $\G$ (Remark \ref{rmk:bilocal}). Indeed one immediately sees that the action of $\G$ on $\X$ translates into the given action of $G^{\pp\Sigma}_0$ on $\XAS$. 

Finally, it is immediate to see that, when $G$ is Abelian, the subgroup $G\hookrightarrow G^{\pp\Sigma}_0$ of constant maps acts trivially.
\end{proof}

\begin{remark}[Residual gauge of the AS modes] \label{rmk:ASmodesGauge}
If $G$ is Abelian, the action of $(\xi_\i,\xi_\f)\in\fg^{\pp\Sigma}$ on $\hcalA$ can be expressed in terms of the AS modes $\tilde{\ASa}(k) \in \Omega^1(S,\fg)$ of Proposition \ref{prop:modedecomp}. One then finds that only the real part of the AS zero-mode transforms $2\tilde{\ASa}(0) = \ASa^{\int} - \ASa^\av + i \ASa^\dif$ transforms nontrivially:
\[
\underline{\varrho}(\xi_\i,\xi_\f) \mathfrak{Re}(2\tilde{\ASa}(k=0)) = D\xi_\f,
\]
while 
\[
\underline{\varrho}(\xi_\i,\xi_\f) \mathfrak{Im}(\tilde{\ASa}(k=0)) = \underline{\varrho}(\xi_\i,\xi_\f) \tilde{\ASa}(k\geq1)= 0.
\]
In the next proposition we will see that, since the Hamiltonian counterpart of $\mathfrak{Re}(2\tilde{\ASa}(0))=\ASa^{\int}-\ASa^\av$ is $\mathfrak{Im}(2\tilde{\ASa}(0))=\ASa^\dif$, in the extended AS phase space it is indeed $\ASa^\dif$ that ``generates the large gauge transformations'' associated with the final copy of $G^S_0$ (cit. \cite[Page 23]{StromingerLectureNotes}, see also Remark \ref{rmk:Stromingercf}).
\end{remark}

We now see that the action of $G^{\pp\Sigma}_0$ on $\XAS$ is Hamiltonian:

\begin{proposition}[Reduced flux map]\label{prop:red-flux-map}
The action 
\[
G^{\pp\Sigma}_0 \circlearrowright (\XAS,\omeAS) ,
\]
detailed in Lemma \ref{lem:residualAction}, is Hamiltonian,
\[
\bi_{\underline\varrho(\xi_\i,\xi_\f)} \omeAS  = \bd \langle \uhAS, (\xi_\i, \xi_f) \rangle,
\]
with momentum map $\uhAS: \XAS  \to (\fg^{\pp\Sigma})^*$,\footnote{Recall: $Q^{\int} \doteq \int_{-1}^1 du' Q(u')$.} 
\begin{subequations}
\begin{equation}\label{e:reducedfluxmap}
\langle \uhAS, (\xi_\i,\xi_\f) \rangle = \int_S \btr\Big( ({\cD}{}^i L_{\ell}\ASa{}_i)^{\int} \xi_\f  + \ASe\, ( \Ad(\Lambda)^{-1}\cdot\xi_\i  - \xi_\f)\Big),
\end{equation}
such that
\[
\iota_\C^* h = \cV^*\uhAS,
\]
where in this equality we left understood the inclusion $(\fg^{\pp\Sigma})^* \hookrightarrow  (\fg^\Sigma)^*_\mathrm{str}$. 
In particular, for $G$ Abelian, the previous expression reduces to:
\begin{equation}
    \label{eq:huAS-Ab}
\langle \uhAS, (\xi_\i,\xi_\f) \rangle = \int_S \btr\Big( (D^i\ASa_{i}^\dif) \xi_\f + \ASe\, (\xi_\i-\xi_\f) \Big)   
\qquad\mathrm{(Abelian)}.
\end{equation}
\end{subequations}
\end{proposition}

\begin{proof}
By direct computation we prove that the $G^{\pp\Sigma}_0$ action on the extended AS phase space $\XAS$ of Lemma \ref{lem:residualAction} is Hamiltonian  with momentum map given by the r.h.s. of Equation \eqref{e:reducedfluxmap}:
\begin{align*}
\bi_{\underline{\varrho}(\xi_\i,\xi_\f)} & \omeAS 
= \bi_{\underline{\varrho}(\xi_\i,\xi_\f)} (\omAS  + \Omega_S) \notag\\
& = \int_\Sigma \btr\Big( (L_\ell {\cD}_i\xi_\f ) \bd \ASa{}^i - (L_\ell \bd \ASa{}^i ) {\cD}_i \xi_\f \Big)\notag\\
& \qquad +\int_S \btr\Big( -[\xi_\i, \ASe ] \bd \Lambda \Lambda^{-1} + (\bd \ASe)\, ( \Lambda^{-1} \xi_\i\Lambda -\xi_\f) \Big) \notag\\
&\qquad+  \int_S \btr\Big(  \ASe\, \Big[ \Lambda^{-1} \xi_\i\Lambda  -\xi_f, \Lambda^{-1}\bd\Lambda\Big]\Big)  \notag\\
& = \int_\Sigma \btr\Big([L_{\ell}\ASa_i, \xi_\f] \bd \ASa{}^i - (L_\ell \bd \ASa{}^i) {\cD}_i \xi_\f \Big)\notag\\
&\qquad\int_S \btr\Big(  (\bd \ASe)  \Lambda^{-1}\cdot\xi_\i \Lambda + \ASe\, (\Lambda^{-1} \xi_\i \bd\Lambda -\Lambda^{-1}\bd \Lambda \Lambda^{-1}\xi_\i) - \bd \ASe \xi_\f\Big)  \notag\\
& = \bd \int_\Sigma \btr\Big( - (L_{\ell}\ASa{}^i) {\cD}_i \xi_\f \Big)
+  \bd \int_S \btr\Big(  \ASe\, (\Ad(\Lambda)^{-1}\cdot\xi_\i -  \xi_\f)\Big)  \notag\\
& = \bd \int_S \btr\Big( ({\cD}{}_i L_{\ell}\ASa{}^i)^{\int} \xi_\f  + \ASe\, (\Ad(\Lambda)^{-1}\cdot\xi_\i  - \xi_\f)\Big) , 
\end{align*}
where in the first step we used $\bi_{\underline{\varrho}(\xi_\i,\xi_\f)}\Lambda^{-1}\bd \Lambda = -\Lambda^{-1}\xi_\i \Lambda + \xi_\f$ (Lemma \ref{lem:residualAction}); in the second we used that $L_\ell\xi_\f =0 $ since $\xi_\f$ is $u$-independent; in the third we used the $\ad$-invariance of the trace viz. $\tr([L_\ell\ASa^i,\xi_\f]\bd \ASa_i) = -\tr(L_\ell\ASa^i [\ASa_i,\xi_\f])$; and finally, in the last step, we performed the integration in $u$ by remembering once again that $\xi_\f$ is independent of $u$.

We now compute explicitly $\cU^* \uhAS$ and verify that it equals $\iota_\C^*\h$. First, we note the following identities involving the curvature $F=F(A)$: first, we have
\begin{align*}
F_{\ell i}(A^{\dfU(A)}) = L_\ell \ASa(A)_i - D_i \ASa(A)_\ell + \tfrac12[\ASa(A)_\ell,\ASa(A)_i ] = L_\ell \ASa(A)_i,
\end{align*}
which follows from $(A^{\dfU(A)})_\ell = 0$ (Lemma \ref{lemma:dressing-properties-ii}); second, from the equivariance of the Gauss constraint $\mathsf{G} =  \cL_\ell E + \cD^iF_{\ell i}$, we obtain $E(A^{\dfU(A)},E_\i^{\dfU(A)^\i}) \equiv \Ad(\dfU(A)^{-1})\cdot E(A,E_\i) \equiv E^{\dfU(A)}$ as well as
\[
\cD^i_{\ASa(A)} L_\ell \ASa(A)_i =  \cD^i_{A^{\dfU(A)}} F_{\ell i}(A^{\dfU(A)}) \stackrel{\mathrm{(Gauss)}}{=}  -L_\ell E^{\dfU(A)} - [A^{\dfU(A)}_\ell , E^{\dfU(A)}] = - L_\ell E^{\dfU(A)}
\]
where $\ASa(A)\doteq A^{\dfU(A)}$ and $\cD_{\ASa} = D + [\ASa, \cdot]$; third, and last, using $\dfU(A)^\f =1$ and $\dfU(A)^\i \doteq\Lambda(A)$:
\begin{align*}
(\cD^i L_\ell \ASa(A)_i)^{\int} = -( L_\ell E^{\dfU(A)} )^{\int}
 &= -(E^{\dfU(A)})^\f + (E^{\dfU(A)})^\i \\ 
 &= -E(A,E_\i)^\f + \Ad(\Lambda(A)^{-1})\cdot E_\i.  
\end{align*}

This, together with the expression $(\cU^*\ASe)(A,E_\i) = \Ad(\Lambda(A)^{-1})\cdot E_\i$ allows us to compute:
\begin{align*}
\cU^*&\langle \uhAS, (\xi_\i,\xi_\f)\rangle
= \cU^*\int_S \btr\Big( ({\cD}{}_i L_{\ell}\ASa{}^i)^{\int} \xi_\f  + \ASe \big( \Ad(\Lambda^{-1})\cdot\xi_\i  - \xi_\f\big)\Big) \\
& = \int_S \btr\Big( \big(-E(A,E_\i)^\f + \Ad(\Lambda^{-1})\cdot E_\i\big) \xi_\f   \\ & \qquad \qquad + \big(\Ad(\Lambda(A)^{-1})\cdot E_\i\big) \big( \Ad(\Lambda(A)^{-1})\cdot\xi_\i  - \xi_\f\big)\Big) \\
& = - \int_S \btr\Big( E(A,E_\i)^\f \xi_\f - E_\i \xi_\i \Big)  = \langle \iota_\C h(A,E_\i), (\xi_\i,\xi_\f)\rangle.
\end{align*}
Finally it is immediate to verify that the Abelian formula descends from the general case, since then:
\[
(\cD^i L_\ell \ASa_i)^{\int} = (D^iL_\ell \ASa_i)^{\int} = D^i (L_\ell \ASa_i)^{\int} = D^i\ASa^\dif_i \qquad \mathrm{(Abelian)}
\]
(An alternative, straightforward, derivation of Equation \eqref{eq:huAS-Ab}, can be obtained using the AS modes by combining the results of Proposition \ref{prop:modedecomp} and Remark \ref{rmk:ASmodesGauge}.)
\end{proof}

\begin{remark}[Interpretation]
With reference to Equation \eqref{e:reducedfluxmap}, the fact that the contraction between $\ASe$ and $\xi_\f$ is mediated by the adjoint action of $\Lambda$ has an intuitive explanation: whereas $\ASe$ and $\xi_\i$ are quantities naturally defined on $S_\i$, the initial sphere, $\xi_\f$ is naturally defined on the final sphere $S_\f$ and therefore needs to be parallel-transported back to $S_\i$ by means of the appropriate Wilson line $\Lambda(A) \doteq \dfU(A)^\i = \Pexp \int_{-1}^1 du' \, A_\ell(u')$, so that it can be contracted with $\ASe$. 
\end{remark}

\begin{remark}[On-shell fluxes]\label{rmk:ASfluximage}
We note that a consequence of Proposition \ref{prop:red-flux-map}---in particular of the equation $\iota_\C^*h = \cV^*\uhAS$, together with surjectivity of $\cV$, is that
\[
\Im(\uhAS) \simeq \Im(\iota_\C^*h) \doteq \F \simeq
\begin{dcases}
(\fg^{\pp\Sigma})^* & \text{$G$ semisimple}\\
\Ann(\fg,(\fg^{\pp\Sigma})^*) & \text{$G$ Abelian}
\end{dcases}
\]
where the last isomorphism is proved in Proposition \ref{prop:fluxes}.
Therefore, there is a 1-to-1 correspondence between the values the momentum map $\uhAS$ can take and the on-shell fluxes: if $G$ is semisimple these correspond to all possible elements of $(\fg^{\pp\Sigma})^*$, whereas if $G$ is Abelian only the elements of $(\fg^{\pp\Sigma})^*$ compatible with the integrated Gauss's, viz. $\int_{\pp\Sigma} \boldsymbol{E} = 0$, are allowed.
\end{remark}

\subsection{Superselection sectors}
Now that we have established that there is a model Hamiltonian space $G^{\pp\Sigma}_0\circlearrowright(\XAS,\omeAS)$, reduction in the absence of a preferred value of the momentum map applied to the model requires\footnote{Alternatively, a symplectomorphic space is obtained by reducing at a fix value $f=(f_\i,f_\f)$ of the flux $\uhAS$ by the action of the \emph{stabiliser} $\mathrm{Stab}(f)\subset G^{\pp\Sigma}_0$. See the discussion of point- and orbit-reduction in Section \ref{sec:theoreticalframework}. \label{fnt:pointreduction}} one to first select a $G^{\pp\Sigma}_0$-coadjoint orbit in the image of the flux map $\uhAS$.

Since $G_0^{\pp\Sigma}\simeq G^S_0\times G^S_0$, a coadjoint orbit of $G_0^{\pp\Sigma}$ is determined by pairs of ``initial and final'' elements of $(\fg^S)^*$:
\[
\mathcal{O}_{(f_\i,f_\f)}= \mathcal{O}_{f_\i}\times \mathcal{O}_{f_\f}\subset (\fg^S)^*\times (\fg^S)^*.
\]
Moreover, in virtue of Proposition \ref{prop:red-flux-map}, the preimage along $\uhAS$ of a generic flux  $(f'_\i,f'_\f)$ is given by all the $(\ASa, \Lambda, \ASe) \in \XAS$ such that
\begin{equation}
    \label{eq:fluxes-sss}
\begin{dcases}
\langle f'_\i,\bullet\rangle  = 
\int_S \btr\left( (\Ad(\Lambda)\cdot \ASe) \,\bullet\,\right) \\
\langle f'_\f,\bullet\rangle  =  
\int_S \btr\left(  ( {\cD}^i L_\ell \ASa_i(A))^{\int} - \ASe) \,\bullet\, \right)  .
\end{dcases}
\end{equation}

In the Abelian case, the flux $f=(f_\i,f_\f)$ in the image of $\uhAS$ is actually an element of $\Ann(\fg,(\fg^{\pp\Sigma})^*) \subset (\fg^{\pp\Sigma})^*$, as it is readily verified using Equation \eqref{eq:huAS-Ab}.

\medskip
From the general theory of Hamiltonian reduction, the preimage in $\XAS$ of the orbits $\mathcal{O}_{f_\i}\times \mathcal{O}_{f_\f}= \mathcal{O}_{(f_\i,f_\f)}$, modulo the action of $G^{\pp\Sigma}_0 \simeq G^S_0 \times G^S_0$, is a symplectic manifold (when smooth). In view of this observation we are going to describe the symplectic leaves (superselection sectors) of ${\uuC}$ through the symplectic leaves of the Hamiltonian reduction $\XAS/G^{\pp\Sigma}_0$, which we call ``Ashtekar--Streubel sectors'', by means of a symplectomorphism.
\begin{definition}\label{def:ASSectors}
The \emph{Ashtekar--Streubel sector} associated to $\mathcal{O}_{(f_\i,f_\f)}$ is the sympelctic manifold $\uuS^{\mathsf{eAS}}_{[f_\i,f_\f]}$:
\[
\uuS^{\mathsf{eAS}}_{[f_\i,f_\f]} \doteq \uS^{\mathsf{eAS}}_{[f_\i,f_\f]}/G^{\pp\Sigma}_0, 
\]
where
\begin{align*}
\uS^{\mathsf{eAS}}_{[f_\i,f_\f]} &\doteq \uhAS^{-1}(\mathcal{O}_{(f_\i,f_\f)}) \\
&=  \big\{ (\ASa, \Lambda, \ASe) \ : \ \text{Eq. \eqref{eq:fluxes-sss} holds for some $(f'_\i,f'_\f)\in\mathcal{O}_{(f_\i,f_\f)}$} \big\}.
\qedhere
\end{align*}
\end{definition}

\begin{theorem}\label{thm:modelHamreduction}

The (Poisson) reduction of $(\uCo,\uomegao)$ by the action of $\Gred \doteq \G/\Go$, described in Section \ref{sec:theoreticalframework}, is diffeomorphic to the (Poisson) reduction of $(\XAS, \omeAS)$ by the action of $G^{\pp\Sigma}_0$, i.e.
\[
\uCo/\Gred \simeq \XAS/G^{\pp\Sigma}_0
\]
as Poisson manifolds. 
In particular, the symplectic leaves of $\uCo/\Gred$ are symplectomorphic to the Ashtekar--Streubel sectors, i.e.\ the fully-reduced phase space ${\uuC}\doteq \C/\G\simeq\uCo/\Gred$ decomposes as
\[
{\uuC} \simeq \bigsqcup_{\mathcal{O}_{(f_\i,f_\f)}} 
\uuS^{\mathsf{eAS}}_{[f_\i,f_\f]},
\]
where $\F\simeq (\fg^{\pp\Sigma})^*$ if $G$ is semisimple, and $\F\simeq \Ann(\fg, \fg^{\pp\Sigma})^*)$ if $G$ is Abelian (and in both cases, $\F\simeq \fGred^*$).
\end{theorem}

\begin{proof}
    Recall that $(\uCo,\uomegao)\stackrel{{p_V}}{\to}\XAS$ is a symplectic covering of the AS extended phase space with fibre $\mathcal{K}$ (Theorem \ref{thm:nonAb-constr-red}), and that $\Gred$ is a connected central extension of $G^{\pp\Sigma}_0$ by the discrete group $\mathcal{K}$ (Proposition \ref{prop:Gred}). This, together with the fact that the flux map $\uh$ is $\mathcal{K}$-invariant (Lemma \ref{lem:flux-K-invariance}), implies that, given $f\in\fGred^*$ and $\mathcal{O}_f$ its coadjoint orbit, $\uh^{-1}(\mathcal{O}_f)={p_V}^{-1}\circ\uh_{AS}^{-1}(\mathcal{O}_f)$, and 
    \[
    \uh^{-1}(\mathcal{O}_f)/\Gred={p_V}^{-1}\circ\uh_{AS}^{-1}(\mathcal{O}_f)/\Gred \simeq ({p_V}^{-1}\circ\uh_{AS}^{-1}(\mathcal{O}_f)/\mathcal{K})/G^{\pp\Sigma}_0\simeq\uh_{AS}^{-1}(\mathcal{O}_f)/G^{\pp\Sigma}_0,
    \]
    where we used the fact that $\fGred^*\simeq(\fg^{\pp\Sigma})^*$ and $\F=\Im(\uh) \simeq \Im(\uhAS)$ (Remark \ref{rmk:ASfluximage}) to view a flux $f\in\F$ simultaneously in the image of $\uh$ and of $\uhAS$. The fact that the coadjoint orbits of an element $f$ along the action of $\Gred$ and $G_0^{\pp\Sigma}$ coincide follows from the fact that $\Gred$ is a central extension of $G_0^{\pp\Sigma}$, i.e.\ $\mathcal{K}$ acts trivially on $\F$.

    Since all orbit-reductions of the space $\uCo$ are symplectomorphic, the reduction ${\uuC}=\uCo/\Gred$ is diffeomorphic (as a Poisson manifold) to $\XAS/G_0^{\pp\Sigma}$, concluding the proof. 
\end{proof}

\section{Asymptotic symmetries and memory as superselection}\label{sec:alternatelabelmemory}

\noindent\fbox{%
    \parbox{\textwidth}{\small\center We invite the reader to follow the summary of notations and concepts introduced up to here, available in Appendix \ref{app:notations}, for quick reference.}}

\bigskip

In this section we discuss asymptotic symmetries of QED \cite{HeMitra2014,KapecPateStrominger,StromingerLectureNotes} and electromagnetic memory \cite{Staruszkiewicz:1981,Herdegen95,BieriGarfinkle,GarfinkleHollandsetal,Pasterskimemory} in the light of our results. Our work is tailored to (and rigorous for) the treatment of finite-distance boundaries.  Since YM theory in 4 spacetime dimension is (classically) conformally invariant, one can take our results as a model for asymptotic case when $\dim\Sigma=3$, up to a discussion of equivalence classes of embeddings $\Sigma\hookrightarrow\scri$.

With this in mind, we henceforth interpret $\Sigma = I \times S$ as a compact subset of (past) null infinity $\scri = \mathbb{R}\times S$,  i.e.
\[
\Sigma \subset \scri.
\]
with $S\simeq S^2$ the ``celestial sphere''.
The reason for choosing \emph{past} asymptotic null infinity will be explained shortly.\footnote{We apologise for the ensuing mislabeling of the advanced time coordinate, here denoted $u$ rather than the conventional $v$.}

Note that in this context $A$ and $\ASa$ become (spatial) connection one-forms over $\Sigma\subset\scri$, while the ``electric field'' $E$ and ``magnetic field'' $F$ stand for the $r^{-2}$-order of the $(ur)$- and $(ij)$-component of the spacetime field strength $\bar F_{\mu\nu}$ in Bondi coordinates, i.e.\ 
\[
E = \lim_{r\to\infty} r^2 \bar F_{ru}
\quad\text{and}\quad
F_{ij} = \lim_{r\to\infty} r^2 \bar F_{ij}.
\]

\begin{remark}[Caveat]
Our treatment of the asymptotic limit is equivalent to working in a set of fixed Bondi coordinates with the embedding $S_{\i/\f}\hookrightarrow\scri$ defined by cuts of constant $u$. 
This is tantamount to taking a naive $r\to \infty$ limit, as done in the standard reference \cite{StromingerLectureNotes} and much of the related literature (see, however, \cite{Ashtekar1984,AshtekarCmapigliaLaddha2018,LaddhaHolographicnullinfinity}).
Therefore, ``asymptotic infinity" is here understood as the boundary of \emph{one} conformal compactification of spacetime. That is, we are not implementing the notion of equivalence classes of different conformal compactifications that plays a central role in the geometric approach to conformal infinity \cite{Penrose1965,penrose1984spinors, Geroch1977}, nor do we study the dependence of our results from the choice of ``cut'' of $\scri$, i.e. of the embedding of $\Sigma \hookrightarrow \scri$. 
We leave the investigation of these important questions to future work.
\end{remark}

\begin{assumption}
In this section we fix $\Sigma$ to be a 3d null cylinder with the (celestial) 2-sphere as a base:
\[
\dim\Sigma \doteq n = 3,\qquad \pp\Sigma = S_\i\times S_\f, \qquad  S_{\i/\f}\simeq S^2.
\]
Hence, $G^{\pp\Sigma}\simeq G^S_\i\times G^S_\f$, where $G^S_{\i/\f} \doteq C^\infty(S_{\i/\f},G)$, and, since $\pi_2(G) = 0$ for all Lie group \cite{MIMURA}, these mapping groups are necessarily connected, i.e.\ $G^S_0 = G^S$. 
\end{assumption}

\subsection{$G$ Abelian}
Consider first Maxwell theory, which is a particular case of the general results of Theorems \ref{thm:nonAb-constr-red}, \ref{thm:fibrecharacterisation}, and \ref{thm:modelHamreduction} (see also Lemma \ref{lem:residualAction}, Proposition \ref{prop:red-flux-map}, and Appendix \ref{app:Abelianreduction})---which we summarise here for ease and are proven in Sections \ref{sec:reduction-first} and \ref{sec:reduction-second} (see also Appendix \ref{app:Abelianreduction}):

\begin{theorem}[Maxwell: Asymptotic constraint reduction]\label{thm:summaryMax}~
Let $G=\mathrm{U}(1)$, $\dim\Sigma\doteq n=3$, and $S\simeq S^2$. Then:
\begin{enumerate}
\item The constraint gauge group $\Go\subset \G$ is given by the identity component of the relative mapping group on $\Sigma$, $\Go \simeq G^\Sigma_{\rel,0}$. The constraint-reduced phase space $(\uCo,\uomegao) \doteq (\X,\omega)//_{0}\Go$ at $\Sigma=\scri$ is symplectomorphic to the linearly extended Ashtekar-Streubel phase space $(\Xas,\omeas)$:
\[
\Xas \doteq \hcalA \times \fg^S \times \fg^S \ni (\ASa , \lambda ,\ASe) 
\]
with 
\[
\omeas(\ASa,\lambda,\ASe) = \int_{\Sigma} \sqrt{\gamma} \ \gamma^{ij} (L_\ell \bd \ASa_i) \wedge \bd \ASa_j  + \int_S \sqrt{\gamma} \ \bd \ASe \wedge \bd \lambda,
\]
where, owing to the Gauss constraint,
the asymptotic (gauge-invariant) electromagnetic field $(E,F)$ at $(u,x)\in\Sigma\subset \scri$ is given by  
\[
    E = \ASe + D^i \ASa_i^\i- D^i\ASa_i(u) 
        \quad\text{and}\quad
     F = D \ASa.
\]
In particular, $E$ and $F$ are independent of $\lambda$, while the initial and final electric fields at $\pp\Sigma$ are given by:
\[
E^\i = \ASe \quad \text{and}\quad E^\f = \ASe - D^i \ASa_i^\dif.
\]
\item \label{thm:summaryMax-2} The residual flux gauge group $\Gred \doteq \G/\Go$ is a connected central extension of the identity component of the boundary mapping group $G^{\pp\Sigma}/G\simeq ({G^S_\i}\times {G^S_\f})/G$ by $\mathbb{Z}$:
\[
1\to \mathbb{Z} \to \Gred \to (G^S_\i \times G^S_\f)/G
\]
The constraint reduced phase space $(\uCo,\uomegao)\simeq (\Xas,\omeas)$ carries the Hamiltonian action of $\fGred\simeq (\fg^S_\i\times \fg^S_\f)/\fg \ni (\xi_\i,\xi_\f)$,
\[
\underline{\varrho}(\xi_\i,\xi_\f)\begin{pmatrix}
    \ASa\\ \lambda\\ \ASe
\end{pmatrix} = \begin{pmatrix}
    D\xi_\f \\ \xi_\f - \xi_\i\\ 0
\end{pmatrix},
\]
with momentum map
\[
\langle \uhAS,(\xi_\i,\xi_\f) \rangle = \int_S \sqrt{\gamma}\ \big( (D^i\ASa_i^\dif) \xi_\f - \ASe (\xi_\f - \xi_\i) \big).
\]
\item \label{thm:summaryMax-3} The fully reduced phase space ${\uuC} \doteq \uCo/\Gred$ is a Poisson space foliated by symplectic leaves $\uuS_{[f_\i,f_\f]}$, called flux superselection sectors, labelled by pair of initial and final fluxes $(f_\i, f_\f) \in \Ann(\fg,(\fg^{\pp\Sigma})^*)$, and defined as
\begin{align*}
\uuS_{[f_\i,f_\f]} & \doteq \uh_\mathsf{AS}^{-1}( f_\i,f_\f)/\Gred\\ 
&\simeq \big\{ (\ASa,\ASe)\in \hcalA \times \fg^S \ \text{that\ satisfy \ equation\ \eqref{eq:MaxSSS} \ below} \big\} \big/ G^S_\f
\end{align*}
viz.
\begin{equation}
    \label{eq:MaxSSS}
\begin{dcases}
    \langle f_\i ,\ \cdot\ \rangle = \int_S \sqrt{\gamma} \ \big(\ASe \ \cdot \ \big)\\
    \langle f_\f ,\ \cdot\ \rangle = \int_S \sqrt{\gamma} \ \big((\ASe - D^i\ASa_i^\dif) \ \cdot \ ).
\end{dcases}
\end{equation}
All superselection sectors are symplectomorphic to each other.\footnote{This is analogous to what happens in Maxwell theory when $\Sigma$ spacelike---cf. Appendix \ref{app:Maxwell-short}.}
\end{enumerate}
\end{theorem}

\begin{remark}
If instead of $G=\mathrm{U}(1)$ one chooses $G=\mathbb{R}$, the only difference in the above theorem would be in point (2), where $\Gred \simeq G^{\pp\Sigma}/G$---with no extension.
\end{remark}

\subsubsection{Memory as superselection}

In the absence of massive, charged, particles, one can restrict the attention to the case in which the electric field at $\scri^-_-$ (and $\scri^+_+$) vanishes. If we wanted to implement such a condition in our formalism we would have to require that $E^\i$ vanishes. (Note that since $\Sigma$ must be compact, $E^\i$ is necessarily computed at a finite value of advanced time albeit possibly one ``very far into the past''.) 
This corresponds to performing a partial flux-superselection at which only the value of $E^\i = \ASe$ is fixed---at the special value of zero. Mathematically, this is yet another reduction in stages: we first perform symplectic reduction at the zero level set of the (component of the) momentum map $\int_S\sqrt{\gamma}\ (\ASe \ \cdot \ )$ for the action of a copy of $G^{S}_\i$. 
As a result of this reduction, not only is $\ASe$ fixed to zero, but $\lambda$ is quotiented out as well.
Then, and only then, one is left with a symplectic manifold given by the (\emph{non}-extended) Ashtekar-Streubel phase space \cite{AshtekarStreubel}:

\begin{proposition}\label{prop:MaxPartialSSS}
The symplectic reduction of $(\uCo,\uomegao,\Gred)\simeq (\Xas, \omeas, G^{\pp\Sigma} /G)$, with respect to the Hamiltonian action of the initial copy of $G^S_\i\subset G^{\pp\Sigma}$ at $f_\i = 0$ (i.e.\ $\ASe=0$), yields the Ashtekar--Streubel symplectic space $(\hcalA,\omAS)$,
\[
(\Xas,\omeas) //_{0} G^S_\i \simeq (\hcalA,\omAS).
\]
This space carries the following residual action of the gauge symmetry group $G^S/G$:
\[
\underline{\varrho}{}_\mathsf{AS}(\xi_\f) \ASa = D\xi_\f.
\]
with momentum map
\[
\langle \uh_{\mathsf{AS}},\xi_\f\rangle = \int_S \sqrt{\gamma} \ \big( (D^i\ASa_i^\dif)  \xi_\f  \big).
\]
The (asymptotic, gauge-invariant) electromagnetic field $(E,F)$ at $(u,x)\in\Sigma\subset \scri$ is then given by   
\[
E =  D^i \ASa_i^\i- D^i\ASa_i(u) 
\quad\text{and}\quad
F = D \ASa 
\qquad (\ASe = 0).
\]
\end{proposition}

\begin{remark}
Observe that this corresponds to ``gauge-fixing'' $A_\ell=0$. The combination of the two stages is tantamount to symplectic reduction w.r.t. $\Go$ plus the (superselection) condition  $\E_\i = 0$. This can be thought of symplectic reduction with respect to the subgroup of gauge transformations that are trivial at $S_\f$ at the on-shell configurations such that $E_\i=0$.
\end{remark}

This \emph{partially superselected phase space} is the one that directly compares to the setup of \cite{HeMitra2014, KapecPateStrominger, StromingerLectureNotes}.\footnote{To avoid clutter, we will only keep citing Strominger's lecture notes.} What in their language is a ``new symmetry of QED'' is here rather found to be a \emph{residual gauge symmetry} of QED that still acts on the AS phase space.\footnote{With reference to the reduction by stages procedure, quotienting out this residual gauge symmetry is interpreted as ``third stage'' reduction.}. We notice that this symmetry is associated to $S_\f\subset\pp\Sigma$, which (morally) corresponds to $\scri^-_+$: which is precisely what one would expect according to the analysis of \cite{StromingerLectureNotes}.

Moreover, this residual action of $G^S_\f$ on $(\hcalA,\omAS)$ is Hamiltonian, and the momentum map is given by the sphere-divergence of the AS zero mode $\mathfrak{Im}(2\tilde \ASa(0)) = \ASa^\dif$ (cf. Proposition \ref{prop:modedecomp}):
\[
\ASa \mapsto \int_S \sqrt{\gamma}\big((D^i\ASa_i^\dif) \ \cdot \ ).
\] 
This is of course in agreement with \cite{StromingerLectureNotes}, where \emph{additionally} $\ASa_i^{\i/\f}$ is assumed to be $D$-exact (absence of magnetic fluxes at $\scri^-_\pm$). We see that no such restriction is in fact necessary.

We also observe that $\ASa^\dif \doteq \ASa^\f - \ASa^\i$ is invariant under the residual gauge symmetry at $S_\f$. In particular, its divergence equals (minus) the difference between the initial and final electric fields (see e.g. Equation \eqref{eq:EdiffAbelian}):
\[
E^\dif \equiv E^\f(\ASa,\ASe) - E^\i(\ASe) = - D^i\ASa_i^\dif.
\]
This is the (so-called ``ordinary'') \emph{electromagnetic memory}, which plays an important role in the \emph{memory effect,} first identified by \cite{Staruszkiewicz:1981} and confirmed e.g.\ by \cite{Herdegen95}, and later revisited by \cite{BieriGarfinkle} in the context of particle scattering off a burst of electromagnetic radiation, as well as \cite{Pasterskimemory} in the context of asymptotic symmetries and Weinberg's soft theorems (see also \cite{GarfinkleHollandsetal}). (Note that recent controversy on such interpretations was brought to light by \cite{herdegen2023velocity}.)\footnote{Note:
$E\vert_\text{here} = \lim_{r\to\infty} r^2 \bar F_{ur}$ is denoted by $
W\vert_\text{there}$ in \cite{BieriGarfinkle} and by $A_u\vert_\text{there}$ in \cite{Pasterskimemory}}

It is often convenient to express the memory in term of the quantity \cite{BieriGarfinkle}\footnote{With reference to \cite{BieriGarfinkle}, $\Phi\vert_{\text{there}}=\mu\vert_{\text{here}}$.} 
\begin{equation}
\label{eq:memoryBG}
\mu \doteq \Delta^{-1} E^\dif = - \Delta^{-1} D^i\ASa_i^\dif \in \Ann(\fg,(\fg^S)^*),
\end{equation}
where $\Delta$ is the Laplacian on $S$. If $\ASa^{\i/\f} = D\varphi^{\i/\f}$, as is assumed in \cite{KapecPateStrominger,StromingerLectureNotes,Pasterskimemory}, then
\[
\mu= -\Delta^{-1} \Delta\varphi^\dif = - \varphi^\dif.
\]
We will come back to this equation shortly. The take away message is that ``memory is a (co-) momentum map'' on the  (partially superselected) AS phase space.

(In both \cite{BieriGarfinkle} and \cite{Pasterskimemory}, an extra contribution is present due to the flux of particles through $\Sigma \hookrightarrow M$, which they call ``null memory'' and ``hard contribution'' respectively. This term does not appear here because we have worked in the absence of matter fields. Including these contributions would be straightforward and would yield results in agreement with theirs.)

As we already observed, the (residual) \emph{gauge} symmetry $G^S_\f$ is generated by the charge/momentum map $D^i\ASa^\dif_i$. Using the nonlocal basis of $\hcalA$ given by the AS modes introduced of Proposition \ref{prop:modedecomp}, it is then immediate to see that only the AS zero mode $\mathfrak{Re}(2\tilde \ASa(0)) = \ASa^{\int}-\ASa^\av$ is affected by this action:
\begin{equation*}
\underline{\varrho}{}_\mathsf{AS}(\xi_\f) \begin{pmatrix}
\mathfrak{Re}(2\tilde \ASa(0))\\
\mathfrak{Im}(2\tilde \ASa(0)) 
\end{pmatrix}
= 
\begin{pmatrix}
 D\xi_\f \\ 0
\end{pmatrix}
\quad\text{and}\quad
\underline{\varrho}{}_\mathsf{AS}(\xi_\f) \tilde \ASa(k\geq 1) = 0.
\end{equation*}
We are going to explicitly describe the fully reduced phase space ${\uuC}$ by means of a suitable gauge-fixing for this action.

Consider the Hodge decomposition of the AS modes $\tilde\ASa(k,x)$:
\begin{subequations}\label{eq:HodgeASa}
\begin{equation}
\tilde\ASa(k) = D \tilde\varphi(k) +  D \times \tilde\beta(k),
\end{equation}
where $D\equiv d_S$ is the differential on the sphere and $D \times \doteq \star D \star \equiv d_S^\star$ denotes the codifferential (a.k.a.\ curl); we also used the fact that no non-trivial harmonic 1-form exists on the two sphere. 

In this decomposition, the $D$-exact (a.k.a.\ electric) zero-mode $\mathfrak{Re}(2 \tilde\varphi(0))$ of $\tilde\ASa(k)$ is ``pure gauge'':\footnote{Mathematically, this means that $\fg^S$ acts transitively on it, so that the quotient is a point, meaning that $\mathfrak{Re}(2\tilde\varphi(0))$ is eliminated by reduction.}
\begin{equation}\label{eq:MaxResidualGauge}
\underline{\varrho}{}_\mathsf{AS}(\xi_\f) \mathfrak{Re}(2\tilde\varphi(0)) = \xi_\f,
\end{equation}
whereas all other modes are gauge invariant, including $\mathfrak{Im}(2\tilde\varphi(0))$. This quantity parametrises, in this language, the momentum map $\uh_\mathsf{AS}$, since
\begin{equation}\label{eq:MaxResidualGauge2}
\mathfrak{Im}( 2\tilde\varphi(0) ) = \varphi^\dif = \Delta^{-1} D^i\ASa^\dif_i = - \mu.
\end{equation}
\end{subequations}

From this we conclude:

\begin{theorem}[Maxwell: memory as superselection]\label{thm:memorySSS}
The flux superselection sector $(\uuS_{[f_\i,f_\f]},\uuomegao_{[f_\i,f_\f]})$ associated with the flux $f=(f_\i,f_\f) \in \Ann(\fg,(\fg^{\pp\Sigma})^*)$ with
\[
\langle f_\i, \cdot \rangle = 0 \quad \text{and}\quad \langle f_\f, \cdot \rangle = \int_S\sqrt{\gamma}( \mu \Delta \cdot ) \in \Ann(\fg,(\fg^S)^*)
\]
is symplectomorphic to the symplectic space $(\hcalA_\mu, \omAS_\mu)$ defined by
\[
\hcalA_\mu \doteq \{ \ASa\in\hcalA \ : \ 2 D^i\tilde\ASa_i(0) = i \Delta \mu \}
\]
and
\[
\omAS_\mu = \int_S \sqrt{\gamma}\Big(  \star \tilde\beta(0)^* \wedge \Delta \bd \tilde\beta(0) + \sum_{k\geq 1} \bd \tilde\ASa(k)^*_i \wedge \bd \tilde \ASa(k)^i \Big)
\]
where $\Delta \equiv d_S d_S^\star + d_S^\star d_S$.
The asymptotic (gauge-invariant) electromagnetic field $(E,F)$ is at $(u,x)\in\Sigma\subset \scri$ is then given by   
\[
E =  \Delta(\phi_{k\geq1}(u)-\phi_{k\geq 1}^\i)
\quad\text{and}\quad
F = \Delta \big( \mathfrak{Re}(2 \beta) + \frac{u}{2} \mathfrak{Im}(2\beta)\big), 
\]
where $\varphi_{k\geq1} \doteq \varphi - \mathfrak{Re}(2 \tilde\varphi(0)) - \frac{u}{2} \mathfrak{Im}(2\tilde\varphi(0))$, so that $\varphi_{k\geq1}^\dif = \varphi_{k\geq1}^{\int} - \varphi_{k\geq1}^\av = 0$.
\end{theorem}

It is not hard to gather that all superselection sectors are in fact symplectomorphic to each other, i.e.\ to $(\hcalA_\mu, \omAS_\mu)$, even when the restriction $0 = \langle f_\i, \cdot \rangle = \int_S \sqrt{\gamma}(\ASe \cdot)$ is lifted.

The electromagnetic memory $\mu$ is a viable superselection label because $-D^i\ASa_i^\dif$ is the momentum map for the diagonal subgroup $G^S \hookrightarrow G^S_\i \times G^S_\f$, $g \mapsto(g_\i,g_\f) = (g,g)$. (See item (\ref{thm:summaryMax-2}) of Theorem \ref{thm:summaryMax}.) However, this fact is a consequence of the \emph{Abelian} nature of Maxwell theory, which fails to have a non-Abelian analogue. Moreover, the diagonal subgroup is not normal. If one superselects the memory \emph{before} superselecting the initial electric field, one would obtain a symplectic space without either a residual \emph{group} action or a momentum map.

\begin{remark}[Soft Symmetries]
Another way to see this is by noticing that we have the group homomorphism 
\[
G^{\pp\Sigma} \to G^S_\mathrm{diag} \ltimes_\mathrm{AD} G^S_\dif, \quad 
(g_\f,g_\i)\mapsto(g_\mathrm{diag},g_\dif)\doteq (g_\f,g_\f g_\i^{-1}),
\]
where the multiplication in $G^S_\dif$ is from \emph{the right}:
\[
(g_\mathrm{diag},g_\dif)\cdot(h_\mathrm{diag}, h_\dif)  = \Big(g_\mathrm{diag} h_\mathrm{diag} , g_\mathrm{diag} h_\dif g_\mathrm{diag}^{-1} g_\dif \Big).
\]
At the infinitesimal level, we have that the pairing between $\fg^{\pp\Sigma}$ and its dual can be rewritten as
\begin{align*}
\langle f,\xi\rangle &= \langle f_\f, \xi_\f\rangle - \langle f_\i, \xi_\i\rangle\\ 
& = \langle f_\f-f_\i, \xi_\f \rangle + \langle f_\i, \xi_\f-\xi_\i\rangle \\
& \equiv \langle f_\dif, \xi_\mathrm{diag} \rangle + \langle f_\i, \xi_\dif \rangle.
\end{align*}
where $f_\dif \doteq f_\f - f_\i \in (\fg^S)^*$ and $(\xi_\mathrm{diag}, \xi_\dif) \doteq (\xi_\f, \xi_\f - \xi_\i)\in\fg_\mathrm{diag}^S \ltimes_\ad \fg_\dif^S$.

\emph{In the Abelian case}, The momentum map $\uhAS$ of Theorem \ref{thm:summaryMax}(\ref{thm:summaryMax-3}) splits into its $(\fg^S)_\i^*$ and $(\fg^S)_\dif^*$ components as follows:
\begin{equation}
\label{eq:SSSPhi}
\begin{dcases}
    \langle f_\i ,\ \cdot\ \rangle = \int_S \sqrt{\gamma} \ \big(\ASe \ \cdot \ \big),\\
    \langle f_\dif ,\ \cdot\ \rangle = \int_S \sqrt{\gamma} (\mu \Delta\ \cdot \ )  = \int_S \sqrt{\gamma} \ \big(( - D^i\ASa_i^\dif) \ \cdot \ ).
\end{dcases}
\end{equation}
Note that this is \emph{not} the case for $G$ non-Abelian, owing to the (field-dependent) ``parallel transport'' $\Ad(\Lambda^{-1})\cdot\xi_\i - \xi_\f$ (cf. Eq. \eqref{e:reducedfluxmap}).

The subgroup $G^S_\dif\simeq \{1\}\times G^S_\dif \into G^{\pp\Sigma}$ is normal. The quotient group 
\[
G^S_\so \doteq (G^{\pp\Sigma}/G^S_\dif)/G\simeq G^S_{\mathrm{diag}}/G,
\]
is the gauge group of \emph{soft symmetries}---this can be naturally identified with the diagonal subgroup of $G^{\pp\Sigma}$ (modulo constant gauge transformations). After reduction in stages by the action of $G^S_\dif$, we are left with a Hamiltonian action of $G^S_\so$ on $\AS$, with momentum map given by the electromagnetic memory.
\end{remark}

\subsubsection{The physical meaning of superselection}\label{sec:physicalinterpretation}

We take a passage of the standard reference \cite{StromingerLectureNotes} as a starting point for a series of remarks on superselection.
In \cite[Section 2.11]{StromingerLectureNotes} a parallel is drawn between $\varphi$, that is the electric part of $\ASa_i$, and a Goldstone boson arising in the presence of a broken (large) symmetry (italisation and text within square brackets is our own):

\begin{quotation}\small
We have a \emph{charge} $Q^+_\varepsilon$[$:= \int_S(\sqrt{\gamma}\ (-D^i\ASa_i^\dif)\varepsilon)$] that generates a symmetry of the Lagrangian of any Abelian gauge theory. However, this charge does not annihilate the vacuum [$\ASa_i = 0$]. Instead, it creates an extra \emph{soft photon mode} $\phi$ [our $\mathfrak{Re}(2\tilde\varphi(0))$], which, according to (2.7.3), transforms inhomogeneously under the broken symmetry [see our Equation \eqref{eq:MaxResidualGauge}]. Hence the soft photons are the Goldstone bosons of spontaneously broken \emph{large gauge symmetry} [our $G^S_\f$]. There is an infinite vacuum degeneracy, since we can add any number of soft photons to any vacuum state and obtain another vacuum state with the same zero energy. Classically, the infinite-dimensional space of vacua can be labeled by flat Abelian connections $\pp_z\varepsilon(z,\bar z)$ on the sphere [$D\varepsilon$ in our notation].

There is a crucial difference between the usual Mexican hat story of spontaneous global symmetry breaking and the spontaneous breaking of the large gauge symmetries. In the usual story, the different vacua form superselection sectors (i.e., no physical finite energy operator exists that can move us from one vacuum to another). 

[...] Such superselection sectors clearly do not arise for the large gauge symmetry. The vacuum state is changed by soft photon creation, which occurs in nearly all scattering processes. The S-matrix elements do not factorise into superselection sectors.
\end{quotation}

The dictionary between the terminology used in \cite{StromingerLectureNotes} and our own is the following (to simplify the discussion and be consistent with their choice of phase space, i.e.\ $(\hcalA,\omAS)$, here we assume that we are working in the partial superselection sector: $f_\i = \int_S \sqrt{\gamma} (\mathsf{e}\, \cdot\,)=0$):
\begin{enumerate}[label=\roman*.]
\item the ``charge'' $Q^+_\varepsilon$ is the Hamiltonian generator (momentum map) for $G^S_\f$, i.e.\ $Q^+_\varepsilon := \int_S(\sqrt{\gamma}\ (-D^i\ASa_i^\dif)\varepsilon) = \int_S \sqrt{\gamma} \ E^\f \varepsilon$;
\item the expression ``large [or, asymptotic] gauge symmetries'' (LGS) is defined in \cite[Section 2.10]{StromingerLectureNotes} as the ``allowed gauge symmetries'' (AGS) modulo the ``trivial gauge symmetries" (TGS). The AGS are ``any [gauge symmetries] that respect the boundary [and fall-off] conditions'', while the TGS are ``the ones that act trivially on the physical data of the theory''. In \cite{StromingerLectureNotes}, the state space of the theory is taken to be $(\hcalA,\omAS)$, which we recognise in our language as the constraint-reduced, partially-superselected, space of Proposition \ref{prop:MaxPartialSSS}. On this space there is a residual action of $G^S_\f$. Therefore we take LGT to correspond in our language to the residual gauge symmetry $G^S_\f$ that survives both constraint reduction and the superselection of the initial flux $f_\i=0$.\footnote{More precisely the corrected residual symmetry group is $G^S_\f/G$, but since $G\hookrightarrow G^S_\f$ acts trivially on $\hcalA$, the quotient can be dropped without altering the result.}
This identification of LGS and $G^S_\f$ is consistent with the identification of the associated charge;\footnote{We do not attempt to formalise, in our language, the idea of the quotient LGS = AGS / TGS. But loosely speaking the TGS include spacetime gauge transformations that are trivial at $\scri$, gauge transformations on $\Sigma\subset \scri$ that are trivial at $\pp\Sigma$ (our $\Go$), and gauge transformations at $\pp\Sigma$ that are trivial at $S_\f\subset \pp\Sigma$ (our $G^S_\i$).}
it    \item In \cite{StromingerLectureNotes} both the following quantities are called ``soft photon modes"
\[
\begin{dcases}
N|_\text{there} = \mathfrak{Im}(2\tilde\varphi(0)) = \varphi^\dif = - \mu,\\
\phi|_\text{there} = \mathfrak{Re}(2\tilde\varphi(0)) = \varphi^{\int}-\varphi^\av.
\end{dcases}
\] 
There, this is justified because they are seen as conjugate fields on $(\hcalA,\omAS)$. Here, however, we see that their physical status is different: while $\phi$ is ``pure gauge'' (Eq. \eqref{eq:MaxResidualGauge2}), $N$ is gauge invariant and encodes memory. (As noted in Remark \ref{rmk:Stromingercf} there is a slight discrepancy in the way we and \cite{StromingerLectureNotes} extract the ``zero (or, soft) mode'' from a given asymptotic field.)
\end{enumerate}

\begin{remark}
We can now compare and contrast our results and the picture summarised in the above excerpt. In this comparison, it is important, however, to keep in mind Strominger's warning \cite[Section 2.11]{StromingerLectureNotes}, the italic is ours: 
\begin{quotation}
\small {Large gauge symmetries are unlike any previously discussed symmetry both in their asymptotic angle dependence and in the fact that the action is described at null, not spatial, infinity.} 

Phrases like `spontaneous symmetry breaking,' `Goldstone boson,' `superselection sector,' and even `conservation law' are used with slightly different meanings in different physical contexts. In importing those words to the present context, I have necessarily adapted and refined their meanings. \emph{I have done so in the way I thought most natural, but other adaptations might be possible.}
\end{quotation}
We clearly do not agree with the first assertion in the previous quotation. As for the second: The same cautionary measures apply to our analysis.
\end{remark}

\begin{enumerate}[label=\arabic*.,leftmargin=*]
\item Our reduction by stages procedure is manifestly gauge invariant.
Although the choice of describing the superselection sector $\uuS_{[f_\i,f_\f]} \simeq \hcalA_\mu$ in terms of AS fields with $\mathfrak{Re}(2\tilde\varphi(0))=0$ appears to be breaking the residual gauge symmetry, this is nothing but a choice of gauge-fixing, akin to setting $A_\ell$ to a constant in $u$ to gauge-fix $\Go\circlearrowright \X$. 

Now, the ``gauge-fixed" space $\hcalA_\mu$ carries no residual action of $G^S_\f$. Conversely, before gauge fixing,  all AS fields $\ASa\in\hcalA$ that are related by a gauge transformation in $G^S_\f$ represent the same physical fields (Proposition \ref{prop:MaxPartialSSS}). In fact $\hcalA_\mu = \{ \ASa \ : \ 2\tilde\varphi(0) =  i \mu \}$ intersects all the gauge orbits of those configurations with fixed memory $\mu$ once and only once.\footnote{This follows from the global nature of the Hodge decomposition in $\Omega^1(S)$, Assumption \ref{ass:trivialP}.} Another way of saying this is that the condition $\mathfrak{Re}(2\tilde\varphi(0)) = 0$ specifies a (global!) section of the principal bundle $\hcalA \to \hcalA/G^S_\f$, which is itself foliated into symplectic spaces $\{\hcalA_\mu\}_{\mu\in \Ann(\fg,(\fg^S)^*)}$ of fixed memory.
\item In light of the previous remark, we don't see how to make the ``Goldstone boson'' interpretation stringent. The (admittedly) imported terminology of ``symmetry breaking'' seems to be more appropriately replaced by the more standard notions of gauge fixing and symplectic reduction we just described.
\item We now turn to the notion of ``vacuum degeneracy'' or---more generally---``state degeneracy''. As we described at point (1), all configurations related by the action of $G^S_\f$ are gauge-related and therefore physically equivalent. With reference the quoted passage of \cite{StromingerLectureNotes}, the observation that one can add an infinite number of modes $\phi\vert_{\mathrm{there}}(x)\equiv \mathfrak{Re}(2\tilde\varphi(0))\vert_{\mathrm{here}}(x)$ to any vacuum state is directly explained by the fact that these modes are pure gauge (Eq. \ref{eq:HodgeASa}), and disappear after reduction.
\item All superselection sectors are symplectomorphic to each other and to the space of ``radiative modes'' $\hcalA_\rad$:\footnote{The space $\hcalA_\rad$ corresponds to the space of AS fields with no electric zero-modes, i.e. with $\mathfrak{Re}(2\tilde\varphi(0))=0$. It is readily identified with $\hcalA_{\mu = 0}$ but we restrain from doing so because, as we will see shortly, $\hcalA_\rad \simeq \hcalA_\mu$ for \emph{any} value of $\mu$.}
\[
    \uuS_{[f_\i,f_\f]} \simeq \hcalA_\mathrm{rad} = \{ \ASa^\mathrm{rad} = D \varphi_{k>1} + D\times \beta \}.
\]
Therefore, a given $\ASa^\mathrm{rad}\in \hcalA_\mathrm{rad}$ can correspond to many physically different configurations, each corresponding to a different superselection sector, i.e.
\[
{\uuC} \simeq \bigsqcup_{(f_\i,f_\f) \in \F} \uuS_{[f_\i,f_\f]} \simeq \F \times \hcalA_\mathrm{rad}, \qquad \F \simeq \fGred^*\simeq \Ann(\fg, (\fg^{\pp\Sigma})^*)
\]
From the viewpoint of ${\uuC}$ as a Poisson space foliated by symplectic leaves (the superselection sectors), using $\ASe = \E^\i $ and $\mu =\Delta^{-1} E^\dif $ to label the sectors as per Equation \eqref{eq:SSSPhi}, and $\ASa^\rad$ to label the points in $\hcalA_\rad$, is tantamount to treating $\ASa^\rad$ as a ``coordinate'' on the leaves of ${\uuC}$, and $(\ASe,\mu)$ as ``coordinates'' transverse to said leaves. (In other words, $\hcalA_\rad$ is the generic fibre of the symplectic fibration\footnote{In \cite{RielloSchiavina} we called the space of symplectic leaves the ``space of superselections'' $\mathcal{B}$. In this simple application, $\mathcal{B} \simeq \F$.} ${\uuC} \to \F$.) In particular, the ``coordinates'' $(\ASe,\mu)$ then Poisson commute with each other and with $\ASa^\rad$, which is why we say they are ``superselected''.\footnote{Here, we have reintroduced the possibility that $\ASe \neq 0$ for completeness.}

In this ``coordinatisation'' of ${\uuC}$, the physical magnetic and electric fields at $(\ASe,\mu,\ASa^\rad)\in{\uuC}$ are
\begin{subequations}
    \label{eq:statedegeneracy}
\begin{equation}
    F = D\ASa^\mathrm{rad} = \Delta \beta
\end{equation}
and\footnote{Note that, although $u$-dependent, the second term of $E$ in this formula is bounded by the fact that $u\in [-1,1]$. In the asymptotic limit the factor of $2$ has to be replaced by the length of the (compact) $u$ interval.} 
\begin{align}
    E & = \ASe + \frac{u}{2} \Delta \mu + (D^i \ASa^\mathrm{rad}_i)^\i - D^i\ASa^\mathrm{rad}_i\notag\\& = \ASe + \frac{u}{2}\Delta\mu + \Delta \varphi_{k >1}^\i - \Delta \varphi_{k>1}.
\end{align}
\end{subequations}%
Note that the electric field's zero-modes (intended as in Lemma \ref{lemma:ASmodes}) are the ones that depend on $(\ASe,\mu)$.

Thus we see that a sort of ``state degeneracy'' exists in our formalism, and is given by the notion of \emph{superselection sectors}---thus reaching a conclusion rather different from that of the quote above. However, the ``states'' or, more appropriately, the field configurations of Equation \eqref{eq:statedegeneracy} are not physically degenerate, since they correspond to distinct asymptotic electromagnetic fields.  In particular $\ASa^\rad =0$ could be considered a ``dynamical'' vacuum configuration, but is nevertheless associated to different electric fields in different superselection sectors. (It can be seen as a zero-section of ${\uuC}\to \F$.) The ``true vacuum'' $E=F=0$ belongs to just one superselection sector: the one with $\ASe = \mu = 0$.

The ``state degeneracy'' we just described, in the language of \cite{StromingerLectureNotes} (where $\ASe = 0$, see point (ii) of the above dictionary), corresponds to different radiative states associated to different values of the ``soft photon $N\vert_\text{there} = -\mu$" and not to different values of its pure-gauge conjugate mode $\phi\vert_\text{there} = \mathfrak{Re}(2\tilde\varphi(0))$.
\item The origin and nature of the superselection sectors in our picture is quite different from the one alluded to in the quote, which instead refers to the breaking of a global symmetry by a choice of vacuum state in a Mexican-hat potential. First, not all (our) sectors carry the same energy. Moreover, no (gauge) symmetry is ever broken (but rather fixed), and strictly \emph{no Hamiltonian flow} exists over ${\uuC}$ that connects two different superselection sectors (since the superselection sectors are the symplectic leaves of ${\uuC}$, all Hamiltonian vector fields are tangent to them).\footnote{Adapting the language used in the quote, we could rephrase: over ${\uuC}$, no physical operator exists that can (continuously) move us from one vacuum to the other, neither of finite or infinite energy, for such an operator would have to be gauge-breaking.} Indeed, the whole point here is that the fully-reduced (i.e.\ the fully gauge-invariant) phase space is not symplectic but rather a disjoint collection of symplectic spaces---the superselection sectors. 
\item In light of Appendix \ref{app:Maxwell-short}, the discussion at the previous point holds pretty much unaltered even if $\Sigma$ is a spacelike hypersurface. This in principle allows one to ``glue'' the phase spaces associated to a spacelike and a null hypersurface with a common boundary $S$. This gluing has to happen a superselection sector at the time, i.e.\ preserving the electric field through $S$---which labels the superselection sectors in both phase spaces. Mathematically this might take the form of a fibre product. (See \cite[Sect. 6]{RielloGomes} for a theorem about gluing of spacelike regions; to ensure smoothness of the glued field the introduction of a ``buffer'' zone between the two regions might be required, cf. e.g. \cite{AretakisEtAl2021}; see also \cite{CattaneoMnev} for a different option.)
\item\label{item:quantisation} Upon quantisation, the fact that ${\uuC} = \C/\G$ is foliated by symplectic superselection sectors will presumably translate into the following statements: (1) the algebra of gauge invariant observables possesses a center, corresponding to (electric) flux and memory operators;  (2) the Hilbert space associated to the algebra of \emph{gauge-invariant} observables splits into a direct sum of Hilbert spaces (``soft sectors'') labelled by the values of the electric fluxes (and/or memory).

An analogous conclusion about the split of the Hilbert space into soft sectors was reached, by a variety of means, starting from the 1960s---e.g. by studying IR divergences and asymptotic states in QED scattering amplitudes \cite{Chung1965, Kibble1968-II, Kibble1968-III, Kibble1968-IV, FaddeevKulish1970} (see also \cite{KapecPerry2017, ChoiAkhoury2018, GomezLetschkaZell2018}, as well as \cite{PrabhuWald2022} for a critical analysis), or by quantising the asymptotic AS phase space (``asymptotic quantisation'') \cite{Ashtekar1984} (see also \cite{AshtekarCmapigliaLaddha2018,LaddhaHolographicnullinfinity}), or by studying the consequences of the Gauss constraint in algebraic quantum field theory \cite{FrohlichEtAl1979AnnPhys, FrohlichEtAl1979PhysLett, Buchholz86} (see also \cite{MundEtAl2022} and references therein), or by studying the symmetry structure of asymptotic fields \cite{GervaisZwanzinger1980, Giulini1995-superselection}.

We plan to investigate the quantisation of our formalism, and its relation to these matters, in future work (cf. Section \ref{sec:Introquantisation}).

\item One important point we have not yet touched upon, that is nonetheless central to the analysis of asymptotic symmetries and their relation to the soft theorems, is that of ``soft charge conservation'' \cite{KapecPateStrominger,StromingerLectureNotes}. The reason we have neglected this point so far is because such a conservation is a matter of dynamics, not kinematics, and therefore cannot be fully probed by simply analysing the action of gauge symmetries over phase space. 

However, if our interpretation in terms of superselection sectors is not only kinematically but also dynamically correct, then what emerges here is the prospect that the ``conservation of soft charges'' is nothing else than the statement that dynamics ``happens within a given superselection sector''. This interpretation is in fact compatible with the analysis of \cite{RejznerSchiavina} (see also \cite{Herdegen96,Herdegen05,HerdegenRevisited}, and the approaches of \cite{ChoiAkhoury2018} and \cite{CampigliaSoft}). We plan to come back to this point in the future.
\end{enumerate}

\subsection{$G$ semisimple}\label{sec:semisimplememory}

To conclude this discussion, we briefly comment on the generalisation of memory as a superselection label in the non-Abelian case. The morale of the story is that ``color memory'' of \cite{PateRaclariuStrominger} \emph{fails} in general to be a well-defined, gauge-invariant, notion. We propose below an alternative definition of non-Abelian memory. This discussion will reinforce some of the points made above in relation to the ideas of ``symmetry breaking'' and ``vacua''.

The simplest scenario to consider is once again the one in which a partial reduction with respect to $G^S_\i$ at $f_\i = 0$ is performed.

From Theorem \ref{thm:nonAb-constr-red} and Equation \eqref{e:reducedfluxmap}, one obtains:

\begin{theorem}\label{thm:YMPartialSSS}
The symplectic reduction of $(\uCo,\uomegao,\Gred)\simeq (\XAS, \omeAS, G^{\pp\Sigma})$, with respect to the Hamiltonian action of the initial copy of $G^S_\i\subset G^{\pp\Sigma}$ at $f_\i = 0$ (i.e.\ $\ASe=0$) yields the Ashtekar--Streubel symplectic space $(\hcalA,\omAS)$,
\[
(\XAS, \omeAS) //_{0} G^S_\i \simeq (\hcalA,\omAS).
\]
This space carries the following residual action of the gauge symmetry group $G^S$:
\[
\underline{\varrho}{}_\mathsf{AS}(\xi_\f) \ASa = \cD\xi_\f \doteq D\xi_\f + [\ASa ,\xi_\f].
\]
with momentum map given by 
\[
\langle \uh_{\mathsf{AS}},\xi_\f\rangle 
= \int_S \sqrt{\gamma}\ \tr((\cD^iL_\ell\ASa_i)^{\int} \xi_\f ).
\]
The asymptotic (gauge-covariant) electromagnetic field $(E,F)$ at $(u,x)\in\Sigma\subset \scri$ is then given by  
\[
E =  \int_{-1}^u du' \ (\cD^i L_\ell \ASa_i)(u')
\quad\text{and}\quad
F = D \ASa.
\]
\end{theorem}

\begin{corollary}\label{cor:non-ab-modes}
The real part of the AS zero-mode $\mathfrak{Re}(2\wt{\ASa}(0))$ is a connection for a principal $G$-bundle over $S$, all other AS modes are equivariant w.r.t. the adjoint representation: 
\begin{align*}
\underline{\varrho}{}_\mathsf{AS}(\xi_\f)\mathfrak{Re}(2\wt{\ASa}(0)) &= D \xi_\f + \big[\, \mathfrak{Re}(2\wt{\ASa}(0)), \xi_\f \,  \big], \\
\underline{\varrho}{}_\mathsf{AS}(\xi_\f)\mathfrak{Im}(2\wt{\ASa}(0)) &= \big[\,\mathfrak{Im}(2\wt{\ASa}(0)),\xi_\f\big],\\
\underline{\varrho}{}_\mathsf{AS}(\xi_\f)\wt{\ASa}(k) &= \big[\,\wt{\ASa}(k),\xi_\f\big] \qquad\qquad k\geq1.
\end{align*}
Thus define the AS covariant derivative,
\begin{equation}
\label{eq:AS-cov-der}   
\cD_0 \doteq D + \big[\,\mathfrak{Re}(2\wt{\ASa}(0)), \cdot\, \big].
\end{equation}
The AS symplectic structure, as well as the momentum map $\uhAS$, diagonalise in the AS modes decomposition according to:
\[
\omAS = \int_S \sqrt{\gamma} \sum_{k\geq 0} \tr\big( \bd \mathfrak{Im}(2\wt{\ASa}(k)^i)\wedge \bd \mathfrak{Re}(2\wt{\ASa}(k)_i) \big)
\]
and
\begin{align*}
(\cD^iL_\ell\ASa_i)^{\int} 
&= \cD_0^i \mathfrak{Im}(2\wt{\ASa}_i(0)) + \sum_{k \geq 1 }\big[\mathfrak{Re}(2\wt{\ASa}(k)^i),\mathfrak{Im}(2\wt{\ASa}(k)_i)\big].
\end{align*}
\end{corollary}

\begin{proof}
We prove the first statement last.
The mode decomposition of $\omAS$ was computed in  Proposition \ref{prop:modedecomp}. The mode decomposition of $\uhAS$ can be computed analogously (see Lemma \ref{lemma:ASmodes}):
\begin{align*}
    \int_{-1}^1 du \ \cD^iL_\ell \ASa_i 
    &= \int_{-1}^1du\  \left( D^i L_\ell \ASa_i + [\ASa^i, L_\ell \ASa_i] \right) \\
    &= \int_{-1}^1 du \ L_\ell D^i \ASa_i -  \int_{-1}^1du\  \left(\dot \ASa_{i}^\alpha \ASa^{i\beta} - \dot \ASa_{i}^\beta \ASa^{i\alpha}\right)\tfrac12f_{\alpha\beta}^\gamma \tau_\gamma\\
    &= (D^i \ASa_i)^\dif -  2i\mathbb{G}(\ASa^\alpha, \ASa^\beta) \tfrac{1}{2}f_{\alpha\beta}^\gamma \tau_\gamma
\end{align*}
where we decomposed $\ASa$ in a basis of the Lie algebra $\{\tau_\alpha\}$. Hence, we get the result by observing that
\[
\mathbb{G}(\ASa^\alpha, \ASa^\beta) = \sum_{k\geq0} \widetilde{\ASa}^\alpha(k)^*\, \widetilde{\ASa}^\beta(k) - c.c.= 2i \sum_{k\geq0} \mathfrak{Re}(\widetilde{\ASa}^\alpha(k))\, \mathfrak{Im}(\widetilde{\ASa}^\beta(k)) - (\alpha\leftrightarrow\beta),
\]
and recalling that $\ASa^\dif$ is the imaginary part of $2\wt{\ASa}(k=0)$.

Plugging this expression in the formula for the momentum map for the $G^S$ action on $\XAS$ we immediately obtain the claimed expression for the AS mode decomposition of the infinitesimal action, since the real and imaginary parts of AS modes are canonically conjugate to one another.
\end{proof}

Theorem \ref{thm:memorySSS} and its corollary suggest a non-Abelian generalisation of electromagnetic memory. This definition is \emph{not} equivalent to ``color memory'' as proposed in \cite{PateRaclariuStrominger}. See Remark \ref{rmk:non-Abelianmemory}.

\begin{definition}[Non-Abelian memory]\label{def:non-Abelianmemory}
In the superselection sector where $f_\i = 0$, the non-Abelian memory $\mu$ is the solution to the elliptic equation 
\[
-\cD_0^2 \mu \doteq (\cD^iL_\ell\ASa_i)^{\int} = \cD_0^i \mathfrak{Im}(2\wt{\ASa}(0)_i) + \sum_{k \geq 1 }\big[\mathfrak{Re}(2\wt{\ASa}(k)^i),\mathfrak{Im}(2\wt{\ASa}(k)_i)\big].
\]
where $\cD_0^2$ is the AS Laplacian\footnote{If $\mathfrak{Re}(2\wt{\ASa}(0))$ is irreducibile then $\cD^2_0$ is invertible; otherwise, $\mu$ is only determined up to elements of the stabiliser of $\mathfrak{Re}(2\wt{\ASa}(0))$.} associated to $\cD_0$ (Equation \eqref{eq:AS-cov-der}).
\end{definition}

\begin{remark}\label{rmk:non-Abelianmemory}
Definition \ref{def:non-Abelianmemory} is a viable generalisation of Equation \eqref{eq:memoryBG}.

Although $\mu$ is not gauge invariant, it is gauge-equivariant with respect to the residual gauge action by  $G^S_\f$, and superselection sectors (with $f_i = 0$) are indeed labelled by the coadjoint orbit $\mathcal{O}_\mu$ of $\mu \in (\fg^S)^*$.

Note, however, that the non-Abelian memory $\mu$ (or its coadjoint orbit) fails to be a viable superselection label as soon as $f_\i\neq0$. This is due to the non-linearity of the coadjoint orbits. (See also Remark \ref{rmk:YMpartialSSS}.)

In the scenario where $f_{\i}=0$ (i.e.\ $E^\i = 0$) and non-Abelian memory as superselection \emph{does} make sense, we see however that it fails to split into quantities defined at $S_\i$ and $S_\f$, as opposed to the Abelian case:
\[
(\cD^i L_\ell \ASa_i)^{\int}=D^i\ASa_i^\dif + \sum_{k \geq 0 }\big[\mathfrak{Re}(2\wt{\ASa}(k)^i),\mathfrak{Im}(2\wt{\ASa}(k)_i)\big] \neq D^i\ASa_i^\dif.
\]
owing to $\mathfrak{Im}(2\wt{\ASa}_i(0)) = \ASa_i^\dif$.  

There is therefore no obvious sense in which memory is about a net change in the $\ASa_i$: it is instead about a net change in the value of the electric field, up to an appropriate transport by $\Lambda$ (cf. the proof of Proposition \ref{prop:red-flux-map}):
\begin{equation}
\label{eq:memory-nonab}
E^\f = \underbrace{\Ad(\Lambda^{-1})\cdot E^\i}_{=0} + D^i\ASa_i^\dif + \sum_{k \geq 0 }\big[\mathfrak{Re}(2\wt{\ASa}(k)^i),\mathfrak{Im}(2\wt{\ASa}(k)_i)\big] \neq D^i\ASa_i^\dif.
\end{equation}

In other words, even setting $F^{\i}=F^\f = 0$, as is assumed in \cite{PateRaclariuStrominger,StromingerLectureNotes}, and thus being able to set $\ASa^{\i} = U^{-1}_{\i} D U_{\i}$ and $\ASa_\f=U_\f^{-1} DU_\f$, a description of non-Abelian memory as a superselection label in terms of the ``vacuum transition" (or ``color memory'' \cite{PateRaclariuStrominger}) $U=U_\i^{-1}U_\f$ is not possible: the entire history of $\ASa_i$ along $\scri$ is necessary to compute $\mu$, rather than just its initial and final values. Comparing to \cite{PateRaclariuStrominger}, we see that the Equation \eqref{eq:memory-nonab} is the same as their Equation (19), upon identification our the last term with their integral of the ``color flux $J_u$".
\end{remark}

\begin{remark}\label{rmk:YMpartialSSS}
For $G$ semisimple, reduction at $f_\i = 0$ is \emph{qualitatively} different than reduction at $f_\i\neq 0$, and would \emph{not} yield the AS phase space as a result.
This is because only $\mathcal{O}_{f_\i =0 }$ is point-like. Here, we refrain from providing a more general statement.
\end{remark}

\appendix

~

\section{Example: Maxwell theory on a spacelike slice}\label{app:Maxwell-short}
We now give a succinct summary of Abelian YM theory on a 3-dimensional Riemannian manifold $(\Sigma,\gamma)$---thought of as a spacelike codimension-$1$ submanifold of $M$ Lorentzian---as an exemplification of the theory outlined in Section \ref{sec:theoreticalframework}. Details can be found in \cite{RielloSchiavina,RielloSciPost}.

For $N$ a manifold, we denote $\mathbb{R}^N \doteq C^\infty(N,\mathbb{R})$.

Let $\X = T^*\mathcal{A}\ni (A,E)$, with $A\in\mathcal{A} \simeq \Omega^1(\Sigma)$, $E\in T^*_A\mathcal{A} \simeq \Omega^2(\Sigma)$, and $\bom = \bd E \wedge \bd A$, be the symplectic space of ``magnetic potentials and electric fields'' over $\Sigma$. The (Abelian) gauge algebra is $\fG = \mathbb{R}^\Sigma$; it acts on $\X$ as $\rho(\xi)(A,E) = (d\xi,0)$.  
Assume that the first de Rham cohomology of $\Sigma$ is trivial, and denote $n$ the normal to the boundary of $\pp\Sigma$ and $\ast$ the Hodge dual on $(\Sigma,\gamma)$. , Then, using the Hodge--Helmholtz decomposition, it is not hard to show that $\mathcal{A}_\rad \doteq \{A \in \mathcal{A} \ : \ d\ast A =0,\ i_n A |_{\pp\Sigma} = 0\}$ is (the image of) a global section of $\mathcal{A} \to \mathcal{A}/\G$ and hence $\mathcal{A}_\rad \simeq \mathcal{A}/\G$. 

The momentum form $\bH$ decomposes into the sum of a (Gauss) constraint form $\langle\bHo,\xi\rangle = -\xi dE$, and a flux form $\langle d\bh,\xi\rangle = d(E\xi)$ which, once integrated on $\Sigma$, gives the smeared electric flux through $\pp\Sigma$. The constraint gauge group is $\fGo = \{\xi\in\fG \ : \ \exists \chi \in \mathbb{R} \text{ such that } \xi|_{\pp\Sigma}=\chi\}$ (constant gauge transformations play a role because of Gauss's law); the flux gauge group is $\fGred \simeq \mathbb{R}^{\pp\Sigma}/\mathbb{R}$, where $\mathbb{R}\into \mathbb{R}^{\pp\Sigma}$ as constant functions; while flux space is found to be $\F\simeq \fGred^* \simeq \Ann(\mathbb{R},\mathbb{R}^{\pp\Sigma})$. 

The first- and second-stage reduced phase spaces are isomorphic (as symplectic or Poisson, resp.) manifolds to the following spaces:
\[
\uCo \simeq T^*\mathcal{A}_\rad \times T^*\fGred
\quad\text{and}\quad
{\uuC} \simeq T^*\mathcal{A}_\rad \times \fGred^* 
\]
Manifestly, $\uCo$ is a symplectic manifold, whereas ${\uuC}$ is only Poisson. Since the flux space is $\F\simeq \fGred^*$, the last formula states that all the superselection sectors are all isomorphic: 
\[
\uuS_{[f]} \simeq T^*\mathcal{A}_\rad.
\]
Physically $T^*\mathcal{A}_\rad$ encodes the radiative degrees of freedom (the ``photons") over $\Sigma$, while $\F\simeq \fGred^*$ encode the electric ``Coulombic'' degrees of freedom (i.e.\ the co-exact part of $E$, $d\star\varphi$) as parametrised by $\iota_{\pp\Sigma}^*E|_{\pp\Sigma}$, the electric flux through $\pp\Sigma$. This is possible because of the Gauss constraint $\bHo = dE$, which in the given decomposition is given by the equations $\Delta \varphi = 0$ and $i_nd\varphi|_{\pp\Sigma} = \star_{\pp\Sigma}\iota_{\pp\Sigma}^*E$.

Note that, whereas in $\uCo$ the electric fluxes are conjugate to the gauge variant elements of $\fGred$---sometimes called ``edge modes'' \cite{DonnellyFreidel16}---they have no symplectic partner in ${\uuC}$. Indeed, in ${\uuC}$, they are precisely the central coordinates in ${\uuC}$ whose value labels the superselection sectors.

In the non-Abelian case, analogous results hold where $T^*\fGred$ is replaced by $T^*\Gred$, but the symplectomorphisms are only local \cite[Section 6.5]{RielloSchiavina}.

\section{Notes on Definition \ref{def:LHGT} and locality}\label{rmk:locality}

The central notion used in the definition of a \emph{locally Hamiltonian gauge theory} (Definition \ref{def:LHGT}) is that of locality.
Here, we briefly clarify this notion and refer to \cite{RielloSchiavina} for a detailed discussion (for a more general viewpoint see e.g.\ \cite{BlohmannLFT}).
\begin{enumerate}[label=\arabic*.,leftmargin=*]
\item Let $Y_{i}\to \Sigma$, $i=1,2$, be two fibre bundles, and $\mathcal{Y}_{i}=\Gamma(\Sigma,Y_{i})$ the corresponding spaces of sections, denoted $\varphi_{i}$; then, a map $f:\mathcal{Y}_1 \to \mathcal{Y}_2$ is said \emph{local} iff $\varphi_2(x) \doteq f(\varphi_1)(x)$ can be expressed as a function of $x$, $\varphi_1(x)$, and a finite number of its derivatives (called the \emph{order} of the map) also evaluated at $x$. A map of order $0$ is said \emph{ultralocal}.
\item There exists a notion of ``local forms'' $\balpha \in \oloc^{p,q}(\Sigma\times \X) \subset \Omega^{p,q}(\Sigma\times \X)$ that generalises that of local maps and function(al)s discussed above \cite{BlohmannLFT}.
We denote by $(d,i,L)$ and $(\bd,\bi,\mathbb{L})$ the symbols of Cartan calculus on local forms $\oloc^{\bullet,\bullet}(\Sigma\times\X)$.
We will see that whereas Hamiltonian Yang--Mills theory on a spacelike $\Sigma$ features an ultralocal symplectic density, Hamiltonian Yang--Mills theory on a null $\Sigma$ does \emph{not}.
\item A (real) Lie algebra $\fG$ is said local if its elements $\xi$ are sections of a (real) vector bundle $\Xi\to\Sigma$ \emph{and} its Lie bracket $[\cdot,\cdot] : \fG \wedge \fG \to \fG$ is a local map.
\item The action $\rho:\X\times \fG \to T\X$ is said local iff it is a local map. 
\item Local $\mathbb{R}$-linear maps $\fG\to\Omega^{\mathrm{top},0}(\Sigma\times \X)$ can be identified with local forms $\oloc^{\mathrm{top,0}}(\Sigma\times \X,\fG^*_\loc)$ valued in the local dual (defined in Remark \ref{rmk:dualspaces}), or with local forms in the triple complex $\oloc^{\mathrm{top,0,0}}(\Sigma\times\X\times\fG)$ which are linear in $\fG$ \cite[Defintion 2.12]{RielloSchiavina}. 
Throughout, we will deliberately merge these notions. With reference to the definition of a locally Hamiltonian gauge theory this means merging the \emph{momentum} and \emph{co-momentum} map viewpoints, and only talking about momentum forms and maps when referring to $\bH$ and derived objects.
\item The quantity $d\boldsymbol{k}(\xi,\eta)$ is a Chevalley--Eilenberg cocycle of $\fG$. It does not depend on the fields ($\bd d\boldsymbol{k}(\xi,\eta)=0$) \cite[Sect. 3.5]{RielloSchiavina}. Its $d$-exactness is here imposed to match the Lagrangian origin of the locally Hamiltonian gauge theory, where equivariance up to a boundary cocycle is a consequence of Noether's theorem and encodes the first-class nature of the gauge constraints \cite[Appendix D]{RielloSchiavina}.
\qedhere
\end{enumerate}

\section{Wilson lines and path-ordered exponentials} \label{app:usefulformulas}

In this appendix, $G$ denotes a Lie group which is either (\emph{i}) Abelian, or (\emph{ii}) compact and semisimple.

We begin with a statement about a 1-dimensional parallel transport problem for each of the $k$-components of the map $Y:I\to \Omega^k(S,\fg)$, the lemma readily follows from the theory of ODEs in one variable:

\begin{lemma}[Parallel transport along $\ell$]\label{lemma:paralleltr}
Let $Y,Z:I \to \Omega^k(S,\fg)$ and $Z_0 = \Omega^k(S,\fg)$. Then, for all smooth $A\in\mathcal{A}$, $Z$ and $Z_0$, the boundary value problem\footnote{Recall: $\cL_\ell Y \doteq L_\ell Y + [A_\ell, Y]$.}
\[
\begin{cases}
\cL_\ell Y = Z\\
Y|_{u=-1} = Z_0
\end{cases}
\]
admits a unique smooth solution $Y(u)=(Y(A,Z,Z_0))(u)$.
\end{lemma}

Next, we recall a standard result on the definition of Wilson lines (sometimes referred to as ``holonomies'' or ``parallel transports") of a gauge connection along the integral curves of $\ell=\pp_u$:\footnote{Here we phrase the lemma in terms of a ``final condition'' rather than an ``initial one''. Of course this makes little difference, and is in fact the form of the lemma that will be most useful for this article.}

\begin{notation}
Let $x\in S$. By definition $(L_\ell UU^{-1})(x)\doteq i_\ell U_x^*\vartheta$, where  $U_x \doteq \mathrm{ev}_x\circ U: I \to G$ and $\vartheta \in \Omega^{1}(G,\fg)$ is the right-invariant Maurer-Cartan form on $G$. This quantity is valued in $C^\infty(S,\fg) = \fg^S$.
\end{notation}

\begin{lemma}[Wilson lines along $\ell$]\label{lemma:Wilsonline}
Let $U\in C^\infty(I,G^S)$. Then, for all smooth $A \in\mathcal{A}$, the boundary value problem
\[
\begin{cases}
L_\ell \dfU \dfU^{-1} = - A_\ell\\
\dfU(u=+1) = 1
\end{cases}
\]
admits a unique smooth solution $\dfU(u)=(\dfU(A))(u)$, which we call the \emph{Wilson lines (of $A$ along $\ell$),} and denote
\[
(\dfU(A))(u) \equiv \overrightarrow{\Pexp} \int_u^{1} du' \ A_\ell(u').
\]
In the Abelian case one has the identity
\[
(\dfU(A))(u) = \exp \int_u^1 du' \, A_\ell(u') \qquad \mathrm{(Abelian)}
\]
\end{lemma}

Thanks to the smoothness and uniqueness of the Wilson lines $\dfU(A)$, one deduces:
\begin{corollary}[of Lemma \ref{lemma:Wilsonline}]
The Wilson lines $\dfU(A) = \overrightarrow{\Pexp}\int A_\ell : I \to G^S$ can be interpreted as maps: $\mathcal{A}\to G^\Sigma \simeq C^\infty(I,G^S)$ which are \emph{bi-}locally equivariant under the action of gauge transformations $g \in G^\Sigma$,
\[
(\dfU(A^g))(u) = (g(u))^{-1} (\dfU(A))(u) g^\f.
\]
In particular if $A=g^{-1} dg$, 
\[
(\dfU( g^{-1}dg))(u) \equiv \overrightarrow{\Pexp} \int_u^{1} du' \ g^{-1}L_\ell g(u') = (g(u))^{-1} g^\f.
\]
\end{corollary}

\section{Proofs of some lemmas}

\subsection{Proof of Lemma \ref{lem:annihilators}} \label{app:proof-annihilators}

Let $\mathcal{X}\subset \mathcal{W}$ and $\mathcal{Y}=\Ann(\mathcal{X},\mathcal{W}^*_{\mathrm{str}})$. From
\[
\Ann(\mathcal{X},\mathcal{W}^*_{\mathrm{str}}) \simeq \left(\mathcal{W}/\mathcal{X}\right)^*_{\mathrm{str}}, \qquad \Ann(\mathcal{Y},\mathcal{W}) \simeq \left(\mathcal{W}^*_{\mathrm{str}}/\mathcal{Y}\right)^*_{\mathrm{str}}
\]
we obtain 
\[
\Ann(\Ann(\mathcal{X},\mathcal{W}^*_{\mathrm{str}}),\mathcal{W}) \simeq \left( \mathcal{W}^*_{\mathrm{str}} / \left(\mathcal{W}/\mathcal{X}\right)^*_{\mathrm{str}}\right)^*_{\mathrm{str}}
\]
Dualising the short exact sequence 
\[
\mathcal{X} \to \mathcal{W} \to \mathcal{W}/\mathcal{X} 
\]
we conclude that $\mathcal{X}^*_{\mathrm{str}} \simeq \mathcal{W}^*_{\mathrm{str}} / \left(\mathcal{W}/\mathcal{X}\right)^*_{\mathrm{str}}$. 
Hence, since both nuclear Fre\'echet vector spaces and their strong duals, as well as their closed subspaces,\footnote{Any closed subspace of a nuclear Fr\'echet space is both nuclear and Fr\'echet. A quotient by a closed subspace also retains the nuclear Fr\'echet property.} are reflexive (\cite[Rmk. 6.5]{kriegl1997convenient}), one finds
\[
\Ann(\Ann(\mathcal{X},\mathcal{W}^*_{\mathrm{str}}),\mathcal{W}) \simeq \left(\mathcal{X}^*_{\mathrm{str}}\right)^*_{\mathrm{str}} \simeq \mathcal{X}^{**}_{\mathrm{str}} = \mathcal{X}.
\]

\subsection{The two-form of Definition \ref{def:geom-ph-sp} is symplectic}\label{app:proof-nondeg}
According to Definition \ref{def:LHGT}(\emph{ii}), $\bom_\nYM$ of Definition \ref{def:geom-ph-sp} is a symplectic density iff $\bom_\nYM$ is $\bd$-closed and $\omega_\nYM = \int_\Sigma \bom_\nYM$ is (weakly) symplectic. The form $\bom_\nYM$ is obviously $\bd$-closed, and therefore so is $\omega_\nYM$. Therefore, the crux is proving the (weak) non-degeneracy of $\omega_\nYM$, i.e.\  $\mathbb{i}_{\mathbb{X}}\omega_\nYM = 0$ iff $\mathbb{X}=0$.

For ease of notation, we henceforth drop the subscript $\cdot_\nYM$, i.e.\ $\bom_\nYM \leadsto \bom$ etc.

Using the identity of Equation \eqref{eq:deltaF_ell} for $\bd F_\ell$, the form $\bom$ can be rewritten as
\[
  \bom   =  \Big( \tr(\bd \E \wedge \bd A_\ell ) 
+ \tr( - \cD^i \bd A_\ell \wedge \bd \hat A_i)
+ \tr( \cL_\ell \bd \hat A^i \wedge \bd \hat A_i)
\Big)\vol_\Sigma.
\]
Denoting
\[
\mathbb{X} = \int X_E \frac{\delta}{\delta \E} + \hat{X}^i \frac{\delta}{\delta \hat A^i} +  X_\ell \frac{\delta}{\delta A_\ell},
\]
one finds:
\begin{align*}
\bi_\mathbb{X}\bom 
& =  \Big( \tr(X_E  \bd A_\ell ) -  \tr(\bd \E  X_\ell ) + \\
&\qquad\qquad + \tr( (\cL_\ell \hat X^i - \cD^i X_\ell)   \bd \hat A_i) -\tr( (\cL_\ell \bd \hat A^i - \cD^i\bd A_\ell)   \hat X_i) \Big) \vol_\Sigma.
\end{align*}
Factoring out a total divergence (i.e.\ ``integrating by parts''), one finds:
\begin{subequations}
\begin{align}
\bi_\mathbb{X}\bom 
& =   \Big( -  \tr(  X_\ell \bd \E) + \tr(X_E  \bd A_\ell )  +  \tr( (\cL_\ell \hat X^i  )  \bd \hat A_i) + \notag \\
&\hspace{5cm} +  \tr( (- \cL_\ell \bd \hat A^i + \cD^i\bd A_\ell)   \hat X_i) \Big)\vol_\Sigma\\ \notag
& =  \Big( -  \tr(  X_\ell \bd \E) + \tr( (X_E - \cD^i \hat X_i) \bd A_\ell )  +  2\tr( (\cL_\ell  \hat X^i )  \bd \hat A_i) + \\ \label{e:kernelthirdrow}
& \hspace{5cm} + D^i\tr(\hat X_i \bd A_\ell  ) - \pp_u \tr( \hat X^i \bd  \hat A_i  ) \Big)\vol_\Sigma \\
& \doteq (\bi_\mathbb{X}\bom)_\mathrm{src} + (\bi_\mathbb{X}\bom)_\mathrm{bdry}, 
\end{align}
\end{subequations}
where the ``source'' $(\mathrm{top},1)$-form $(\bi_\mathbb{X}\bom)_\mathrm{src}$ is defined by the first line of Equation \eqref{e:kernelthirdrow}, while the ``boundary'' $(\mathrm{top},1)$-form $(\bi_\mathbb{X}\bom)_\mathrm{bdry}$ is defined by the second line of the same equation.\footnote{Recall: a $(\mathrm{top},1)$-form is said ``source'' iff \emph{no} derivative acts on the field-space 1-form terms $\bd\phi$; it is said ``boundary'' iff it is $d$-exact. A theorem by Takens \cite{Takens77,Takens,Zuckerman} then states that the space of $(\mathrm{top},1)$-forms is the direct sum of the space of source and boundary $(\mathrm{top},1)$-forms. See \cite[Theorem 2.18]{RielloSchiavina} or \cite{BlohmannLFT} for more details.\label{fnt:app-Takens}}

We now use these formulas to prove that $\omega$ is (weakly) nondegenerate, i.e.\ that if $0 = \bi_\mathbb{X}\omega = \int_\Sigma \bi_{\mathbb{X}}\bom$ then $\mathbb{X}=0$.

If $0=\bi_\mathbb{X}\omega$, then  $\bi_{\mathbb{Y}}\bi_\mathbb{X}\omega =0$ for all $\mathbb{Y}$ with $(Y_E, Y_\ell, \hat Y_i)$ of compact support in $\mathring\Sigma$, the open interior of $\Sigma$. Using Takens' theorem, we notice that for such $\mathbb{Y}$'s $0=\bi_\mathbb{Y}\bi_\mathbb{X}\omega = \int_\Sigma \bi_\mathbb{Y}(\bi_\mathbb{X}\bom)_\mathrm{src}$. From the fundamental lemma of the calculus of variations, it follows that $(\bi_\mathbb{X}\bom)_\mathrm{src} = 0$, that is:
\[
\bi_\mathbb{X}\omega = 0  \implies (\bi_\mathbb{X}\bom)_\mathrm{src} = 0 \iff X_\ell = X_E - \cD^i \hat X_i = \cL_\ell \hat X^i = 0,
\]
where Equation \eqref{e:kernelthirdrow} was used in the last step.
Now recall that $\bi_\mathbb{X}\omega = \int_\Sigma (\bi_\mathbb{X}\bom)_\mathrm{src} + (\bi_\mathbb{X}\bom)_\mathrm{bdry}$. Leveraging the previous conditions, and using similar arguments, we also find that
\[
\bi_\mathbb{X}\omega = 0 \implies 0 = \int_\Sigma \bi_\mathbb{X}\bom_{\mathrm{bdry}} = \int_{\pp\Sigma} \btr(\hat X^i \bd A_i) \iff \hat X^i(u=\pm1) = 0.
\]
Combining this with the previous conditions on $X_\ell$, $X_E$ and $\cL_\ell\hat X^i$ as well as Lemma \ref{lemma:paralleltr}, we conclude that $\bi_\mathbb{X}\omega=0$ implies $X_\ell = X_E = \hat X_i = 0$ i.e.\ $\mathbb{X}=0$. 

\begin{remark}
Staring from Equation \eqref{e:kernelthirdrow}, it is similarly possible to use Takens's theorem (Footnote \ref{fnt:app-Takens}) and Lemma \ref{lemma:paralleltr} to show that $\bom$ is (weakly) nondegenerate, i.e.\ that $\bom^\flat$ is injective. However, in general, $\bom^\flat$ injective does not imply $\omega^\flat$ injective (in fact from the previous argument it is clear that the opposite is true, since $\ker(\bom^\flat) \subset \ker(\omega^\flat) \subset \ker(\bom^\flat_\mathrm{src})$). 
\end{remark}

\subsection{Proof of Lemma \ref{lemma:ASmodes}}\label{proof:ASmodes}
Let $f\in C^\infty(I,\mathbb{R})$ and define
\[
\bar f(u) \doteq f(u) - ( f^0 - f^\av) - \frac{u}{2}f^\dif.
\]
Observe that
\[
\bar f^0 = f^0 - 2 (f^0 - f^\av)  = f^\av - (f^0 - f^\av) = \bar f^\av,
\]
whence, $\bar f^\dif = 0= \bar f^0 - \bar f^\av$.
In particular, $\bar f$ is $C^\infty$ and periodic over $I$.
Therefore, the partial Fourier series $F_N[\bar f]$,
\[
F_N[\bar f](u) \doteq \frac12 {\bar a}(0) + \sum_{k=1}^N {\bar a}(k) \cos(\pi k u) + {\bar b}(k) \sin(\pi k u), \quad
{\bar a}(k) + i {\bar b}(k) \doteq \int_{-1}^1 \bar f(u) e^{i\pi k u},
\]
converges to $\bar f$ uniformly.
In particular, 
\begin{equation}
    \label{eq:av-converg}
\frac12 {\bar a}(0) + \sum_{k=1}^N (-1)^k {\bar a}(k) = F_N[\bar f]^\av \to \bar f^\av = \bar f^0 = {\bar a}(0).
\end{equation}

Moreover, note that the Fourier coefficients of $f$ are related to those of $\bar f$ via
\[
a(k) + ib(k) \doteq \int_{-1}^1 f(u) e^{i\pi k u} =  \begin{dcases} 
{\bar a}(0) +2(f^0 - f^\av)  & k=0,\\
{\bar a}(k)  + i \left({\bar b}(k)-  \frac{(-1)^k}{\pi k} f^\dif \right) & k\geq 1.
\end{dcases}
\]

Consider now the partial series
\[
S_N[f](u) \doteq  \sum_{k=0}^N \tilde f(k)^*\psi_k(u) + c.c. ,
\]
where
\begin{equation}
    \label{eq:mode-exp}
2\tilde f(k) \doteq 2\mathbb{G}(\psi_k , f) = 
\begin{dcases}
(f^0 - f^\av) + i f^\dif & k=0\\
a(k) + i \left(2 \pi k b(k) + 2(-1)^k f^\dif \right) & k\geq 1
\end{dcases}
\end{equation}
Then, decomposing $\psi_k=\psi^R_k + i \psi^I_k$, and $\tilde f = \tilde f^R + i \tilde f^I$, into their real and imaginary parts, 
\begin{align*}
S_N[f](u) 
&= 2 \sum_{k=0}^N\, \tilde{f}^R(k) \psi_k^R(u) + \,\tilde{f}^I(k) \psi_k^I(u) \\
& = (f^0 - f^\av) + \frac{u}{2} f^\dif + \sum_{k=1}^N {\bar a}(k) \psi_k^R(u) + 2\pi k {\bar b}(k) \psi^I_k(u)\\
& = (f^0 - f^\av) + \frac{u}{2} f^\dif + \sum_{k=1}^N {\bar a}(k) ( (-1)^k + \cos(\pi k u) ) + {\bar b}(k) \sin(\pi k u)\\
& = (f^0 - f^\av)  + \frac{u}{2} f^\dif + F_N[\bar f](u) - \frac12 {\bar a}(0) + \sum_{k=1}^N (-1)^k {\bar a}(k)
\end{align*}
where we used that for $k\geq 1$, $\bar{a}(k)=a(k)$, as well as the expression for ${b}(k)$ in terms of $\bar b(k)$.
Finally, from the uniform convergence of the Fourier series $F_N[\bar f](u)\to \bar f(u)$ and Equation \eqref{eq:av-converg}, we find
\[
S_N[f](u) \xrightarrow[N\to\infty]{\mathrm{unif.}} (f^0 - f^\av) + \frac{u}{2} f^\dif + \bar f(u) = f(u).
\qedhere
\]

\subsection{Proof of Proposition \ref{prop:modedecomp}}\label{proof:ASmodes-sympl}
Generalising the result of Lemma \ref{lemma:ASmodes} (Appendix \ref{proof:ASmodes}) to the case where $f\in C^\infty(I)$ is replaced by $\ASa \in \hcalA \simeq C^\infty(I,\Omega^1(S))$, one finds
\[
S_N[\ASa_i(x)](u) \doteq \sum_{k=0}^N \left(\, \tilde{\ASa}_i^* (k,x) \psi_k(u) + c.c. \right)  \xrightarrow[N\to\infty]{\mathrm{unif.}} \ASa_i(u,x),
\]
and thus 
\begin{align*}
\omAS  
&\doteq \int_S \int_{-1}^1 (L_\ell\bd \ASa_i(u,x)) \wedge \bd \ASa^i(u,x) du \ \vol_S \\
&= 2i \int_S \mathbb{G}( \bd \ASa_i(x), \bd \ASa^i(x)) \vol_S\\
&= 2i \int_S  \lim_{N,M\to\infty} \mathbb{G}( S_N[\bd \ASa_i(x)] , S_M[\bd \ASa^i(x)] ) \vol_S\\
&= 2i \int_S \sum_{k=0}^\infty \bd\, \tilde{\ASa}_i^*(k,x) \wedge \bd \, \tilde{\ASa}^i(k,x) \vol_S,
\end{align*}
where in the third equality we used the fact that the convergence is uniform to pull out the limit from the integral over $u\in I$ that is implicit in the definition of $\mathbb{G}$. The final expression of $\omAS$ where the $k=0$ term is singled out is straightforward.

\subsection{Details of Calculation \ref{prop:step2}}\label{app:compute-bom-dress}
We compute $\omega(A^U,E^U)$. Recall Definition \ref{def:geom-ph-sp}:
\[
\bom(A,E) \doteq \tr\big( \bd F_\ell^i \wedge \bd A_i + \bd E \wedge \bd A_\ell )\vol_\Sigma.
\]
Thus, we compute:
\begin{align}
\bom(A^U, E^U) & =\tr\Big( \bd(U^{-1} F_{\ell}^i U) \wedge \bd (U^{-1} A_i U + U^{-1}\pp_i U)  \nonumber\\ 
& \qquad\qquad + \bd(U^{-1} E U) \wedge \bd (U^{-1} A_\ell U + U^{-1}\pp_u\bd U\Big)\vol_\Sigma \nonumber\\
& = \bom(A,E) \label{eq:contostep1}\\& \quad + \tr\Big(
 [ F_{\ell}^i,\bd U U^{-1} ] \wedge \cD_i(\bd U U^{-1}) + [E, \bd U U^{-1}]\wedge \cL_\ell (\bd U U^{-1})\Big)\vol_\Sigma \nonumber\\ 
 & \quad + \tr\Big(\bd E \wedge \cL_\ell\bd U U^{-1} - \bd F_{\ell}^i \wedge \cD_i\bd U U^{-1} \nonumber\\ 
 &\qquad\qquad\qquad\qquad\quad + [F_{\ell}^i , \bd U U^{-1}] \wedge \bd A_i + [E,\bd U U^{-1}]\wedge \bd A_\ell 
\Big)\vol_\Sigma \nonumber
\end{align}
We focus first on the second line of \eqref{eq:contostep1}. Recalling that the Gauss constraint reads $\mathsf{G} = \cL_\ell E + \cD_i F_{\ell}^i$, and using the identity
\[
[\bd U U^{-1}, \cL_\ell \bd U U^{-1}] = \tfrac12 \cL_\ell [\bd U U^{-1},\bd U U^{-1}] = \cL_\ell \bd (\bd U U^{-1}),
\]
we find:
\begin{align*}
[2^\mathrm{nd}] & = \tr\Big(
F_{\ell}^i [\bd U U^{-1}, \cD_i \bd U U^{-1} ] + E [\bd U U^{-1}, \cL_\ell \bd U U^{-1}]
\Big)\vol_\Sigma\\
& = - \tr\big(  \mathsf{G} \bd (\bd U U^{-1} ) \big)\vol_\Sigma \\ &  \quad+ L_\ell \tr\big( E  \bd(\bd U U^{-1})\big)\vol_\Sigma  + D_i \tr\big( F_\ell^i  \bd(\bd U U^{-1}) \big)\vol_\Sigma
\end{align*}
Next, we focus on the third and fourth lines of  \eqref{eq:contostep1}. Using the following identity for the variation of the Gauss constraint, 
\begin{align*}
\tr\big( (\bd\mathsf{G})\xi\big) & = \tr\big((\cL_\ell\bd E + \cD_i \bd F_{\ell}^i + [\bd A_\ell, E] + [\bd A_i, F_{\ell}^i])\xi \big) \\
& = \tr\big( - \bd E \ \cL_\ell \xi - \bd F_{\ell}^i \ \cD_i\xi + [\bd A_\ell, E] \xi + [\bd A_i, F_{\ell}^i] \xi \big) \\&\qquad\qquad +  L_\ell\tr\big( (\bd E) \xi\big) +  D_i\tr\big( (\bd F_{\ell}^i)  \xi\big),
\end{align*}
we find:
\begin{align*}
[3^\mathrm{rd}+4^\mathrm{th}] & = - \tr(\bd \mathsf{G} \wedge \bd U U^{-1})\vol_\Sigma 
\\ & \quad+ L_\ell\tr\big( \bd E \wedge \bd U U^{-1}\big)\vol_\Sigma 
+  D_i\tr\big( \bd F_{\ell}^i \wedge\bd U U^{-1})\vol_\Sigma\\
\end{align*}
Summing the two contributions, we obtain:
\begin{align*}
\bom(A^U,E^U)
& = \bom(A,E)  + \bd \tr\big( 
\mathsf{G} \bd U U^{-1} )\vol_\Sigma \\&\qquad + L_\ell \bd \tr( E  \bd U U^{-1} )\vol_\Sigma + D_i \bd \tr( F_\ell^i \bd U U^{-1})\vol_\Sigma,\\
\end{align*}
and thus, integrating over $\Sigma \simeq S \times I$, we conclude:
\begin{align*}
\omega(A^U,E^U) & = \omega(A,E)  + \bd \int_\Sigma \btr\big( 
\mathsf{G} \bd U U^{-1} )+ \\ &\qquad +  \bd\int_{S_\f} \btr\big( E^\f  \bd U^\f (U^\f)^{-1} \big)  - \bd \int_{S_\i}  \btr\big( E^\i  \bd U^\i (U^\i)^{-1} \big).
\end{align*}

\section{Proof of Equation \eqref{e:AbelianASreduction} on Abelian YM in Theorem \ref{thm:nonAb-constr-red}}\label{app:Abelianreduction}

In this appendix we prove that
\[
(\uCo,\uomegao) \simeq (\Xas,\omeas)\simeq_\loc (\XAS,\omeAS) \qquad \mathrm{\emph{(}Abelian\emph{)}},
\]
where
\[
(\Xas,\omeas)\doteq(\hcalA \times T^*\fg^S , \omAS + \omega_S).
\]
In other words, we prove that there exists a \emph{global} symplectomorphism between $(\uCo,\uomegao) \simeq (\Xas,\omeas)$.

The fact that $(\Xas,\omeas)$ is \emph{locally} symplectomorphic to $(\XAS,\omeAS)$ is straightforward and relies on the multi-valuedness of the $\log : G_0^S \to \fg^S$. Moreover, we note that the local symplectomorphism $(\uCo,\uomegao) \simeq_\loc (\XAS,\omeAS)$ was proved on general grounds in the main Section \ref{sec:reduction-first}.

\subsection{Preliminaries}

\begin{notation}\label{not:AbLiegp}
If (the connected real Lie group) $G$ is Abelian then it is the direct product of a torus and the real line. We write
\[
G\simeq \mathrm{U}(1)^t \times \mathbb{R}^{k} \qquad \mathrm{and}\qquad \fg\simeq (i\mathbb{R})^{t}\times \mathbb{R}^{k} .
\]
\end{notation}

\begin{notation}
Recall that $G\hookrightarrow G^\Sigma$ as the space of constant mapping functions. Then, 
\[
G \cdot G^\Sigma_\rel \equiv \{ g\in\G \ : \ \exists k\in G\text{ such that } g|_{\pp\Sigma}=k\}.
\]
\end{notation}

\begin{lemma}[Abelian winding number]\label{lemma:winding}
Let $g\in G \cdot G^\Sigma_\rel$, then the following $\fg$-valued is valued in an integer lattice of $\fg\hookrightarrow \fg^S$,
\[
w(g)\doteq \frac{1}{2\pi} \int_{-1}^1 du' \ g^{-1} L_\ell g(u')  \in (i\mathbb{Z})^t \times \{0\}.
\]
Moreover, $w(g_1g_2) = w(g_1) + w(g_2)$.
We call $w(g)$ the \emph{winding number} of $g^S\in G \cdot G^\Sigma_\rel$.
\end{lemma}
\begin{proof}
Consider $a\leq b$ in $I$, and recall the identity (for $g\in \G$):
\[
W(g; a,b) \doteq \exp \int_{a}^b du'\, g^{-1}\pp_{u'}g = g^{-1}(a) g(b).
\]
This formula can be read as the gauge-equivariance of the (Abelian) path-ordered exponential---i.e.\ at each $x\in S$ this formula is the holonomy (a.k.a. Wilson line, see Appendix \ref{app:usefulformulas}), drawn at constant $x$, between $(x,u=a)$ and $(x,u=b)$ (see Appendix \ref{app:usefulformulas}).
Then, if $g\in\tilde G \cdot G^\Sigma_\rel$, we have $\exp 2\pi w(g) = W(g; -1,1) = 1$, and therefore $w(g) \in i \mathbb{Z}^{t} \times \{0\}$.

We conclude by observing that the additivity $w(g_1g_2) = w(g_1) + w(g_2)$ follows immediately from the definition since
\[
(g_1g_2)^{-1} L_\ell (g_1g_2) = g_2^{-1}g_1^{-1} (L_\ell g_1 g_2 + g_1 L_\ell g_2) = g_1^{-1} L_\ell g_1 + g_2^{-1} L_\ell g_2.\qedhere
\]
\end{proof}

Recall that the space of \emph{on-shell} fluxes is $\F \doteq \Im(\iota_\C^*h)$, the constraint gauge algebra is $\fGo \doteq \Ann(\F,\fg^\Sigma)$, and the constraint gauge group is $\Go \doteq \langle\exp\fGo\rangle$.
Now, if $G$ is Abelian, Proposition \ref{prop:discretequotientgroup}(2) tells us that
\[
\fGo \simeq \fg^\Sigma_\rel + \fg,
\]
and
\[
\Go \doteq \langle \exp \fGo \rangle \simeq (G\cdot G^\Sigma_\rel)_0 = G\cdot G^\Sigma_{\rel,0}.
\]
The next proposition provides a more explicit characterisation of $\Go$:

\begin{proposition}[Abelian constraint gauge group $\Go$]\label{prop:constr-g-grp-Ab}
The Abelian constraint gauge group $\Go\doteq \langle\mathrm{exp}\fGo \rangle$ is isomorphic to 
\[
\Go \simeq G\cdot G^\Sigma_{\rel,0} = \left\{g\in G \cdot G^\Sigma_\rel \ : \  w(g) \doteq \tfrac{1}{2\pi} (g^{-1}\pp_u g)^{\int} = 0 \right\}.
\]
\end{proposition}

\begin{proof}
We prove the statement by showing that the two inclusions. 

Lemma \ref{lemma:winding} states that  $w(g)$ is valued in an integer lattice $(i\mathbb{Z})^{p}\times \{ 0\}\subset \fg\hookrightarrow\fg^S$. But then, since $w(g)\in\fg^S$ is continuous over $S$,  $w(g)$ is constant on $S$. Additionally, $w\colon G \cdot G^\Sigma_\rel\to \fg$ is a constant function over the connected components of $G \cdot G^\Sigma_\rel$---and in particular it vanishes on the identity component. Therefore:
\[
G\cdot G^\Sigma_{\rel,0} \subset \{g \in G \cdot G^\Sigma_\rel \ : \ w(g) = 0 \}.
\]

Next, let $g\in G\cdot G^\Sigma_\rel$---so that $g^\i=g^\f\equiv k\in G$---and assume $w(g)=0$. For all $t\in[0,1]$, define
\[
g_t \doteq k \exp t \int_{-1}^u du' \, g^{-1}L_\ell g(u').
\]
Then, $g_{t=0}= g^\i$ and $g_{t=1} = k (g^{-1})^\i g = g$. Moreover, since $w(g) = 0$, we have that $g_t^\f = k \exp(2\pi t w(g) ) = k$. Therefore for all $t\in[0,1]$, $g_t^\i = g_t^\f = k\in G$ and hence $g_t \in G\cdot G^\Sigma_\rel$. We thus conclude that $g_t$ is an homotopy in $G\cdot G^\Sigma_\rel$ between $g$ and a constant $k\in G$. Since $G$ is by hypothesis connected, $g$ is not only homotopic within $G\cdot G^\Sigma_\rel$ to $k$ but also to the identity $1\in G\cdot G^\Sigma_\rel$. Thus,
\[
G\cdot G^\Sigma_{\rel,0} \supseteq \{g \in G \cdot G^\Sigma_\rel \ : \ w(g) = 0 \}.\qedhere
\]
\end{proof}

\begin{remark}[Group homomorphism]\label{rmk:grphomo}
Recall Notation \ref{not:AbLiegp}. If $t=0$ then it is easy to see that any $g\in G\cdot G^\Sigma_\rel$ is homotopic to the identity, and thus $\Go =G\cdot G^\Sigma_\rel$. On the other hand, if $G=\mathrm{U}(1)$ and $n\in\mathbb{Z}$, then the only element $g^{(n)}(u,x) = -\exp( \pi n i u)\in G\cdot G^\Sigma_\rel$ that is also in the identity component $\Go = G\cdot G^\Sigma_{\rel,0}= \exp\fGo $ is the identity itself, i.e.\ $g^{(0)}$.\footnote{Indeed, $\xi \in \fGo$ iff $\xi^\i = \xi^\f \in \fg$, and $ -\pi i n = \pi i n$ iff $n=0$.}
The quantity $w(g)$ is a topological invariant classifying the connected components of $G\cdot G^\Sigma_\rel$; in particular $w(g^{(n)}) = i n$. Moreover, $w(g_1g_2)=w(g_1)+w(g_2)$. Indeed, if $S =S^{n-1}$ and $n\neq 2$, the winding number $w(g)$ provides a group homomorphism between $G\cdot G^\Sigma_\rel$ and the group of components of $G\cdot G^\Sigma_\rel \simeq G^{S^1\times S} \simeq \mathbb{Z}^t$.
\end{remark}

\subsection{Gauge fixing}

We now turn our attention to the action of the constraint gauge group on the constraint surface $\C\subset \X$. 
Since $\fGo\simeq \fg^\Sigma_\rel + \fg$ ``differs'' from $\fg^\Sigma_\rel$ only by the ($\C$-global) isotropy $\ker(\rho) = \fg \hookrightarrow \fG$, one has
\[
\uCo \doteq \C / \Go = \C / G^\Sigma_{\rel,0}.
\]

Recall Proposition \ref{prop:constr-surf} which states that 
\[
\C \simeq \mathcal{A}\times(\fg^S)^*.
\]

\begin{lemma}\label{lemma:Grelaction}
The action of $G^\Sigma_{\rel,0}$ on $\C$ induces the following action of $G^\Sigma_{\rel,0}$ on $\mathcal{A}\times{(\fg^S)^*}$:
\[
G^\Sigma_{\rel,0} \circlearrowright \mathcal{A}\times{(\fg^S)^*}, 
\qquad 
(g_\circ, A, {E_\i}) \mapsto (A^{g_\circ}, {E_\i}).
\]
This action defines the principal $G^\Sigma_{\rel,0}$-bundle
\[
\underline{\pi}_\circ :\mathcal{A}\times{(\fg^S)^*} \to (\mathcal{A}\times{(\fg^S)^*})/G^\Sigma_{\rel,0},
\]
with $(\mathcal{A}\times{(\fg^S)^*})/G^\Sigma_{\rel,0} = (\mathcal{A}\times{(\fg^S)^*})/\Go \simeq \uCo$. Diagrammatically:
\[
\xymatrix{
\mathcal{A} \times (\fg^S)^* \ar[r]_-{s_\i}^-{\simeq} \ar@<-4pt>[d]_{\underline\pi_\circ}
&
\C\; \ar[d]^{\pi_\circ}\\
(\mathcal{A} \times (\fg^S)^*)/ G^\Sigma_{\rel,0} \ar[r]^-{\ \simeq \ ~}
& 
\uCo 
}
\]
\end{lemma}
\begin{proof}
The only step of the proof that is not obvious and requires some care is proving that the action of $G^\Sigma_{\rel,0}$ is free on $\mathcal{A}$ and therefore on $\mathcal{A}\times{(\fg^S)^*}$. This follows from $\ker(\rho) = \fg\hookrightarrow \fg^\Sigma$ and $\fg\cap\fg^\Sigma_\rel = \{0\}$. 
\end{proof}

Now, since $\C\simeq \mathcal{A}\times(\fg^S)^*$ is a principal $G^\Sigma_{\rel,0}$-bundle and $\uCo \doteq \C/\Go \simeq \C/G^\Sigma_{\rel ,0}$, the question of reduction can be directly addressed if we can find a (global) gauge fixing of the action of $G^\Sigma_{\rel,0}$ on $\C$, i.e.\ a (global) section of the principal $G^\Sigma_{\rel,0}$-bundle $(\mathcal{A}\times(\fg^S)^*, \underline\pi_\circ)$. We proceed in two steps: first we introduce a trivialisation of this principal bundle, and then leverage it to construct a global section.

\begin{lemma}\label{lemma:constr-trivialis-Ab}
Viewing $ G^\Sigma_{\rel,0} \subset G^\Sigma \simeq C^\infty(I,G^S)$, let
\[
\dfuo:\mathcal{A}\to \fg^\Sigma \simeq C^\infty(I,\fg^S), 
\quad 
A \mapsto (\dfuo(A))(u)\doteq - \int_{-1}^u du' \,A_\ell(u') + \tfrac12 (u+1) (A_\ell)^{\int}.
\]
Then, the following map is a principal $G^\Sigma_{\rel,0}$-bundle trivialisation: 
\[
\begin{array}{rl}
\underline{\tau}_\circ: &\mathcal{A}\times(\fg^S)^* \to \Xas \times G^\Sigma_{\rel,0} \equiv \hcalA \times T^*\fg^S \times G^\Sigma_{\rel,0}\\\\
& \,\;
\begin{pmatrix}
A\\ {E_\i}
\end{pmatrix}
\;\mapsto 
\begin{pmatrix}
\ASa(A)\\ {\lambda(A)} \\ \ASe(E_\i) \\ \dfUo(A)
\end{pmatrix}
\doteq 
\begin{pmatrix}
\hat{A} + D\dfuo(A) - \tfrac12 (u-1) D(A_\ell)^{\int}\\ (A_\ell)^{\int}\\
{E_\i}\\
\exp\dfuo(A)
\end{pmatrix}
\end{array}
\]
where ${\ASe}\in T_{\lambda}^*\fg^S\simeq (\fg^S)^*$ (cf. Remark \ref{rmk:dual}). Diagrammatically:
\[
\xymatrix{
\Xas \times G^\Sigma_{\rel,0} \ar[dr]_{\pi_{12}} &
\mathcal{A} \times (\fg^S)^* \ar@<-4pt>[d]^{\underline\pi_\circ}
\ar[l]_-{\underline{\tau}_\circ}^-{\simeq}\\
&\Xas
}
\]
where $\pi_{12}$ denotes the projection on the first two factors, and we have $\underline{\pi}_\circ = \pi_{12} \circ \underline{\tau}_\circ$.
\end{lemma}

\begin{proof}
Proving the lemma requires showing that (1) $\underline{\tau}_\circ$ is invertible, and (2) it maps right-translations of the $G^\Sigma_{\rel,0}$-factor in $\hcalA\times T^*\fg^S\times {G^\Sigma_{\rel,0}}$ onto the action of ${G^\Sigma_{\rel,0}}$ on $\mathcal{A}\times {(\fg^S)^*}$. 
Throughout the proof, it is useful to keep in mind that $\lambda(A)$ is constant in $u$.

To prove (1), it is enough to show that the following map is indeed both a left and a right inverse of $\underline\tau_\circ$:
\begin{align}\label{eq:tau-1}
\underline\tau_\circ^{-1} :\hcalA \times T^*\fg^S \times {G^\Sigma_{\rel,0}} &\to \mathcal{A}\times{(\fg^S)^*}\\
(\ASa,\lambda,\ASe,\dfUo) &\mapsto 
\begin{pmatrix}
    A \\ E_\i
\end{pmatrix} = 
\begin{pmatrix}
\ASa + \tfrac12(u-1)D\lambda + \tfrac12\lambda du - d\dfUo \dfUo^{-1}\\ \ASe
\end{pmatrix}. \notag
\end{align}

Proving the left-inverse property is straightforward once one notices that the first component ${\ASa(A)}$ of $\underline{\tau}_\circ(A,{E_\i})$ satisfies
\[
 {\ASa(A)}  + \tfrac12(u-1)D\lambda(A) + \tfrac12\lambda(A) du = A  + \dfUo(A)^{-1} d\dfUo(A) .
\]
We leave the details to the reader. 

Proving the right-inverse property is instead subtler. For this we need some preliminary results. Denote the first component of $\underline\tau_\circ^{-1}(\ASa,\lambda,\ASe,U_\circ)$ by 
\[
A'=A'(\ASa,\lambda,\ASe,\dfUo) \doteq \ASa + \tfrac12(u-1)D\lambda + \tfrac12\lambda du - d\dfUo \dfUo^{-1}.
\]
Then,
\[
A'_\ell = \tfrac12 \lambda - L_\ell \dfUo \dfUo^{-1}. 
\]
From these, recalling that $\lambda$ is constant in $u$, we compute:
\[
\lambda(A') \doteq ( A'_\ell)^{\int} = \lambda  
- (L_\ell \dfUo \dfUo^{-1})^{\int}
= \lambda
\]
where we used the fact that $U_\circ\in G^\Sigma_{\rel,0}$ and therefore $ w(U_\circ)= \frac{1}{2\pi}(U_\circ^{-1}L_\ell U_\circ)^{\int}=0$ (Proposition \ref{prop:constr-g-grp-Ab}).
Similarly, we find
\[
(\dfuo(A'))(u)\doteq  - \int_{-1}^u du'\, A'_\ell + \tfrac12(u+1) (A'_\ell)^{\int}  = - \int_{-1}^u du'\, U_\circ^{-1}L_\ell U_\circ(u'),
\]
and thus (see also Appendix \ref{app:usefulformulas})
\[
\dfUo(A') \doteq \exp\dfuo(A') = \dfUo (\dfUo^\i)^{-1}= \dfUo.
\]
Now, with these two formulas, it is easy to see that $\underline\tau_\circ^{-1}$ is indeed a right inverse of $\underline{\tau}_\circ$ as well. This concludes the proof of (1). 

Finally, to prove (2) we appeal to (1) which tells us that it is enough to prove that $U_\circ\mapsto g_\circ U_\circ $ maps $A' \mapsto A' + g_\circ^{-1}d g_\circ$, a fact that is obvious when looking at the explicit form of $\underline\tau_\circ^{-1}$.
\end{proof}

Now that we have an explicit, global, trivialisation $\underline{\tau}_\circ$ of the principal $G^\Sigma_{\rel,0}$-bundle $\underline{\pi}_\circ:\mathcal{A}\times (\fg^S)^*\to \hcalA \times T^*\fg^S$, it is immediate to produce one of its (global) sections, $\underline{\sigma}_\circ$:

\begin{corollary}\label{cor:Gobundles}
The map $\underline{\sigma}_\circ: \Xas \to \mathcal{A}\times(\fg^S)^*$ defined as
\[
\underline{\sigma}_\circ \doteq \underline{\tau}^{-1}_\circ( \cdot,\cdot,\cdot,1),
\]
where $1 \in G^\Sigma_{\rel,0}$  (cf. Equation \eqref{eq:tau-1}), is a section of $\mathcal{A}\times (\fg^S)^* \xrightarrow{\underline{\pi}_\circ} \Xas$.
\end{corollary}

\subsection{Abelian constraint reduction} 

Recall now that $\Xas$ is equipped with the symplectic structure  $\omeas\doteq\omAS  + \omega_S$, and that $\mathcal{A}\times (\fg^S)^*\simeq \C \subset \X$:

\begin{lemma}\label{lemma:sigmacircX}
Define the embedding map $\sigma_\circ^\X \doteq \iota_\C \circ s_\i \circ  \underline\sigma_\circ$ as per the commutative diagram
\[
\xymatrix{ \mathcal{A} \times (\fg^S)^* \ar[r]^{\quad s_\i}_{\quad \simeq}
&
\C\; \ar@{^(->}[r]^{\iota_\C}
& 
\X\\
\hcalA\times T^*\fg^S\; \ar@{^(->}[u]^{\quad\underline{\sigma}_\circ} \ar@/_1pc/[urr]_{\sigma^\X_\circ}
}
\]
Then, $\sigma^\X_\circ$ is a symplectic embedding, i.e.\ $(\sigma^\X_\circ)^*\omega = \omeas $ or, more explicitly,
\begin{align*}
(\sigma_\circ^\X)^*\omega(\ASa,\lambda,\ASe) = \omeas(\ASa,\lambda,\ASe)
&\doteq \omAS (\ASa) + \omega_S(\lambda,\ASe) \\
&= \int_\Sigma \btr( L_\ell \bd \ASa^i \wedge \bd \ASa_i) + \int_S \btr( \bd \ASe \wedge \bd \lambda) .
\end{align*}
\end{lemma}
\begin{proof}
This can be proved by inserting the following formula for $(A,E)$ into $\omega(A,E)$,
\[
(A,E) = \sigma^\X_\circ(\ASa,\lambda, \ASe)
= s_\i(\ASa +\tfrac12(u-1)D\lambda(A) + \lambda du , \ASe)
\in \C \subset \X,
\]
and recalling that $E(A,E_\i)$ is by construction a solution of the Gauss constraint $\mathsf{G}=0$: in our notation, $s_\i(A,{E_\i}) \equiv (A, E(A,E_\i))$ (Proposition \ref{prop:constr-surf}).

A more conceptual proof along the lines of that provided for Calculation \ref{prop:step2} (see Appendix \ref{app:compute-bom-dress}) can be found in v1 of this article on the arXiv repository. The advantage of this alternative proof is that it shines a light on the relationship between symplectic reduction, the dressing field method, and its close relation to ``edge modes''.
\end{proof}

Combining these results, we can finally prove:

\begin{theorem}[Eq. \eqref{e:AbelianASreduction} of Theorem \ref{thm:nonAb-constr-red}]\label{thm:nullconstraintred}
For $G$ Abelian, the constraint-reduced phase space $(\uCo,\uomegao)$ of Yang--Mills theory at a null boundary $\Sigma$, which is defined by
\[
\uCo \doteq \C /\Go \quad\mathrm{and}\quad \pi_\circ^*\uomegao \doteq \iota_\C^*\omega,
\]
is symplectomorphic to the linearly-extended Ashtekar--Streubel phase space,
\[
(\uCo, \uomegao) \simeq (\Xas,\omeas) \doteq (\hcalA\times T^*\fg^S, \omAS  +\omega_S),
\]
by means of the map $\underline{s}_\i \doteq \pi_\circ\circ s_\i \circ \underline\sigma_\circ: \hcalA \times T^*\fg^S \to \uCo$.

Moreover, 
\[
(\uCo, \uomegao) \simeq_\loc (\XAS,\omAS) \doteq (\hcalA \times T^*G^S_0, \omAS  + \Omega_S),
\]
by means of the map\footnote{If $\exp$ is not a global diffeomorphism, the logarithm is defined only locally. This happens if $G$ has at least a $\mathrm{U}(1)$ factor.} $\underline{s}^\mathrm{log}_\i \doteq \underline{s}_\i\circ(\mathrm{id}_{\hcalA}, \log_{G^S_0}, \mathrm{id}_{(\fg^S)^*}): \hcalA\times T^*G^S_0\to \uCo$.
Explicitly,
\begin{align*}
(\underline{s}_\i^\mathrm{log})^*\uomegao(\ASa,\Lambda,\ASe) = \omeAS(\ASa,\Lambda,\ASe)
& \doteq \omAS (\ASa) + \Omega_S(\Lambda, \ASe) \\
& = 
\int_\Sigma \btr(L_\ell \bd \ASa{}^i\wedge\bd \ASa{}_i)+ \int_S \btr( \bd \ASe\wedge \bd \Lambda \Lambda^{-1}) .
\end{align*}
\end{theorem}

\begin{proof}
The theorem is a consequence of the previous lemmas and the following commutative diagram:
\[
\xymatrix{
\mathcal{A} \times (\fg^S)^* \ar[r]^{\quad s_\i}_{\quad \simeq}\ar@<-4pt>[d]_{\underline\pi_\circ}
&
\C\; \ar@{^(->}[r]^{\iota_\C}\ar[d]^{\pi_\circ}
& 
\X \ar@/^1pc/@{~>}[dl]^{~~\tbox{1.2cm}{\vspace{-2em}~\\constraint\\ reduction}}\\
\Xas\; \ar[r]^{\;\quad\underline{s}_\i}_{\;\quad\simeq}  \ar@<-6pt>@{^(->}[u]_{\underline{\sigma}_\circ}
& 
\uCo 
}
\]
Indeed, Corollary \ref{cor:Gobundles} tells that the first and second columns of the diagram are diffeomorphic principal $G^\Sigma_{\rel,0}$-bundles.

Then, using Lemma \ref{lemma:sigmacircX}, we compute: 
\begin{align*}
    \omeas 
     = (\sigma_\circ^\X)^*\omega 
    & \equiv (\iota_\C\circ s_\i\circ\underline\sigma_\circ)^*\omega\\
    & = (s_\i \circ\underline\sigma_\circ)^*\iota_\C^*\omega = (s_\i \circ\underline\sigma_\circ)^*\pi_\circ^*\uomegao\\
    & = (s_\i \circ\underline\sigma_\circ)^*(\underline{s}_\i\circ\underline{\pi}_\circ \circ s_\i^{-1})^*\uomegao\\
    & = ( \underline{s}_\i \circ \underline{\pi}_\circ \circ \underline{\sigma}_\circ)^*\uomegao 
    = \underline{s}_\i^* \uomegao
\end{align*}
from which we obtain the sought identity:
\[
 \omeas = \underline{s}_\i^*\uomegao .
\]
The remainder of the proof is obvious.
\end{proof}

\section{Summary of notations}\label{app:notations}

\subsection*{Geometry and differential calculus on spacetime}
~

\begin{longtable}{{p{2cm} p{10cm}   }}
   $M \hookleftarrow \Sigma$  & $M$ is the smooth, orientable, spacetime (i.e. Lorentzian) manifold, and $\Sigma$ a codimension-1 manifold inside it. Possibly, $\pp\Sigma\neq \emptyset$. 
   Section \ref{sec:nullmfd}.\\
   $n \geq 1$ & The (finite!) dimension of $\Sigma$. Section \ref{sec:nullmfd}.\\
   $L, i, d$ & Lie, interior, and exterior derivatives of differential forms over $\Sigma$.\\
   $\gamma, \ell$ & $\gamma$ is the metric over $\Sigma$. $\Sigma$ is said \emph{null} iff there exists a vector field $\ell$ such that $\ker(\gamma: T\Sigma \to T^*\Sigma) = \mathrm{Span}(\ell)$. We assume that the non-null eigenvalues of $\gamma$ are positive, and (for simplicity) that $L_\ell \gamma =0$. Section \ref{sec:nullmfd}. (C.f. the entry about $(u,x^i)$ for more information.) \\
   $\Sigma \simeq I \times S$ & When $\Sigma$ is null, we take it to be a cylinder: $I = [-1,1]$ and $\pp S=\emptyset$; typically, $S$ is an $(n-1)$-sphere. Section \ref{sec:nullmfd}.\\
   $\bullet_\i, \bullet_\f$ & These subscripts highlight that a map or field is defined on $S_{\i/\f} = S \times \{u=\mp1\} \hookrightarrow \Sigma$, which stand for the ``initial and final'' spheres.\\
   $\bullet^\i, \bullet^\f$ & These superscripts denote the evaluations of a quantity $\bullet$ defined on $\Sigma$, at $u=\mp1$ respectively. In most cases $\bullet_{\i/\f}$ and $\bullet^{\i/\f}$ can be conflated; however, note that objects defined on $S_{\i/\f}$ may not extend to $\Sigma$. Notation \ref{not:subsupscripts}.\\
   $\bullet^{\int}, \bullet^{\av}, \bullet^{\dif}$ & Different types of ``zero modes'' for functions of $u\in I$: namely $Q^{\int} \doteq \int_{-1}^1 du'\, Q(u')$, $Q^\av  \doteq  \tfrac12 (Q^\i + Q^\f)$, and $Q^\dif \doteq  Q^\f - Q^\i$. Notation \ref{not:subsupscripts}.\\
   $x^a=(u,x^i)$ & Coordinates on $\Sigma=I\times S$. We assume that $\ell=\pp_u$. Since $\gamma(\ell)=0$ and $L_\ell\gamma=0$, $\gamma$ can be identified with a spacelike metric over $S$. Section \ref{sec:nullmfd}.\\
   $\vol_S, \ \vol_\Sigma$ & The volume forms over $S$ and $\Sigma$ respectively. Note that $\vol_S$ is a metric volume form, while $\vol_{\Sigma}=du \wedge \vol_S$ is not. Notation \ref{notation:vol}.\\
   $\bullet_\ell$ & Contraction of a differential form over $\Sigma$ with the vector $\ell$, i.e. $\alpha_\ell = i_\ell \alpha$. Section \ref{sec:nullmfd}.\\
   $\hat{\bullet}$ & Projector on the space of \emph{spatial} tensors, i.e.\ forms $\alpha$ and vectors $v$ on $\Sigma$ such that $i_\ell\alpha = 0$ and $i_v du = 0$. Definition \ref{def:spatial-f+v}.\\
   $D$ & Exterior derivative on the space of spatial forms. It can be identified with the de Rham differential over $S$. Definition \ref{def:spatial-f+v}.\\
   $G, \fg$ & $G$ is a real, connected, Lie group assumed to be either Abelian or semisimple. We denote $\fg = \mathrm{Lie}(G)$ its Lie algebra. Section \ref{sec:geomphsp}.\\
   $\mathrm{tr}(\cdot \cdot)$ & Non-degenerate, $\Ad$-invariant, bilinear form on $\fg$. Section \ref{sec:geomphsp}.\\
   $\int_N\mathbf{tr}( \cdot\cdot )$ & Short-hand notation for $\int_N \mathrm{tr}(\cdot\cdot) \vol_N$. Notation \ref{notation:tr}\\
   $P \to \Sigma$ & Principal $G$-bundle over $\Sigma$, assumed to be trivial, $P\simeq \Sigma \times G$. Section \ref{sec:geomphsp}, Assumption \ref{ass:trivialP}.\\
   $\mathcal{A} \ni A$ & The space of connections on the (trivial) principal $G$-bundle $P$. We identify $\mathcal{A} \simeq \Omega^1(\Sigma,\fg)$. This is the space of ``gauge fields'' on $\Sigma$. The spatial/null decomposition of $\mathcal{A}$ is denoted $\mathcal{A} \simeq \hcalA \times \mathcal{A}_\ell$, with $A = \hat A + A_\ell du$. Section \ref{sec:geomphsp}, Assumption \ref{ass:trivialP}, Remark \ref{rmk:u-functions}.\\
   $\mathcal{L}_\ell, \mathcal{D}$ & Gauge-covariant analogues of $L_\ell$ and $D$, i.e. $\mathcal{L}_\ell = L_\ell + A_\ell$ and $\mathcal{D} = D + \hat A$. Section \ref{sec:geomphsp}.\\
   $F$ & The curvature (field strength) of $A$, $F = dA + A\wedge A$. Section \ref{sec:geomphsp}.
\end{longtable}

~\\

\subsection*{Geometry and differential calculus on field space.}
~
\begin{longtable}{{p{2cm} p{10cm}   }}
$\mathcal{W}^*_\text{str} , \mathcal{W}^*_\loc{}, \mathcal{W}^*$ & Let $\mathcal{W}=\Gamma(W\to\Sigma)$ be the (nuclear Fréchet) space of sections of a real vector bundle over a compact manifold. $\mathcal{W}^*_\text{str}$ is the topological dual endowed with the strong topology. $\mathcal{W}^*_\loc\subset \mathcal{W}^*_\text{str}$ is the local dual while $\mathcal{W}^*\subset \mathcal{W}^*_\loc$ denotes the densitised dual. Remark \ref{rmk:dualspaces}, Appendix \ref{rmk:locality}.\\
$\Ann(\mathcal{X},\mathcal{Y})$ & If $\mathcal{X}\subset \mathcal{W}$ and $\mathcal{Y}\subset \mathcal{W}^*_\text{str}$, $\Ann(\mathcal{X},\mathcal{Y}) = \{ y \in \mathcal{Y} : \langle y, x\rangle = 0 \ \forall x \in \mathcal{X} \} \subset \mathcal{W}^*_\text{str}$. Definition \ref{def:annihilators}, Lemma \ref{lem:doubleannihilator}.\\
$G^N, G^N_\text{rel}$ & The mapping group $G^N=C^\infty(N,G)$ equipped with pointwise multiplication and the relative mapping group: the subgroup $G^N_\text{rel}\subset G^N$ of functions that are equal to the identity (in $G$) at the boundary of $N$, $g|_{\pp N} = 1$. Definitions \ref{def:mappinggroups} and \ref{def:id+rel-comp}.\\
$\fg^N, \fg^N_\text{rel}$ & As above, but for Lie algebras. Definition \ref{def:mappinggroups}.\\
$\bullet_0$ & Denotes the identity component of a group, e.g. $G_0$ is the identity component of $G$ and $G^N_{\text{rel},0}$ is the identity component of the relative mapping group. Since $G$ is connected by assumption, we write $G^N_0 \equiv (G^N)_0$. Definition \ref{def:id+rel-comp}.\\
$\L, \bi, \bd$ & Lie, interior, and exterior derivatives of differential forms over an infinite dimensional manifold. (C.f. the discussion in \cite[Remark A.5]{RielloSchiavina}.) \\ 
$\mathcal{E}\ni \bE, \ E$ & $\mathcal{E}=\Omega^{n-1}_\text{spatial}(\Sigma, \Ad^*P)\simeq\Omega^{n-1}_\text{spatial}(\Sigma,\fg^*)$ is the space of electric fields. We identify $\bE \in \mathcal{E}$ with a $E\in C^\infty(\Sigma,\fg)$ via $\mathbf{E} = \tr(E \cdot) \vol_S$. Section \ref{sec:geomphsp}, Assumption \ref{ass:trivialP}, Remark \ref{rmk:u-functions}.\\ 
$(\X, \bom_\nYM)$ & $\X = \mathcal{A} \times \mathcal{E}$, and $\bom_\nYM$ is a (weak) symplectic density on it, i.e. $\bom_\nYM \in \Omega^{\text{top},2}_\loc(\Sigma\times\X)$. This means that $\omega_\nYM = \int_\Sigma \bom_\nYM$ is a (weak) symplectic two-form on $\X$. The symplectic space $(\X, \omega_\nYM)$ is the off-shell, ``geometric'', phase space of null-YM theory. Definition \ref{def:geom-ph-sp}.\\
$\bom, \omega$ & Shorthand notation for $\bom_\nYM$ and $\omega_\nYM$.\\
$(\hcalA, \omAS)$ & The Ashtekar--Streubel phase space \cite{AshtekarStreubel}. See the entry for $\mathcal{A}$ and Definition \ref{def:hatA-sympl}.\\
$\ASa$ & Configuration in $\hcalA$. Definition \ref{def:hatA-sympl}.\\
$\tilde\varphi(k),\ \tilde\ASa_i(k,x)$ & Ashtekar--Streubel $k$-mode, $k\in\mathbb{N}\cup\{0\}$, of a function $\varphi\in C^\infty(I)$ and of the Ashtekar--Streubel field $\ASa(u,x)$, respectively. Lemma \ref{lemma:ASmodes} and Proposition \ref{prop:modedecomp}.\\
$(\XAS, \omeAS)$ & The extended Ashtekar--Streubel phase space, $\XAS = \hcalA \times T^*G^S_0$ and $\omeAS = \omAS + \Omega_S$ where $\Omega_S$ is the canonical symplectic form on the densitised cotangent bundle $T^*G^S_0$  (cf. Remark \ref{rmk:dualspaces}). Definition \ref{def:extAS}, Theorem \ref{thm:nonAb-constr-red}. \\
$(\Xas, \omeas)$ & The linearly extended Ashtekar--Streubel phase space, $\Xas = \hcalA \times T^*\fg^S$ and $\omeAS = \omAS + \omega_S$ where $\omega_S$ is the canonical symplectic form on the densitised cotangent bundle $T^*\fg^S$ (cf. Remark \ref{rmk:dualspaces}). Definition \ref{def:extAS}, Theorem \ref{thm:nonAb-constr-red}.\\
$(\ASa, \Lambda, \ASe) $ & Configurations in $\XAS\simeq \hcalA \times G^S_0 \times (\fg^S)^*$.\\
$(\ASa, \lambda, \ASe) $ & Configurations in $\Xas\simeq \hcalA \times \fg^S \times (\fg^S)^*$.\\
$\underline\varrho$ & The action of $\fg^{\pp\Sigma}$ on $\XAS$, defined in Lemma \ref{lem:residualAction}.\\
$\mathcal{G}, \fG, \rho, \bH$ & $\G$ is the (connected) gauge group acting (in a locally Hamiltonian way) on $\X$, and $\fG$ its Lie algebra. In null-YM theory we take $\fG = \fg^\Sigma$ and $\G = G^\Sigma_0$. The action $\rho: \fG \times \X \to T\X$ is a locally Hamiltonian action on $(\X,\bom_\nYM)$ given by $(\xi,(A,E)) \mapsto (d_A \xi, [E,\xi])$. Its momentum form is the linear, equivariant, local map $\bH : \fG \to \Omega^{\text{top},0}(\Sigma\times\X)$. Definitions \ref{def:fluxes+constr-red}, \ref{def:gaugetransformations} and \ref{def:gauge-action}; Proposition \ref{prop:nYMisLocHam}.\\
$\mathcal{K}$ & The (discrete) group of components  $\G_\rel/(\G_\rel)_0$. Theorem \ref{thm:fibrecharacterisation}.\\
$\bHo,d\bh, \mathsf{G}$ & According to \cite[Prop. 2.2 and Def 4.1]{RielloSchiavina} the momentum form $\bH$ uniquely decomposes into the sum of a \emph{constraint form} $\bHo$ (such that $\langle \bH, \xi\rangle$ does not involve any derivative of $\xi$) and of a (boundary) \emph{flux form} $d \bh$, i.e. $\bH = \bHo + d \bh$. Finally, $\mathsf{G}(A,E) = \mathcal{L}_\ell E + \mathcal{D}^i F_{\ell i}$ is the Gauss constraint, $\bHo = \tr(\mathsf{G} \cdot) \vol_\Sigma$. Definitions \ref{def:LHGT} and \ref{prop:constr-surf}, Proposition \ref{prop:nYMisLocHam}.\\
$H,\Ho, h $ & These are the integrals over $\Sigma$ of $\bH$, $\bHo$, and $d\bh$ respectively, e.g. $\langle h,\xi\rangle = \int_\Sigma \langle d\bh , \xi\rangle$. One can think of these maps as being valued in $\fG^*_\text{str}$. $\Ho$ is equivariant, and in null-YM so is $h$. Proposition \ref{prop:nYMisLocHam}.\\
$\F\ni f$ & The (on-shell) flux space $\F = \Im(\iota_\C^* h) \subset \fG^*_\loc$. Its elements $f$ are called fluxes. Definition \ref{def:fluxes+constr-red}, Propositions \ref{prop:resHamaction} and \ref{prop:fluxes}.\\
$(\mathcal{O}_f, \Omega_{[f]})$ & In null-YM, where $h$ is equivariant, $\mathcal{O}_f$ denotes the coadjoint orbit of $f \in \F \subset \fG^*_\text{str}$. One has $\mathcal{O}_f \subset \F$. Furthermore, $\Omega_{[f]}$ is the Kirillov--Konstant--Souriau symplectic structure on $\mathcal{O}_f$. See \cite{RielloSchiavina} for more on the definition of $\Omega_{[f]}$ in infinite dimensions.\\
$\fGo, \fGred, \Go, \Gred$ & $\fGo = \Ann(\F, \fG) = \{\xi \in \fG \ : \ \langle h(A,E), \xi\rangle = 0 \ \forall (A,E)\in \C\subset \X\}$ is the maximal constraint Lie ideal (see \cite[Thm 4.33]{RielloSchiavina}). The flux gauge algebra is the quotient $\fGred = \fG/\fGo$. The constraint group $\Go= \langle \exp \fGo\rangle$, is the normal Lie subgroup generated by $\fGo$, while the flux gauge group $\Gred = \G/\Go$ is a Lie group with Lie algebra $\fGred$. See Definition \ref{def:fluxes+constr-red}, including Footnote \ref{fnt:fGo}, as well as Propositions \ref{prop:Go-normal} and \ref{prop:resHamaction}. For the local nature of $\fGred$, see Remark \ref{rmk:Gredlocality}. For specifics in the case of null YM theory: Proposition \ref{prop:discretequotientgroup}\\
$\C\hookrightarrow\X$ & The constraint set $\C = \bHo^{-1}(0) = \{ (A,E)\in \X : \mathsf{G}(A,E)=0\}$. Proposition \ref{prop:constr-surf}\\
$V$ & The dressing field, $V\colon\X \to C^\infty(I,G^S_0)$. It can be interpreted as a (family of) Wilson lines of $A$ along the flow of $\ell$ starting from the initial sphere $S^\i \hookrightarrow \Sigma$. Cf. Definition \ref{def:dressingfield}.\\
$\check{V}$ & The on-shell dressing map, $\check{V}:\C\to \XAS$. Definition \ref{def:dressingmap}. \\
$p_V$ & A map $\uCo \to \XAS$ such that $\check{V} = \pi_\circ^*p_V$ ($\pi_\circ$ is defined below). Theorem \ref{thm:fibrecharacterisation}\\
$\underline\bullet, \underline{\underline{\bullet}}$ & As a rule of thumb, one underscore refers to (first stage) constraint reduction, i.e. to quantities defined modulo to the action of $\Go$. Similarly, two underscores refer to (second stage) flux superselection, i.e. to quantities defined modulo the action of $\G$ or, heuristically, modulo the subsequent actions of $\Go$ and $\Gred$.\\
$\C^\omega$ & The characteristic foliation of $\C$. \\
$(\uCo, \uomegao), \pi_\circ$ & The constraint-reduced phase space, $\uCo = \C/\C^\omega \simeq \C/\Go$. The projection corresponding to this quotient is denoted $\pi_\circ : \C \to \uCo$. And $\uomegao$ is the (Marsden--Weinstein) reduced symplectic structure on $\uCo$, defined by $\iota_\C^*\omega = \pi_\circ^*\uomegao$.\\
$\uh, \underline{\rho}$ & The reduced flux map, $\uh : \fGred \to C^\infty(\uCo)$. Intuitively, it can be thought of as the ``projection'' to $\uCo$ of the flux map $\h:\fG \to C^\infty(\C)$. It is a momentum map for the (residual) action of $\fGred$ on $(\uCo, \uomegao)$, denoted $\underline{\rho} : \fGred \times \uCo \to T\uCo$. Propositions \ref{prop:resHamaction}, \ref{prop:red-flux-map}.\\
$\underline{\iota}_{[f]}, \ \underline{\pi}_{[f]}$ & Respectively, the embedding $\uh^{-1}(\mathcal{O}_f) \hookrightarrow \uCo$ and the projection $\uh^{-1}(\mathcal{O}_f) \to \uh^{-1}(\mathcal{O}_f)/\Gred = \uuS_{[f]}$.\\
$(\uuC, \uuPi), \underline{\pi}$ & The fully reduced phase space, $\uuC = \C/\G \simeq \uCo/\Gred$, with its (partial) Poisson bivector $\uuPi$ \cite{RielloSchiavina}. The associated (second stage) quotient map is denoted $\underline{\pi} : \uCo \to \uuC$. Definition \ref{eq:fully-red-space}, Theorems \ref{thm:generalsuperselection}, \ref{thm:superselections}, and \ref{thm:modelHamreduction}.\\
$(\uuS_{[f]}, \uuomegao_{[f]})$ & The flux superselection sector associated to $\mathcal{O}_f$, is the symplectic manifold $\uuS_{[f]} = \uh^{-1}(\mathcal{O}_f)/\Gred \subset \uuC$ equipped with the unique symplectic structure such that $\underline{\pi}^*_{[f]} \uuomegao_{[f]} = \underline\iota_{[f]}^*\big( \uomegao - \uh^*\Omega_{[f]}\big)$.
The superselection sectors provide a symplectic foliation of $(\uuC,\uuPi)$, $\uuC \simeq \bigsqcup_{\mathcal{O}_f\subset \F}\uuS_{[f]}$. This turns $\uuC$ into a (partial) Poisson manifold \cite{RielloSchiavina}. Lemma \ref{lem:fluxSS}, Theorems \ref{thm:generalsuperselection}, \ref{thm:superselections}, and \ref{thm:modelHamreduction}.\\
\end{longtable}

The following diagram summarises the two-stage reduction procedure:
\begin{equation*}
\xymatrix@C=.75cm@R=1cm{
(\X,\omega)
    \ar@{~>}[rr]^-{\tbox{2.2cm}{constraint reduction \\(w.r.t. $\Go$ at $0$)}}
&&(\uCo,\uomegao)
    \ar@{~>}[rr]^-{\tbox{2.2cm}{flux superselection (w.r.t. $\Gred$ at $\mathcal{O}_f$)}}
&&(\uuS_{[f]},\uuomegao_{[f]})\\
&{
    \;\qquad\C\qquad\;
    \ar@{_(->}[ul]^-{\iota_\C}
    \ar@{->>}[ur]_-{\pi_\circ}
    }
&&{
    \;\;\uh^{-1}(\mathcal{O}_{f})
    \ar@{_(->}[ul]^-{\underline{\iota}_{[f]}}
    \ar@{->>}[ur]_-{\underline{\pi}_{[f]}}
    }\\
&&{
    \;\;\h^{-1}(\mathcal{O}_f)\cap\C
    \ar@{_(->}[ul]^-{{\iota}^\C_{[f]}}
    \ar@{->>}[ur]_-{\pi_{\circ,[f]}}
    \ar@{_(-->}@/^2.7pc/[uull]^{\iota_{[f]}}
    \ar@{-->>}@/_2.7pc/[uurr]_{\pi_{[f]}}
    }\\
}
\end{equation*}

\section*{Declarations}
\subsubsection*{Conflict of interest statement} The authors have no competing interests to declare that are relevant to the content of this article.
\subsubsection*{Data availability statement} Data sharing not applicable to this article as no datasets were generated or analysed during the current study.
\medskip 

\begingroup
\sloppy
\printbibliography
\endgroup

\end{document}